\documentclass[10pt,a4paper,english]{article}\usepackage[]{graphicx}\usepackage[]{color}
\makeatletter
\def\maxwidth{ %
  \ifdim\Gin@nat@width>\linewidth
    \linewidth
  \else
    \Gin@nat@width
  \fi
}
\makeatother

\definecolor{fgcolor}{rgb}{0.345, 0.345, 0.345}
\newcommand{\hlnum}[1]{\textcolor[rgb]{0.686,0.059,0.569}{#1}}%
\newcommand{\hlstr}[1]{\textcolor[rgb]{0.192,0.494,0.8}{#1}}%
\newcommand{\hlcom}[1]{\textcolor[rgb]{0.678,0.584,0.686}{\textit{#1}}}%
\newcommand{\hlopt}[1]{\textcolor[rgb]{0,0,0}{#1}}%
\newcommand{\hlstd}[1]{\textcolor[rgb]{0.345,0.345,0.345}{#1}}%
\newcommand{\hlkwa}[1]{\textcolor[rgb]{0.161,0.373,0.58}{\textbf{#1}}}%
\newcommand{\hlkwb}[1]{\textcolor[rgb]{0.69,0.353,0.396}{#1}}%
\newcommand{\hlkwc}[1]{\textcolor[rgb]{0.333,0.667,0.333}{#1}}%
\newcommand{\hlkwd}[1]{\textcolor[rgb]{0.737,0.353,0.396}{\textbf{#1}}}%

\usepackage{framed}
\makeatletter
\newenvironment{kframe}{%
 \def\at@end@of@kframe{}%
 \ifinner\ifhmode%
  \def\at@end@of@kframe{\end{minipage}}%
  \begin{minipage}{\columnwidth}%
 \fi\fi%
 \def\FrameCommand##1{\hskip\@totalleftmargin \hskip-\fboxsep
 \colorbox{shadecolor}{##1}\hskip-\fboxsep
     \hskip-\linewidth \hskip-\@totalleftmargin \hskip\columnwidth}%
 \MakeFramed {\advance\hsize-\width
   \@totalleftmargin\z@ \linewidth\hsize
   \@setminipage}}%
 {\par\unskip\endMakeFramed%
 \at@end@of@kframe}
\makeatother

\definecolor{shadecolor}{rgb}{.97, .97, .97}
\definecolor{messagecolor}{rgb}{0, 0, 0}
\definecolor{warningcolor}{rgb}{1, 0, 1}
\definecolor{errorcolor}{rgb}{1, 0, 0}
\newenvironment{knitrout}{}{} 

\usepackage{alltt}

\usepackage[hyphens]{url}
\usepackage{amsmath}
\usepackage{amsfonts}
\usepackage{amssymb}
\usepackage{amsthm}

\theoremstyle{plain}
\newtheorem{theorem}{Theorem}[section]
\newtheorem{lemma}[theorem]{Lemma}
\newtheorem{corollary}[theorem]{Corollary}

\theoremstyle{definition}
\newtheorem{definition}[theorem]{Definition}

\theoremstyle{remark}

\newtheorem*{caution}{Caution}

\PassOptionsToPackage{hyphens}{url}\usepackage{hyperref}
\usepackage{hyperref}
\usepackage[square,numbers]{natbib}

\usepackage[newitem,newenum,increaseonly]{paralist}


\usepackage[notheorems]{SharpeR}

    \usepackage{color} 
    \usepackage{enumerate} 
    \usepackage{amsmath} 
    \usepackage{amssymb} 
		\usepackage{fancyvrb} 
    \usepackage{hyperref}

		\definecolor{white}{rgb}{1.0,1.0,1.0}
    \definecolor{orange}{cmyk}{0,0.4,0.8,0.2}
    \definecolor{darkorange}{rgb}{.71,0.21,0.01}
    \definecolor{darkgreen}{rgb}{.12,.54,.11}
    \definecolor{myteal}{rgb}{.26, .44, .56}
    \definecolor{gray}{gray}{0.45}
    \definecolor{lightgray}{gray}{.95}
    \definecolor{mediumgray}{gray}{.8}
    \definecolor{inputbackground}{rgb}{.95, .95, .85}
    \definecolor{outputbackground}{rgb}{.95, .95, .95}
    \definecolor{traceback}{rgb}{1, .95, .95}
    \definecolor{red}{rgb}{.6,0,0}
    \definecolor{green}{rgb}{0,.65,0}
    \definecolor{brown}{rgb}{0.6,0.6,0}
    \definecolor{blue}{rgb}{0,.145,.698}
    \definecolor{purple}{rgb}{.698,.145,.698}
    \definecolor{cyan}{rgb}{0,.698,.698}
    \definecolor{lightgray}{gray}{0.5}
    
    \definecolor{darkgray}{gray}{0.25}
    \definecolor{lightred}{rgb}{1.0,0.39,0.28}
    \definecolor{lightgreen}{rgb}{0.48,0.99,0.0}
    \definecolor{lightblue}{rgb}{0.53,0.81,0.92}
    \definecolor{lightpurple}{rgb}{0.87,0.63,0.87}
    \definecolor{lightcyan}{rgb}{0.5,1.0,0.83}
    
    \DefineShortVerb[commandchars=\\\{\}]{\|}
    \DefineVerbatimEnvironment{Highlighting}{Verbatim}{commandchars=\\\{\}}

\makeatletter
\def\PY@reset{\let\PY@it=\relax \let\PY@bf=\relax%
    \let\PY@ul=\relax \let\PY@tc=\relax%
    \let\PY@bc=\relax \let\PY@ff=\relax}
\def\PY@tok#1{\csname PY@tok@#1\endcsname}
\def\PY@toks#1+{\ifx\relax#1\empty\else%
    \PY@tok{#1}\expandafter\PY@toks\fi}
\def\PY@do#1{\PY@bc{\PY@tc{\PY@ul{%
    \PY@it{\PY@bf{\PY@ff{#1}}}}}}}
\def\PY#1#2{\PY@reset\PY@toks#1+\relax+\PY@do{#2}}

\expandafter\def\csname PY@tok@gd\endcsname{\def\PY@tc##1{\textcolor[rgb]{0.63,0.00,0.00}{##1}}}
\expandafter\def\csname PY@tok@gu\endcsname{\let\PY@bf=\textbf\def\PY@tc##1{\textcolor[rgb]{0.50,0.00,0.50}{##1}}}
\expandafter\def\csname PY@tok@gt\endcsname{\def\PY@tc##1{\textcolor[rgb]{0.00,0.27,0.87}{##1}}}
\expandafter\def\csname PY@tok@gs\endcsname{\let\PY@bf=\textbf}
\expandafter\def\csname PY@tok@gr\endcsname{\def\PY@tc##1{\textcolor[rgb]{1.00,0.00,0.00}{##1}}}
\expandafter\def\csname PY@tok@cm\endcsname{\let\PY@it=\textit\def\PY@tc##1{\textcolor[rgb]{0.25,0.50,0.50}{##1}}}
\expandafter\def\csname PY@tok@vg\endcsname{\def\PY@tc##1{\textcolor[rgb]{0.10,0.09,0.49}{##1}}}
\expandafter\def\csname PY@tok@m\endcsname{\def\PY@tc##1{\textcolor[rgb]{0.40,0.40,0.40}{##1}}}
\expandafter\def\csname PY@tok@mh\endcsname{\def\PY@tc##1{\textcolor[rgb]{0.40,0.40,0.40}{##1}}}
\expandafter\def\csname PY@tok@go\endcsname{\def\PY@tc##1{\textcolor[rgb]{0.53,0.53,0.53}{##1}}}
\expandafter\def\csname PY@tok@ge\endcsname{\let\PY@it=\textit}
\expandafter\def\csname PY@tok@vc\endcsname{\def\PY@tc##1{\textcolor[rgb]{0.10,0.09,0.49}{##1}}}
\expandafter\def\csname PY@tok@il\endcsname{\def\PY@tc##1{\textcolor[rgb]{0.40,0.40,0.40}{##1}}}
\expandafter\def\csname PY@tok@cs\endcsname{\let\PY@it=\textit\def\PY@tc##1{\textcolor[rgb]{0.25,0.50,0.50}{##1}}}
\expandafter\def\csname PY@tok@cp\endcsname{\def\PY@tc##1{\textcolor[rgb]{0.74,0.48,0.00}{##1}}}
\expandafter\def\csname PY@tok@gi\endcsname{\def\PY@tc##1{\textcolor[rgb]{0.00,0.63,0.00}{##1}}}
\expandafter\def\csname PY@tok@gh\endcsname{\let\PY@bf=\textbf\def\PY@tc##1{\textcolor[rgb]{0.00,0.00,0.50}{##1}}}
\expandafter\def\csname PY@tok@ni\endcsname{\let\PY@bf=\textbf\def\PY@tc##1{\textcolor[rgb]{0.60,0.60,0.60}{##1}}}
\expandafter\def\csname PY@tok@nl\endcsname{\def\PY@tc##1{\textcolor[rgb]{0.63,0.63,0.00}{##1}}}
\expandafter\def\csname PY@tok@nn\endcsname{\let\PY@bf=\textbf\def\PY@tc##1{\textcolor[rgb]{0.00,0.00,1.00}{##1}}}
\expandafter\def\csname PY@tok@no\endcsname{\def\PY@tc##1{\textcolor[rgb]{0.53,0.00,0.00}{##1}}}
\expandafter\def\csname PY@tok@na\endcsname{\def\PY@tc##1{\textcolor[rgb]{0.49,0.56,0.16}{##1}}}
\expandafter\def\csname PY@tok@nb\endcsname{\def\PY@tc##1{\textcolor[rgb]{0.00,0.50,0.00}{##1}}}
\expandafter\def\csname PY@tok@nc\endcsname{\let\PY@bf=\textbf\def\PY@tc##1{\textcolor[rgb]{0.00,0.00,1.00}{##1}}}
\expandafter\def\csname PY@tok@nd\endcsname{\def\PY@tc##1{\textcolor[rgb]{0.67,0.13,1.00}{##1}}}
\expandafter\def\csname PY@tok@ne\endcsname{\let\PY@bf=\textbf\def\PY@tc##1{\textcolor[rgb]{0.82,0.25,0.23}{##1}}}
\expandafter\def\csname PY@tok@nf\endcsname{\def\PY@tc##1{\textcolor[rgb]{0.00,0.00,1.00}{##1}}}
\expandafter\def\csname PY@tok@si\endcsname{\let\PY@bf=\textbf\def\PY@tc##1{\textcolor[rgb]{0.73,0.40,0.53}{##1}}}
\expandafter\def\csname PY@tok@s2\endcsname{\def\PY@tc##1{\textcolor[rgb]{0.73,0.13,0.13}{##1}}}
\expandafter\def\csname PY@tok@vi\endcsname{\def\PY@tc##1{\textcolor[rgb]{0.10,0.09,0.49}{##1}}}
\expandafter\def\csname PY@tok@nt\endcsname{\let\PY@bf=\textbf\def\PY@tc##1{\textcolor[rgb]{0.00,0.50,0.00}{##1}}}
\expandafter\def\csname PY@tok@nv\endcsname{\def\PY@tc##1{\textcolor[rgb]{0.10,0.09,0.49}{##1}}}
\expandafter\def\csname PY@tok@s1\endcsname{\def\PY@tc##1{\textcolor[rgb]{0.73,0.13,0.13}{##1}}}
\expandafter\def\csname PY@tok@sh\endcsname{\def\PY@tc##1{\textcolor[rgb]{0.73,0.13,0.13}{##1}}}
\expandafter\def\csname PY@tok@sc\endcsname{\def\PY@tc##1{\textcolor[rgb]{0.73,0.13,0.13}{##1}}}
\expandafter\def\csname PY@tok@sx\endcsname{\def\PY@tc##1{\textcolor[rgb]{0.00,0.50,0.00}{##1}}}
\expandafter\def\csname PY@tok@bp\endcsname{\def\PY@tc##1{\textcolor[rgb]{0.00,0.50,0.00}{##1}}}
\expandafter\def\csname PY@tok@c1\endcsname{\let\PY@it=\textit\def\PY@tc##1{\textcolor[rgb]{0.25,0.50,0.50}{##1}}}
\expandafter\def\csname PY@tok@kc\endcsname{\let\PY@bf=\textbf\def\PY@tc##1{\textcolor[rgb]{0.00,0.50,0.00}{##1}}}
\expandafter\def\csname PY@tok@c\endcsname{\let\PY@it=\textit\def\PY@tc##1{\textcolor[rgb]{0.25,0.50,0.50}{##1}}}
\expandafter\def\csname PY@tok@mf\endcsname{\def\PY@tc##1{\textcolor[rgb]{0.40,0.40,0.40}{##1}}}
\expandafter\def\csname PY@tok@err\endcsname{\def\PY@bc##1{\setlength{\fboxsep}{0pt}\fcolorbox[rgb]{1.00,0.00,0.00}{1,1,1}{\strut ##1}}}
\expandafter\def\csname PY@tok@kd\endcsname{\let\PY@bf=\textbf\def\PY@tc##1{\textcolor[rgb]{0.00,0.50,0.00}{##1}}}
\expandafter\def\csname PY@tok@ss\endcsname{\def\PY@tc##1{\textcolor[rgb]{0.10,0.09,0.49}{##1}}}
\expandafter\def\csname PY@tok@sr\endcsname{\def\PY@tc##1{\textcolor[rgb]{0.73,0.40,0.53}{##1}}}
\expandafter\def\csname PY@tok@mo\endcsname{\def\PY@tc##1{\textcolor[rgb]{0.40,0.40,0.40}{##1}}}
\expandafter\def\csname PY@tok@kn\endcsname{\let\PY@bf=\textbf\def\PY@tc##1{\textcolor[rgb]{0.00,0.50,0.00}{##1}}}
\expandafter\def\csname PY@tok@mi\endcsname{\def\PY@tc##1{\textcolor[rgb]{0.40,0.40,0.40}{##1}}}
\expandafter\def\csname PY@tok@gp\endcsname{\let\PY@bf=\textbf\def\PY@tc##1{\textcolor[rgb]{0.00,0.00,0.50}{##1}}}
\expandafter\def\csname PY@tok@o\endcsname{\def\PY@tc##1{\textcolor[rgb]{0.40,0.40,0.40}{##1}}}
\expandafter\def\csname PY@tok@kr\endcsname{\let\PY@bf=\textbf\def\PY@tc##1{\textcolor[rgb]{0.00,0.50,0.00}{##1}}}
\expandafter\def\csname PY@tok@s\endcsname{\def\PY@tc##1{\textcolor[rgb]{0.73,0.13,0.13}{##1}}}
\expandafter\def\csname PY@tok@kp\endcsname{\def\PY@tc##1{\textcolor[rgb]{0.00,0.50,0.00}{##1}}}
\expandafter\def\csname PY@tok@w\endcsname{\def\PY@tc##1{\textcolor[rgb]{0.73,0.73,0.73}{##1}}}
\expandafter\def\csname PY@tok@kt\endcsname{\def\PY@tc##1{\textcolor[rgb]{0.69,0.00,0.25}{##1}}}
\expandafter\def\csname PY@tok@ow\endcsname{\let\PY@bf=\textbf\def\PY@tc##1{\textcolor[rgb]{0.67,0.13,1.00}{##1}}}
\expandafter\def\csname PY@tok@sb\endcsname{\def\PY@tc##1{\textcolor[rgb]{0.73,0.13,0.13}{##1}}}
\expandafter\def\csname PY@tok@k\endcsname{\let\PY@bf=\textbf\def\PY@tc##1{\textcolor[rgb]{0.00,0.50,0.00}{##1}}}
\expandafter\def\csname PY@tok@se\endcsname{\let\PY@bf=\textbf\def\PY@tc##1{\textcolor[rgb]{0.73,0.40,0.13}{##1}}}
\expandafter\def\csname PY@tok@sd\endcsname{\let\PY@it=\textit\def\PY@tc##1{\textcolor[rgb]{0.73,0.13,0.13}{##1}}}

\makeatother

    \definecolor{incolor}{rgb}{0.0, 0.0, 0.5}
    \definecolor{outcolor}{rgb}{0.545, 0.0, 0.0}

    \hypersetup{
      breaklinks=true,  
      colorlinks=true,
      urlcolor=blue,
      linkcolor=darkorange,
      citecolor=darkgreen,
      }
    
\IfFileExists{upquote.sty}{\usepackage{upquote}}{}
\begin{document}

\title{Asymptotic Distribution of the \txtMwtz Portfolio}
\author{Steven E. Pav \thanks{\email{steven@gilgamath.com}}}

\maketitle

\begin{abstract}
The asymptotic distribution of the \txtMP, \minvAB{\svsig}{\svmu},
is derived, for the general case (assuming fourth moments of returns exist), 
and for the case of multivariate normal returns. 
The derivation allows for inference which is robust to heteroskedasticity 
and autocorrelation of moments up to order four. As a side effect, one 
can estimate the proportion of error in the \txtMP due to 
mis-estimation of the covariance matrix. A likelihood ratio test is given
which generalizes Dempster's Covariance Selection test to allow inference
on linear combinations of the precision matrix and the 
\txtMP. \cite{dempster1972} Extensions of the main method to deal with
hedged portfolios, conditional heteroskedasticity, conditional 
expectation, and constrained estimation are given.
It is shown that the Hotelling-Lawley
statistic generalizes the (squared) \txtSR under the conditional
expectation model.
Asymptotic distributions of all four of the common `MGLH' statistics are found,
assuming random covariates. \cite{shieh2005}  
Examples are given demonstrating the possible uses of these results.
\end{abstract}

\section{Introduction}

Given \nlatf assets with expected return \pvmu and covariance of return \pvsig,
the portfolio defined as 
\begin{equation}
\pportwopt \defeq \lambda \minvAB{\pvsig}{\pvmu}
\end{equation}
plays a special role in modern portfolio 
theory. \cite{markowitz1952portfolio,brandt2009portfolio,GVK322224764}
It is known as the `efficient portfolio', the `tangency portfolio', 
and, somewhat informally, the `\txtMP'. 
It appears, for various $\lambda$, in the solution to numerous
portfolio optimization problems.  
Besides the classic mean-variance formulation,
it solves the (population) \txtSR maximization problem:
\begin{equation}
\max_{\pportw : \qform{\pvsig}{\pportw} \le \Rbuj^2} 
\frac{\trAB{\pportw}{\pvmu} - \rfr}{\sqrt{\qform{\pvsig}{\pportw}}},
\label{eqn:opt_port_I}
\end{equation}
where $\rfr\ge 0$ is the risk-free, or `disastrous', rate of return, and 
$\Rbuj > 0$ is some given `risk budget'. 
The solution to this optimization problem is $\lambda \minvAB{\pvsig}{\pvmu}$,
where $\lambda = \fracc{\Rbuj}{\sqrt{\qiform{\pvsig}{\pvmu}}}.$
\nocite{markowitz1952portfolio}  

In practice, the \txtMP has a somewhat checkered history. 
The population parameters \pvmu and \pvsig are not known and must
be estimated from samples. Estimation error results in a feasible
portfolio, \sportwopt, of dubious value. Michaud went so far as to call
mean-variance optimization, 
``error maximization.'' \cite{michaud1989markowitz}  
It has been suggested that simple portfolio heuristics outperform the
\txtMP in practice.  \cite{demiguel2009optimal}  

This paper focuses on the asymptotic distribution of the sample \txtMP. 
By formulating the problem as a linear regression, Britten-Jones 
very cleverly devised hypothesis tests on elements of \pportwopt,
assuming multivariate Gaussian returns. \cite{BrittenJones1999}  
In a remarkable series of papers, Okhrin and Schmid, and 
Bodnar and Okhrin give the (univariate) density of
the dot product of \pportwopt and a deterministic vector, again for the
case of Gaussian returns. \cite{okhrin2006distributional,SJOS:SJOS729}  
Okhrin and Schmid also show that all moments of
$\fracc{\sportwopt}{\trAB{\vone}{\sportwopt}}$ of order greater than or
equal to one do not exist. \cite{okhrin2006distributional}

Here I derive asymptotic normality of \sportwopt, the sample
analogue of \pportwopt, assuming only that the first four moments
exist. Feasible estimation of the variance of \sportwopt is amenable
to heteroskedasticity and autocorrelation robust 
inference. \cite{Zeileis:2004:JSSOBK:v11i10}
The asymptotic distribution under Gaussian returns is also derived.

After estimating the covariance of \sportwopt, one can compute Wald
test statistics for the elements of \sportwopt, possibly leading one
to drop some assets from consideration (`sparsification'). Having
an estimate of the covariance can also allow portfolio shrinkage.
\cite{demiguel2013size,kinkawa2010estimation}

The derivations in this paper actually solve a more general problem
than the distribution of the sample \txtMP. The covariance of 
\sportwopt and the `precision matrix,' \minv{\svsig} are derived.
This allows one, for example, to estimate the proportion of error
in the \txtMP attributable to mis-estimation of the covariance
matrix. According to lore, the error in portfolio weights is
mostly attributable to mis-estimation of \pvmu, 
not of \pvsig. \cite{chopra1993effect,NBERw0444}

Finally, assuming Gaussian returns, a likelihood ratio test for
performing inference on linear combinations of elements of
the \txtMP and the precision matrix is derived. This test
generalizes a procedure by Dempster for inference on
the precision matrix alone. \cite{dempster1972}


\section{The augmented second moment}

Let \vreti be an array of returns of \nlatf assets, with mean \pvmu, and
covariance \pvsig.  Let \avreti be \vreti prepended with a 1:
$\avreti = \asvec{1,\tr{\vreti}}$. Consider the second
moment of \avreti:
\begin{equation}
\pvsm \defeq \E{\ogram{\avreti}} = 
	\twobytwo{1}{\tr{\pvmu}}{\pvmu}{\pvsig + \ogram{\pvmu}}.
\label{eqn:pvsm_def}
\end{equation}
By inspection one can confirm that the
inverse of \pvsm is
\begin{equation}
\minv{\pvsm} 
= \twobytwo{1 + \qiform{\pvsig}{\pvmu}}{-\tr{\pvmu}\minv{\pvsig}}{-\minv{\pvsig}\pvmu}{\minv{\pvsig}}
= \twobytwo{1 + \psnrsqopt}{-\tr{\pportwopt}}{-\pportwopt}{\minv{\pvsig}},
\label{eqn:trick_inversion}
\end{equation}
where $\pportwopt=\minvAB{\pvsig}{\svmu}$ is the \txtMP,
and $\psnropt=\sqrt{\qiform{\pvsig}{\pvmu}}$ is the \txtSR of
that portfolio. The matrix \pvsm contains the first and second
moment of \vreti, but is also the uncentered second moment of
\avreti, a fact which makes it amenable to analysis via the
central limit theorem.

The relationships above are merely facts of linear algebra, and so
hold for sample estimates as well:
\begin{equation*}
\minv{\twobytwo{1}{\tr{\svmu}}{\svmu}{\svsig + \ogram{\svmu}}}
= {\twobytwo{1 +
\ssrsqopt}{-\tr{\sportwopt}}{-\sportwopt}{\minv{\svsig}}},
\end{equation*}
where \svmu, \svsig are some sample estimates of \pvmu and \pvsig, and 
$\sportwopt = \minvAB{\svsig}{\svmu}, \ssrsqopt = \qiform{\svsig}{\svmu}$.

Given \ssiz \iid observations \vreti[i], let \amreti be the matrix
whose rows are the vectors \tr{\avreti[i]}. The na\"{i}ve sample estimator
\begin{equation}
\svsm \defeq \oneby{\ssiz}\gram{\amreti}
\end{equation}
is an unbiased estimator since $\pvsm = \E{\gram{\avreti}}$.

\subsection{Matrix Derivatives}

\label{subsec:matrix_derivatives}

Some notation and technical results concerning matrices are required.
\begin{definition}[Matrix operations]
For matrix \Mtx{A},
let \fvec{\Mtx{A}}, and \fvech{\Mtx{A}} be the vector and half-space
vector operators.  The former turns an \bby{\nlatf}{\nlatf} matrix into
an $\nlatf^2$ vector of its columns stacked on top of each other; the 
latter vectorizes a symmetric (or lower triangular) matrix into a vector
of the non-redundant elements.  
Let \Elim be the `Elimination Matrix,' a matrix of zeros and ones with the
property that 
$\fvech{\Mtx{A}} = \Elim\fvec{\Mtx{A}}.$ The `Duplication Matrix,' \Dupp,
is the matrix of zeros and ones that reverses this operation: 
$\Dupp \fvech{\Mtx{A}} = \fvec{\Mtx{A}}.$ \cite{magnus1980elimination} 
Note that this implies that 
$$\Elim\Dupp = \eye \wrapParens{\ne \Dupp\Elim}.$$

We will let \Komm be the 'commutation matrix', \ie the matrix whose rows are
a permutation of the rows of the identity matrix such that 
$\Komm\fvec{\Mtx{A}} = \fvec{\tr{\Mtx{A}}}$ for square matrix \Mtx{A}.


\end{definition}
\begin{definition}[Derivatives]
For $m$-vector \vect{x}, and $n$-vector \vect{y}, let the derivative
\dbyd{\vect{y}}{\vect{x}} be the \bby{n}{m} matrix whose first column
is the partial derivative of \vect{y} with respect to $x_1$.  
This follows the so-called `numerator layout' convention. 
For matrices \Mtx{Y} and \Mtx{X}, define
\begin{equation*}
\dbyd{\Mtx{Y}}{\Mtx{X}} \defeq \dbyd{\fvec{\Mtx{Y}}}{\fvec{\Mtx{X}}}.
\end{equation*}
\end{definition}
\begin{lemma}[Miscellaneous Derivatives]
\label{lemma:misc_derivs}
For symmetric matrices \Mtx{Y} and \Mtx{X}, 
\begin{equation}
\dbyd{\fvech{\Mtx{Y}}}{\fvec{\Mtx{X}}} = \Elim \dbyd{\Mtx{Y}}{\Mtx{X}},\quad
\dbyd{\fvec{\Mtx{Y}}}{\fvech{\Mtx{X}}} = \dbyd{\Mtx{Y}}{\Mtx{X}}\Dupp,\quad
\dbyd{\fvech{\Mtx{Y}}}{\fvech{\Mtx{X}}} = \EXD{\dbyd{\Mtx{Y}}{\Mtx{X}}}.
\end{equation}
\end{lemma}
\begin{proof}
For the first equation, note that 
$\fvech{\Mtx{Y}} = \Elim\fvec{\Mtx{Y}}$, thus by the chain rule:
$$
\dbyd{\fvech{\Mtx{Y}}}{\fvec{\Mtx{X}}} = 
\dbyd{\Elim \fvec{\Mtx{Y}}}{\fvec{\Mtx{Y}}} = \Elim \dbyd{\Mtx{Y}}{\Mtx{X}},
$$
by linearity of the derivative. The other identities follow similarly.
\end{proof}
\begin{lemma}[Derivative of matrix inverse]
For invertible matrix \Mtx{A}, 
\begin{equation}
\dbyd{\minv{\Mtx{A}}}{\Mtx{A}} 
= - \wrapParens{\trminv{\Mtx{A}}\kron\minv{\Mtx{A}}}
= - \minvParens{\tr{\Mtx{A}}\kron\Mtx{A}}.
\label{eqn:deriv_vec_matrix_inverse}
\end{equation}
For \emph{symmetric} \Mtx{A}, the derivative with respect to the
non-redundant part is 
\begin{equation}
\dbyd{\fvech{\minv{\Mtx{A}}}}{\fvech{\Mtx{A}}} 
= - \EXD{\wrapParens{\minv{\Mtx{A}}\kron\minv{\Mtx{A}}}}.
\label{eqn:deriv_vech_matrix_inverse}
\end{equation}

\label{lemma:deriv_vech_matrix_inverse}
\end{lemma}
Note how this result generalizes the scalar derivative:
$\dbyd{x^{-1}}{x} = - \wrapParens{x^{-1} x^{-1}}.$
\begin{proof}
\eqnref{deriv_vec_matrix_inverse} is a known 
result. \cite{facklernotes,magnus1999matrix}
\eqnref{deriv_vech_matrix_inverse} then follows using \lemmaref{misc_derivs}.
\end{proof}

\subsection{Asymptotic distribution of the \txtMP}

\label{subsec:dist_markoport}
\nocite{BrittenJones1999}

Collecting the mean and covariance into the second moment matrix 
gives the asymptotic distribution of the sample \txtMP
without much work. In some sense, this computation
generalizes the `standard' asymptotic analysis of Sharpe ratio of
multiple assets. \cite{jobsonkorkie1981,lo2002,Ledoit2008850,Leung2008} 

\begin{theorem}
\label{theorem:inv_distribution}
Let \svsm be the unbiased sample estimate of 
\pvsm, based on \ssiz \iid samples of \vreti.
Let \pvvar be the variance of $\fvech{\ogram{\avreti}}$.
Then, asymptotically in \ssiz, 
\begin{equation}
\sqrt{\ssiz}\wrapParens{\fvech{\minv{\svsm}} - \fvech{\minv{\pvsm}}} 
\rightsquigarrow \normlaw{0,\qoform{\pvvar}{\Mtx{H}}},
\label{eqn:mvclt_isvsm}
\end{equation}
where
\begin{equation}
\Mtx{H} = -\EXD{\wrapParens{\AkronA{\minv{\pvsm}}}}.
\end{equation}
Furthermore, we may replace \pvvar in this equation with an asymptotically
consistent estimator, \svvar.
\end{theorem}
\begin{proof}
Under the multivariate central limit theorem \cite{wasserman2004all}
\begin{equation}
\sqrt{\ssiz}\wrapParens{\fvech{\svsm} - \fvech{\pvsm}} 
\rightsquigarrow 
\normlaw{0,\pvvar},
\label{eqn:mvclt_svsm}
\end{equation}
where \pvvar is the variance of $\fvech{\ogram{\avreti}}$, which, in general, 
is unknown. 
By the delta method \cite{wasserman2004all},
\begin{equation*}
\sqrt{\ssiz}\wrapParens{\fvech{\minv{\svsm}} - \fvech{\minv{\pvsm}}} 
\rightsquigarrow 
\normlaw{0,\qoform{\pvvar}{\wrapBracks{\dbyd{\fvech{\minv{\pvsm}}}{\fvech{\pvsm}}}}}.
\end{equation*}
The derivative is given by \lemmaref{deriv_vech_matrix_inverse}, and
the result follows.
\end{proof}

To estimate the covariance of $\fvech{\minv{\svsm}}$,
plug in \svsm for \pvsm in the covariance computation, and use some
consistent estimator for \pvvar, call 
it \svvar. 
One way to compute \svvar is to via the sample covariance of the
vectors $\fvech{\ogram{\avreti[i]}} =
\asvec{1,\tr{\vreti[i]},\tr{\fvech{\ogram{\vreti[i]}}}}$. 
More elaborate covariance estimators can be used, for example, to deal with
violations of the \iid assumptions. \cite{Zeileis:2004:JSSOBK:v11i10}
\nocite{magnus1999matrix,magnus1980elimination}
\nocite{BrittenJones1999}
Note that because the first element of 
$\fvech{\ogram{\avreti[i]}}$ is a deterministic $1$, the first row and
column of \pvvar is all zeros, and we need not estimate it.

\subsection{The \txtSR optimal portfolio}

\begin{lemma}[\txtSR optimal portfolio]
\label{lemma:sr_optimal_portfolio}
Assuming $\pvmu \ne \vzero$, and \pvsig is invertible,
the portfolio optimization problem
\begin{equation}
\argmax_{\pportw :\, \qform{\pvsig}{\pportw} \le \Rbuj^2} 
\frac{\trAB{\pportw}{\pvmu} - \rfr}{\sqrt{\qform{\pvsig}{\pportw}}},
\label{eqn:sr_optimal_portfolio_problem}
\end{equation}
for $\rfr \ge 0, \Rbuj > 0$ is solved by
\begin{equation}
\pportwoptR \defeq \frac{\Rbuj}{\sqrt{\qiform{\pvsig}{\pvmu}}}
\minvAB{\pvsig}{\pvmu}.
\label{eqn:pportwoptR_def}
\end{equation}
Moreover, this is the unique solution whenever $\rfr > 0$.
The maximal objective achieved by this portfolio is
$\sqrt{\qiform{\pvsig}{\pvmu}} - \fracc{\rfr}{\Rbuj} = 
\psnropt - \fracc{\rfr}{\Rbuj}$.
\end{lemma}
\begin{proof}
By the Lagrange multiplier technique, the optimal portfolio
solves the following equations:
\begin{equation*}
\begin{split}
0 &= c_1 \pvmu - c_2 \pvsig \pportw - \gamma \pvsig \pportw,\\
\qform{\pvsig}{\pportw} &\le \Rbuj^2,
\end{split}
\end{equation*}
where $\gamma$ is the Lagrange multiplier, and $c_1, c_2$ are 
scalar constants.
Solving the first equation gives us
$$
\pportw = c\,\minvAB{\pvsig}{\pvmu}.
$$
This reduces the problem to the univariate optimization
\begin{equation}
\max_{c :\, c^2 \le \fracc{\Rbuj^2}{\psnrsqopt}} 
\sign{c} \psnropt - \frac{\rfr}{\abs{c}\psnropt},
\end{equation}
where $\psnrsqopt = \qiform{\pvsig}{\pvmu}.$ The optimum
occurs for $c = \fracc{\Rbuj}{\psnropt}$, moreover the optimum
is unique when $\rfr > 0$.
\end{proof}

Note that the first element of \fvech{\minv{\pvsm}} is $1 +
\qiform{\pvsig}{\pvmu}$, and elements 2 through $\nlatf+1$ are 
$-\pportwopt$. Thus, \pportwoptR, the portfolio that maximizes the \txtSR, 
is some transformation of \fvech{\minv{\pvsm}}, and another 
application of the delta method gives its asymptotic distribution,
as in the following corollary to \theoremref{inv_distribution}.

\begin{corollary}
\label{corollary:portwoptR_dist}
Let 
\begin{equation}
\pportwoptR = \frac{\Rbuj}{\sqrt{\qiform{\pvsig}{\pvmu}}}
\minvAB{\pvsig}{\pvmu},
\end{equation}
and similarly, let \sportwoptR be the sample analogue, where \Rbuj is
some risk budget. Then
\begin{equation}
\sqrt{\ssiz}\wrapParens{\sportwoptR - \pportwoptR} 
\rightsquigarrow 
\normlaw{0,\qoform{\pvvar}{\Mtx{H}}},
\label{eqn:mvclt_portfolio}
\end{equation}
where
\begin{equation}
\begin{split}
\Mtx{H} &= \wrapParens{- \asrowvec{\oneby{2\psnrsqopt} \pportwoptR,
\frac{\Rbuj}{\psnropt}\eye[\nlatf],\mzero}} 
\wrapParens{-\EXD{\wrapParens{\AkronA{\minv{\pvsm}}}}},\\
\psnrsqopt &\defeq \qiform{\pvsig}{\pvmu}.
\end{split}
\end{equation}

Moreover, we may express \Mtx{H} as
\begin{equation}
\label{eqn:shorter_H_form}
	- \wrapParens{\wrapParens{\trAB{\basev[1]}{\minv{\pvsm}}} \kron
	\asrowvec{\frac{1 - \psnrsqopt}{2\psnropt}\pportwoptR,
	\frac{\Rbuj}{\psnropt}\wrapParens{\minv{\pvsig} -
	\frac{\minv{\pvsig}\ogram{\pvmu}\minv{\pvsig}}{2\psnrsqopt}}}}\Dupp.
\end{equation}
\end{corollary}
\begin{proof}
By the delta method, and \theoremref{inv_distribution}, it suffices
to show that 
$$\dbyd{\pportwoptR}{\fvech{\minv{\pvsm}}} 
= - \asrowvec{\oneby{2\psnrsqopt} \pportwoptR,
\frac{\Rbuj}{\psnropt}\eye[\nlatf],\mzero}.
$$
To show this, note that \pportwoptR is $-\Rbuj$ times elements
2 through $\nlatf+1$ of \fvech{\minv{\pvsm}} divided by 
$\psnropt = \sqrt{\trAB{\basev[1]}{\fvech{\minv{\pvsm}}} - 1}$, where
$\basev[i]$ is the \kth{i} column of the identity matrix. The result
follows from basic calculus.

To establish \eqnref{shorter_H_form}, note that only the first $\nlatf+1$
columns of $\dbyd{\pportwoptR}{\fvech{\minv{\pvsm}}}$ have non-zero entries,
thus the elimination matrix, \Elim, can be ignored in the term on the
right, and we could write the derivative as
\begin{equation*}
	\dbyd{\pportwoptR}{\fvech{\minv{\pvsm}}} = \tr{\basev[1]} \kron -
	\asrowvec{\oneby{2\psnrsqopt} \pportwoptR,
	\frac{\Rbuj}{\psnropt}\eye[\nlatf]}.
\end{equation*}
And we can write the product as
\begin{equation*}
\begin{split}
	\Mtx{H} &= 
		- \wrapParens{\wrapParens{\tr{\basev[1]}\minv{\pvsm}}\kron
	\wrapParens{\asrowvec{\oneby{2\psnrsqopt} \pportwoptR,
	\frac{\Rbuj}{\psnropt}\eye[\nlatf]}\minv{\pvsm}}}\Dupp.
\end{split}
\end{equation*}
Perform the matrix multiplication to find
\begin{equation*}
\begin{split}
\asrowvec{\oneby{2\psnrsqopt} \pportwoptR, \frac{\Rbuj}{\psnropt}\eye[\nlatf]} \minv{\pvsm} &=
	\asrowvec{\frac{1 + \psnrsqopt}{2\psnrsqopt} \pportwoptR - \frac{\Rbuj}{\psnropt} \minv{\pvsig}\pvmu, 
	-\oneby{2\psnrsqopt} \pportwoptR \tr{\pvmu}\minv{\pvsig} +
	\frac{\Rbuj}{\psnropt}\minv{\pvsig}},\\
	&= 
	\asrowvec{\frac{1 - \psnrsqopt}{2\psnrsqopt} \pportwoptR ,
	-\oneby{2\psnrsqopt} \pportwoptR \tr{\pvmu}\minv{\pvsig} +
	\frac{\Rbuj}{\psnropt}\minv{\pvsig}},\\
\end{split}
\end{equation*}
which then further simplifies to the form given.


\end{proof}

The sample statistic \ssrsqopt is, up to scaling involving \ssiz, just
Hotelling's \Tstat statistic.  \cite{anderson2003introduction} 
One can perform inference on \psnrsqopt via this statistic, at least under
Gaussian returns, where the distribution of \Tstat takes a (noncentral)
\Fstat{}-distribution. Note, however, that \psnropt is the maximal 
population \txtSR of \emph{any} portfolio, so it is an upper bound of the
\txtSR of the sample portfolio \sportwoptR. It is of little comfort
to have an estimate of \psnropt when the sample portfolio may have 
a small, or even negative, \txtSR.

Because \psnropt is an upper bound on the \txtSR of a portfolio, it
seems odd to claim that the \txtSR of the sample portfolio might be
asymptotically normal with mean \psnropt. In fact, the delta method
will fail because the gradient of \psnropt with respect to the portfolio
is zero at \pportwoptR. One solution to this puzzle
is to estimate the `\txtSNR,' incorporating a strictly positive \rfr. 
In this case a portfolio may achieve a higher value than 
$\psnropt - \fracc{\rfr}{\Rbuj}$, which is achieved by \pportwoptR,
by violating the risk budget.  To push this argument forward,
we construct a quadratic approximation to the \txtSNR function.

Suppose $\rfr > 0, \Rbuj > 0$, and assume the population parameters, \pvsm
are fixed.  Define the \txtSNR as
\begin{equation}
\pSNR{\pportw ; \pvsm, \rfr} \defeq \frac{\trAB{\pportw}{\pvmu} -
\rfr}{\sqrt{\qform{\pvsig}{\pportw}}}.
\label{eqn:snr_def}
\end{equation}
We will usually drop the dependence on \pvsm and \rfr and simply write
\pSNR{\pportw}. Defining \pportwoptR as in 
\eqnref{pportwoptR_def}, note that
\begin{equation*}
\pSNR{\pportwoptR} = \psnropt - \fracc{\rfr}{\Rbuj}.
\end{equation*}

\begin{lemma}[Quadratic Taylor expansion of \txtSNR]
\label{lemma:snr_quadratic_taylor}
Let \pportwoptR be defined in \eqnref{pportwoptR_def}, 
and let \pSNR{\pportw} be defined as in \eqnref{snr_def}. Then
\begin{multline*}
\pSNR{\pportwoptR + \vect{\epsilon}} =
\pSNR{\pportwoptR} + 
\frac{\rfr}{\Rbuj^2\psnropt}\tr{\pvmu} \vect{\epsilon} + \\ 
	\half \oneby{\Rbuj^2} 
\qform{\wrapParens{\wrapParens{\pSNR{\pportwoptR} -
	2\frac{\rfr}{\Rbuj}}\frac{\ogram{\pvmu}}{\psnrsqopt} -
\pSNR{\pportwoptR}{\pvsig}}}{%
	\vect{\epsilon}} + 
	\ldots
\end{multline*}
\end{lemma}
\begin{proof}
By Taylor's theorem,
\begin{multline*}
\pSNR{\pportwoptR + \vect{\epsilon}} =
\pSNR{\pportwoptR} + 
\wrapParens{\evalat{\dbyd{\pSNR{\vect{x}}}{\vect{x}}}{\vect{x} = \pportwoptR}}
\vect{\epsilon}
+ \\ 
\half \qform{\wrapParens{\evalat{\Hessof[\vect{x}]{\pSNR{\vect{x}}}}{\vect{x} =
\pportwoptR}}}{\vect{\epsilon}} + \ldots
\end{multline*}
By simple calculus,
\begin{equation}
\label{eqn:deriv_snr}
\dbyd{\pSNR{\pportw}}{\pportw} 
= \frac{\sqrt{\qform{\pvsig}{\pportw}}\pvmu 
- \frac{\trAB{\pportw}{\pvmu} -
  \rfr}{\sqrt{\qform{\pvsig}{\pportw}}}\pvsig\pportw}{\qform{\pvsig}{\pportw}}
= \frac{\sqrt{\qform{\pvsig}{\pportw}}\pvmu 
- \pSNR{\pportw}\pvsig\pportw}{\qform{\pvsig}{\pportw}}
\end{equation}
To compute the Hessian, take the derivative of this gradient:
\begin{equation}
\label{eqn:hess_snr}
\Hessof{\pSNR{}} = 
\frac{%
- \wrapBracks{%
\pvmu\tr{\sportw}\pvsig + \pvsig\sportw\tr{\pvmu}
- 3 \pSNR{\sportw}\frac{\pvsig\sportw\tr{\sportw}\pvsig}{\sqrt{\qform{\pvsig}{\sportw}}}
+ \pSNR{\sportw}\sqrt{\qform{\pvsig}{\sportw}}\pvsig}}{%
\wrapParens{\qform{\pvsig}{\sportw}}^{3/2}} 
\end{equation}
At $\sportw = \pportwoptR =
\wrapParens{\fracc{\Rbuj}{\sqrt{\qiform{\pvsig}{\pvmu}}}}
\minvAB{\pvsig}{\pvmu},$ the derivative takes value
$$
\evalat{\dbyd{\pSNR{\vect{x}}}{\vect{x}}}{\vect{x} = \pportwoptR} = 
\frac{\Rbuj\pvmu - \wrapParens{\psnropt -
\fracc{\rfr}{\Rbuj}}\frac{\Rbuj}{\psnropt}\pvmu}{\Rbuj^2}
= \frac{\rfr}{\Rbuj^2\psnropt}\pvmu,
$$
and the Hessian takes value
$$
\evalat{\Hessof[\vect{x}]{\pSNR{\vect{x}}}}{\vect{x}=\pportwoptR} = 
\frac{
	\oneby{\psnrsqopt}\wrapParens{\psnropt - 3\frac{\rfr}{\Rbuj}}\ogram{\pvmu} -
	\wrapParens{\psnropt - \frac{\rfr}{\Rbuj}}\pvsig}{\Rbuj^2},
$$
completing the proof.
\end{proof}

Combining \lemmaref{snr_quadratic_taylor} and 
\corollaryref{portwoptR_dist}, we get the following:
\begin{corollary}
\label{corollary:portwoptR_hc_dist}
Let \pportwoptR and \sportwoptR be defined as in 
\corollaryref{portwoptR_dist}. As per \lemmaref{sr_optimal_portfolio},
$\pSNR{\pportwoptR} = \psnropt - \fracc{\rfr}{\Rbuj}$.
Let \pvvar be the variance of $\fvech{\ogram{\avreti}}$.

Then, asymptotically in \ssiz,
\begin{equation}
\pSNR{\sportwoptR} 
\rightsquigarrow 
\normlaw{\pSNR{\pportwoptR},\oneby{\ssiz}\qform{\pvvar}{\vect{h}}},
\label{eqn:mvclt_snr}
\end{equation}
where
\begin{equation}
\tr{\vect{h}} 
 = - \frac{\rfr}{\Rbuj\psnrsqopt}
\asrowvec{\oneby{2},\tr{\pvmu},\vzero}
\wrapParens{-\EXD{\wrapParens{\AkronA{\minv{\pvsm}}}}}.
\end{equation}

Moreover, we may express \tr{\vect{h}} as
\begin{equation}
\label{eqn:shorter_h_form}
\tr{\vect{h}} =
- \frac{\rfr}{2\Rbuj\psnrsqopt}
\wrapParens{\asrowvec{1 + \psnrsqopt, -\tr{\pvmu}\minv{\pvsig}} \kron
	\asrowvec{1 - \psnrsqopt, {\tr{\pvmu}\minv{\pvsig}}}}\Dupp.
\end{equation}

\end{corollary}
\begin{proof}
By the delta method, and the chain rule,
$$
\pSNR{\sportwoptR} 
\rightsquigarrow 
\normlaw{\pSNR{\pportwoptR},\oneby{\ssiz}\qform{\pvvar}{\vect{h}}},
\,\,\,
\tr{\vect{h}} = \tr{\dbyd{\pSNR{\pportwoptR}}{\pportwoptR}}
\dbyd{\pportwoptR}{\fvech{\pvsm}}. 
$$
From \corollaryref{portwoptR_dist}, we have
$$
\tr{\vect{h}} = {\tr{\dbyd{\pSNR{\pportwoptR}}{\pportwoptR}} 
{\asrowvec{-\oneby{2\psnrsqopt} \pportwoptR,
-\frac{\Rbuj}{\psnropt}\eye[\nlatf],\mzero}}}
\wrapParens{-\EXD{\wrapParens{\AkronA{\minv{\pvsm}}}}}.
$$
Taking the derivative of \pSNR{\cdot} from \lemmaref{snr_quadratic_taylor}, 
\begin{equation}
\begin{split}
{\tr{\dbyd{\pSNR{\pportwoptR}}{\pportwoptR}} 
{\asrowvec{-\oneby{2\psnrsqopt} \pportwoptR,
-\frac{\Rbuj}{\psnropt}\eye[\nlatf],\mzero}}} 
 &= - \frac{\rfr}{\Rbuj^2\psnropt}\tr{\pvmu}
\asrowvec{\oneby{2\psnrsqopt} \pportwoptR,
\frac{\Rbuj}{\psnropt}\eye[\nlatf],\mzero},\\
 &= - \frac{\rfr}{\Rbuj\psnrsqopt}
\asrowvec{\oneby{2},\tr{\pvmu},\vzero}.
\end{split}
\end{equation}

To establish \eqnref{shorter_h_form}, one proceeds as in the proof of
\corollaryref{portwoptR_dist}.
\end{proof}
\begin{caution}
Since \pvmu and \pvsig are population parameters, \pSNR{\sportwoptR} is
an unobserved quantity. Nevertheless, we can estimate the variance of
\pSNR{\sportwoptR}, and possibly construct confidence intervals on it
using sample statistics.
\end{caution}
This corollary is useless in the case where $\rfr=0$, and
gives somewhat puzzling results when one considers 
$\rfr\searrow 0$, since it suggest that the variance of
\pSNR{\sportwoptR} goes to zero. This is \emph{not} the
case, because the rate of convergence in the corollary
is a function of \rfr. To consider the $\rfr=0$ case,
one must take the quadratic Taylor expansion of the 
\txtSNR function.
\begin{corollary}
\label{corollary:portwoptR_hc_dist_two}
Suppose $\Rbuj > 0$, and $\rfr=0$, thus
\begin{equation}
\pSNR{\sportw} \defeq \frac{\trAB{\sportw}{\pvmu}}{%
\sqrt{\qform{\pvsig}{\sportw}}}.
\label{eqn:snr_def}
\end{equation}
Let \pportwoptR and \sportwoptR be defined as in 
\corollaryref{portwoptR_dist}. As per \lemmaref{sr_optimal_portfolio},
$\pSNR{\pportwoptR} = \psnropt$, where we take $\rfr=0$.
Let \pvvar be the variance of $\fvech{\ogram{\avreti}}$.

Then, asymptotically in \ssiz,
\begin{equation}
\ssiz\wrapBracks{\pSNR{\sportwoptR} - \pSNR{\pportwoptR}} 
\rightsquigarrow 
\half 
\trace{\qform{\Mtx{F}}{\wrapParens{{\Mtx{H}}\chol{\pvvar}}} \ogram{\vect{z}}},
\end{equation}
where $z\sim\normlaw{\vzero,\eye[\nlatf]}$, and where
\begin{align*}
\Mtx{F}
&= \oneby{\Rbuj^2}
\wrapParens{\frac{\ogram{\pvmu}}{\psnropt} - \psnropt{\pvsig}},&\\
\Mtx{H} &= {- \asrowvec{\oneby{2\psnrsqopt} \pportwoptR,
\frac{\Rbuj}{\psnropt}\eye[\nlatf],\mzero}} 
\wrapParens{-\EXD{\wrapParens{\AkronA{\minv{\pvsm}}}}},&
\end{align*}
as in \corollaryref{portwoptR_dist}.
\end{corollary}
\begin{proof}
From \lemmaref{snr_quadratic_taylor} 
\begin{align*}
\pSNR{\pportwoptR + \vect{\epsilon}} &=
	\pSNR{\pportwoptR} + \half \qform{\Mtx{F}}{\vect{\epsilon}} + \ldots,\\
\Mtx{F} &= \oneby{\Rbuj^2} \wrapParens{\frac{\ogram{\pvmu}}{\psnropt} - \psnropt{\pvsig}}.
\end{align*}

By \corollaryref{portwoptR_dist}, 
$$
\vect{\epsilon} = \sportwoptR - \pportwoptR 
\rightsquigarrow 
\normlaw{0,\oneby{\ssiz}\qoform{\pvvar}{\Mtx{H}}},
$$
so, asymptotically $\vect{\epsilon}\rightsquigarrow \oneby{\sqrt{\ssiz}} 
\wrapParens{{\Mtx{H}}\chol{\pvvar}}\vect{z},$ where
$\vect{z}\sim\normlaw{\vzero,\eye[\nlatf]},$ and
the result follows.
\end{proof}



We now seek to link the sample estimate of the optimal achievable \txtSR with
the achieved \txtSNR of the sample \txtMP. The magnitude of
\ssrsqopt (along with \ssiz) is the only information on which we might estimate
whether the sample \txtMP is any good. If we view them both as functions of
\minv{\pvsm}, we can find their expected values and any covariance between
them. Somewhat surprisingly, the two quantities are asymptotically almost
uncorrelated. 

Thus for the following theorem, let us abuse notation to express the \txtSNR
function, defined in \eqnref{snr_def} as a function of some vector:
\begin{equation}
\pSNR{\vect{x} ; \pvsm, \rfr} =
	\frac{\trAB{\wrapParens{\asrowvec{\vzero,\eye[\nlatf],\mzero}\vect{x}}}{\pvmu} -
	\rfr}{\sqrt{\qform{\pvsig}{\wrapParens{\asrowvec{\vzero,\eye[\nlatf],\mzero} \vect{x}}}}}.
\label{eqn:snr_def_II}
\end{equation}
Then \pSNR{\fvech{\minv{\svsm}} ; \pvsm, \rfr} corresponds to the definition in
\eqnref{snr_def}.
Moreover, define the ``optimal \txtSNR'' function also as a function of a
vector as
\begin{equation}
	\pSNROPT{\vect{x}} = \sqrt{\trAB{\basev[1]}{\vect{x}} - 1}.
\label{eqn:snropt_def}
\end{equation}
We have $\pSNROPT{\fvech{\minv{\pvsm}}} = \psnropt$, as desired.
In an abuse of notation we will simply write $\pSNR{\minv{\svsm}; \pvsm, \rfr}$ and
$\pSNROPT{\minv{\svsm}}$ instead of writing out the vector function.

\begin{theorem}
\label{theorem:psnr_and_ssropt_moments}
Let \svsm be the unbiased sample estimate of 
\pvsm, based on \ssiz \iid samples of \vreti.
Assume $\psnropt > 0$.
Let \pvvar be the variance of $\fvech{\ogram{\avreti}}$.
	Define \pSNR{\cdot} and \pSNROPT{\cdot} as in 
\eqnref{snr_def_II} and \eqnref{snropt_def}.
Then, asymptotically in \ssiz, 
\begin{equation}
\begin{split}
\pSNR{\minv{\svsm} ; \pvsm, 0} &\rightsquigarrow \psnropt + 
	  \oneby{2\ssiz} \trace{\qform{\Mtx{F}}{\wrapParens{{\Mtx{H}}\chol{\pvvar}}} \ogram{\vect{z}}},\\
\ssropt =	\pSNROPT{\minv{\svsm}} &\rightsquigarrow \psnropt +
	\frac{\trAB{\vect{h}}{\chol{\pvvar}\vect{z}}}{2\psnropt\sqrt{\ssiz}}
	 - \frac{ \trace{\gram{\wrapParens{\trAB{\vect{h}}{\chol{\pvvar}}}}
	 \ogram{\vect{z}}} }{%
	 8\psnropt^3 \ssiz},\\
\ssropt^{-1} = \pSNROPT{\minv{\svsm}}^{-1} 
	&\rightsquigarrow \psnropt^{-1} \wrapParens{1 -
	\frac{\trAB{\vect{h}}{\chol{\pvvar}\vect{z}}}{2\psnropt^2\sqrt{\ssiz}}
	 + \frac{3 \trace{\gram{\wrapParens{\trAB{\vect{h}}{\chol{\pvvar}}}}
	 \ogram{\vect{z}}} }{%
	 8\psnropt^4 \ssiz}},
\end{split}
\end{equation}
where $z\sim\normlaw{\vzero,\eye[\nlatf]}$, where the matrices
\Mtx{H} and \Mtx{F} are as given in \corollaryref{portwoptR_hc_dist_two}, and
where
\begin{equation*}
	\tr{\vect{h}} = - \wrapParens{\AkronA{\asrowvec{1 + \psnrsqopt,-\trAB{\pvmu}{\minv{\pvsig}}}}}\Dupp.
\end{equation*}
\end{theorem}
\begin{proof}
The distribution of \pSNR{\vect{x} ; \pvsm, 0} is a restatement of 
\corollaryref{portwoptR_hc_dist_two}. 
To prove the distribution of \ssropt, perform a Taylor's series expansion of the
function \pSNROPT{\cdot}:
\begin{equation*}
\begin{split}
\pSNROPT{\vect{x} + \vect{\epsilon}} 
 &= \pSNROPT{\vect{x}} 
  + \dbyd{\pSNROPT{\vect{x}}}{\vect{x}}\vect{\epsilon} 
  + \half\qform{\Hessof[\vect{x}]{\pSNROPT{\vect{x}}}}{\vect{\epsilon}} 
  + \ldots\\
 &= \pSNROPT{\vect{x}} 
  + \frac{\trAB{\basev[1]}{\vect{\epsilon}}}{2 \pSNROPT{x}} 
  - \half\qform{\frac{\ogram{\basev[1]}}{4 \pSNROPT{\vect{x}}^3}}{\vect{\epsilon}}
  + \ldots
\end{split}
\end{equation*}
By similar reasoning,
\begin{equation*}
\begin{split}
\oneby{\pSNROPT{\vect{x} + \vect{\epsilon}}}
&= \oneby{\pSNROPT{\vect{x}}}
  - \frac{\trAB{\basev[1]}{\vect{\epsilon}}}{2 \pSNROPT{x}^3} 
  + \half\qform{\frac{3\ogram{\basev[1]}}{8 \pSNROPT{\vect{x}}^5}}{\vect{\epsilon}}
  + \ldots
\end{split}
\end{equation*}

Now by \theoremref{inv_distribution}, we know the asymptotic distribution of
\minv{\pvsm}, and thus we get the claimed distributional form for some
\vect{h}, with
\begin{equation*}
\vect{h} = \trAB{\basev[1]}{\wrapParens{-\EXD{\wrapParens{\AkronA{\minv{\pvsm}}}}}}.
\end{equation*}
We can eliminate the \Elim and write the \basev[1] as a
$\wrapParens{\nlatf+1}^2$ length vector, or rather as
$\AkronA{\tr{\basev[1]}}$, where this \basev[1] is a $\wrapParens{\nlatf+1}$
length vector. The product of Kronecker products is the Kronecker product of
matrix products so 
\begin{equation*}
	\vect{h} = - \AkronA{\wrapParens{\tr{\basev[1]}\minv{\pvsm}}} \Dupp,
\end{equation*}
establishing the identity of \vect{h}.
\end{proof}

This theorem suggests the use of the observable quantity \ssropt to perform
inference on the unknown quantity, \pSNR{\minv{\svsm} ; \pvsm , 0}.
Asymptotically they have low correlation, and small error between them,
which we quantify in the following corollaries.

\begin{corollary}
The covariance of these two quantities is asymptotically \bigo{\ssiz^{-2}}:
\begin{equation}
	\COV{\pSNR{\minv{\svsm} ; \pvsm, 0}}{\pSNROPT{\minv{\svsm}}} \rightsquigarrow \bigo{\ssiz^{-2}},
\end{equation}
and their correlation is asymptotically \bigo{\ssiz^{-\halff}}.
\end{corollary}
\begin{proof}
To find the moments of \pSNR{\minv{\svsm} ; \pvsm, 0} and \ssropt, we take
\vect{z} to be multivariate standard normal, and thus odd powered products
have zero expectation. Their covariance comes entirely from the product
of the quadratic (in \vect{z}) terms. By the same logic, the asymptotic 
standard error of \pSNR{\minv{\svsm} ; \pvsm, 0} is \bigo{\ssiz^{-1}},
while that of \ssropt is \bigo{\ssiz^{-\halff}}.
\end{proof}

\begin{corollary}
\label{corollary:psnr_minus_ssropt_moments}
The difference $\pSNR{\minv{\svsm} ; \pvsm, 0} - \ssropt$ has the
following asymptotic mean and variance:
\begin{align}
\label{eqn:psnr_minus_ssropt_moments}
\E{\pSNR{\minv{\svsm} ; \pvsm, 0} - \ssropt} &\rightsquigarrow 
	  \oneby{2\ssiz} \trace{%
\wrapParens{ \qform{\Mtx{F}}{\Mtx{H}} + \frac{ \ogram{\vect{h}} }{4\psnropt^3}
} \pvvar }.\\
\VAR{\pSNR{\minv{\svsm} ; \pvsm, 0} - \ssropt} &\rightsquigarrow 
	\oneby{4\psnrsqopt\ssiz} \qform{\pvvar}{\vect{h}},
\end{align}
where \Mtx{H}, \Mtx{F}, and \vect{h} are given in the theorem.
Their \emph{ratio} has the 
following asymptotic mean and variance:
\begin{align}
\label{eqn:psnr_div_ssropt_moments}
\E{\frac{\pSNR{\minv{\svsm} ; \pvsm, 0}}{\ssropt}} &\rightsquigarrow 
	  1 + \oneby{2\ssiz\psnropt} \trace{%
\wrapParens{ \qform{\Mtx{F}}{\Mtx{H}} + \frac{ 3\ogram{\vect{h}} }{4\psnropt^3}
} \pvvar }.\\
\VAR{\frac{\pSNR{\minv{\svsm} ; \pvsm, 0}}{\ssropt}} &\rightsquigarrow 
	\oneby{4\psnropt^4\ssiz} \qform{\pvvar}{\vect{h}}.
\end{align}

\end{corollary}

This corollary gives a recipe for building confidence intervals on
\pSNR{\minv{\svsm} ; \pvsm, 0} via the observed \psnropt, namely by
plugging in sample estimates for \pvvar and \Mtx{H}, \Mtx{F} and \vect{h}.
This result is comparable to the ``Sharpe Ratio Information Criterion''
estimator for \pSNR{\minv{\svsm} ; \pvsm, 0}.  \cite{paulsen2016noise}

To use this corollary, 
one may need a compact expression for \qform{\Mtx{F}}{\Mtx{H}}.
We start with \eqnref{shorter_H_form}, and write
\begin{equation*}
\begin{split}
	\qform{\Mtx{F}}{\Mtx{H}} &=
	\tr{\Mtx{H}} 
	\wrapParens{\wrapParens{\trAB{\basev[1]}{\minv{\pvsm}}} \kron
	\Mtx{F}\asrowvec{\frac{1 - \psnrsqopt}{2\psnropt}\pportwoptR,
	\frac{\Rbuj}{\psnropt}\wrapParens{\minv{\pvsig} -
	\frac{\minv{\pvsig}\ogram{\pvmu}\minv{\pvsig}}{2\psnrsqopt}}}}\Dupp.
\end{split}
\end{equation*}
But note that $\Mtx{F} \minv{\pvsig}\pvmu = \vzero$ (and so $\Mtx{F}
\pportwoptR = \vzero$) so we may write
\begin{align}
	\qform{\Mtx{F}}{\Mtx{H}} &=
	\tr{\Mtx{H}} 
	\wrapParens{\wrapParens{\trAB{\basev[1]}{\minv{\pvsm}}} \kron
	\asrowvec{\vzero, \frac{\Rbuj}{\psnropt}\Mtx{F}\minv{\pvsig}}}\Dupp,\nonumber\\
	&=
	\oneby{\Rbuj} \tr{\Mtx{H}} 
	\wrapParens{\wrapParens{\trAB{\basev[1]}{\minv{\pvsm}}} \kron
	\asrowvec{\vzero, \frac{\ogram{\pvmu}\minv{\pvsig}}{\psnrsqopt} - \eye}} \Dupp,\nonumber\\
	&=
	\oneby{\psnropt}\qform{\wrapParens{%
		\wrapParens{\minv{\pvsm}\ogram{\basev[1]}\minv{\pvsm}} \kron
		\twobytwo{0}{\tr{\vzero}}{\vzero}{\frac{\minv{\pvsig}\ogram{\pvmu}\minv{\pvsig}}{\psnrsqopt}
		- \minv{\pvsig}} }}{\Dupp}.
\label{eqn:shorter_HFH_form}
\end{align}
Similarly, a compact form for \ogram{\vect{h}}:
\begin{align}
	\ogram{\vect{h}} &= 
	\qform{\wrapParens{%
		\AkronA{\wrapParens{\trAB{\basev[1]}{\minv{\pvsm}}\minv{\pvsm}\basev[1]}}}}{\Dupp}.
\label{eqn:shorter_hh_form}
\end{align}





\section{Distribution under Gaussian returns}

The goal of this section is to derive a variant of 
\theoremref{inv_distribution} for the case where \vreti follow a 
multivariate Gaussian distribution. First, 
assuming $\vreti\sim\normlaw{\pvmu,\pvsig}$, we can express the density
of \vreti, and of \svsm, in terms of \nlatf, \ssiz, and \pvsm.

\begin{lemma}[Gaussian sample density]
Suppose $\vreti\sim\normlaw{\pvmu,\pvsig}$. Letting 
$\avreti = \asvec{1,\tr{\vreti}}$, and $\pvsm = \E{\ogram{\avreti}}$,
then the negative log likelihood of \vreti is
\begin{equation}
- \log\normpdf{\vreti}{\pvmu,\pvsig} = 
  c_{\nlatf} 
+ \half \logdet{\pvsm} 
+ \half \trace{\minv{\pvsm}\ogram{\avreti}},
\end{equation}
for the constant 
$c_{\nlatf} = -\half + \half[\nlatf]\log\wrapParens{2\pi}.$
\label{lemma:x_dist_gaussian}
\end{lemma}
\begin{proof}
By the block determinant formula,
\begin{equation*}
\det{\pvsm} 
= \det{1}\det{\wrapParens{\pvsig + \ogram{\pvmu}} - \pvmu 1^{-1} \tr{\pvmu}}
= \det{\pvsig}.
\end{equation*}
Note also that
\begin{equation*}
\qiform{\pvsig}{\wrapParens{\vreti - \pvmu}} = 
\qiform{\pvsm}{\avreti} - 1.
\end{equation*}
These relationships hold without assuming a particular distribution for 
\vreti. 

The density of \vreti is then
\begin{equation*}
\begin{split}
\normpdf{\vreti}{\pvmu,\pvsig} &= \frac{1}{\sqrt{\wrapParens{2\pi}^{\nlatf}\det{\pvsig}}} 
\longexp{-\half \qiform{\pvsig}{\wrapParens{\vreti - \pvmu}}},\\
 &= \frac{\detpow{\pvsig}{-\half}}{\wrapParens{2\pi}^{\nlatf/2}}
\longexp{-\half \wrapParens{\qiform{\pvsm}{\avreti} - 1}},\\
 &= \wrapParens{2\pi}^{-\nlatf/2} \detpow{\pvsm}{-\half}
\longexp{-\half \wrapParens{\qiform{\pvsm}{\avreti} - 1}},\\
 &= \wrapParens{2\pi}^{-\nlatf/2}
\longexp{\half - \half \logdet{\pvsm} - \half \trace{\minv{\pvsm}\ogram{\avreti}}},
\end{split}
\end{equation*}
and the result follows.
\end{proof}

\begin{lemma}[Gaussian second moment matrix density]
Let $\vreti\sim\normlaw{\pvmu,\pvsig}$, 
$\avreti = \asvec{1,\tr{\vreti}}$, 
and $\pvsm = \E{\ogram{\avreti}}$. 
Given \ssiz \iid samples \vreti[i], let 
Let $\svsm = \oneby{\ssiz}\sum_i \ogram{\avreti[i]}$.
Then the density of \svsm is 
\begin{equation}
\FOOpdf{}{\svsm}{\pvsm} =
\longexp{c'_{\ssiz,\nlatf}}\frac{\det{\svsm}^{\half[\ssiz-\nlatf-2]}}{\det{\pvsm}^{\half[\ssiz]}}
\longexp{-\half[\ssiz]\trace{\minv{\pvsm}\svsm}},
\label{eqn:theta_dist_gaussian}
\end{equation}
for some $c'_{\ssiz,\nlatf}.$
\label{lemma:theta_dist_gaussian}
\end{lemma}
\begin{proof}
Let \amreti be the matrix
whose rows are the vectors \tr{\vreti[i]}. 
From \lemmaref{x_dist_gaussian}, and using linearity of the trace, 
the negative log density of \amreti is
\begin{equation*}
\begin{split}
- \log\normpdf{\amreti}{\pvsm} &=
  \ssiz c_{\nlatf} 
+ \half[\ssiz] \logdet{\pvsm} 
+ \half \trace{\minv{\pvsm}\gram{\amreti}},\\
\therefore
\frac{- 2\log\normpdf{\amreti}{\pvsm}}{\ssiz} &= 
  2 c_{\nlatf} 
+ \logdet{\pvsm} 
+ \trace{\minv{\pvsm}\svsm}.
\end{split}
\end{equation*}

By Lemma (5.1.1) of Press \cite{press2012applied}, this can be expressed
as a density on \svsm:
\begin{equation*}
\begin{split}
\frac{- 2\log\FOOpdf{}{\svsm}{\pvsm}}{\ssiz} 
&= \frac{- 2\log\normpdf{\amreti}{\pvsm}}{\ssiz}
	-\frac{2}{\ssiz}\wrapParens{\half[\ssiz-\nlatf-2]\logdet{\svsm}}\\
&\phantom{=}\,
	-\frac{2}{\ssiz}\wrapParens{\half[\nlatf+1]\wrapParens{\ssiz - \half[\nlatf]} \log\pi -
	\sum_{j=1}^{\nlatf+1} \log\funcit{\Gamma}{\half[\ssiz +1-j]}},\\
&= \wrapBracks{2c_{\nlatf} - \frac{\nlatf+1}{\ssiz}\wrapParens{\ssiz - \half[\nlatf]} \log\pi 
	- \frac{2}{\ssiz} \sum_{j=1}^{\nlatf+1}
	\log\funcit{\Gamma}{\half[\ssiz +1-j]}}\\
&\phantom{=}\,
	+ \logdet{\pvsm} - \frac{\ssiz-\nlatf-2}{\ssiz}\logdet{\svsm}
	+ \trace{\minv{\pvsm}\svsm},\\
&= c'_{\ssiz,\nlatf} 
	- \log\frac{\det{\svsm}^{\frac{\ssiz-\nlatf-2}{\ssiz}}}{\det{\pvsm}}
	+ \trace{\minv{\pvsm}\svsm},
\end{split}
\end{equation*}
where $c'_{\ssiz,\nlatf}$ is the term in brackets on the third line.
Factoring out $\fracc{-2}{\ssiz}$ and taking an exponent gives the
result.
\end{proof}
\begin{corollary}
The random variable $\ssiz\svsm$ has the same density, up to a constant 
in \nlatf and \ssiz, as a 
$\nlatf+1$-dimensional Wishart random variable with \ssiz degrees of freedom
and scale matrix \pvsm. Thus $\ssiz\svsm$ is a \emph{conditional} Wishart,
conditional on $\svsm_{1,1} = 1$.  \cite{press2012applied,anderson2003introduction}
\end{corollary}
\begin{corollary}
The derivatives of log likelihood are given by 
\begin{equation}
\begin{split}
\drbydr{\log\FOOpdf{}{\svsm}{\pvsm}}{\fvec{\pvsm}} 
	&= - \half[\ssiz]\tr{\wrapBracks{\fvec{\minv{\pvsm} -
\minv{\pvsm}\svsm\minv{\pvsm}}}},\\
\drbydr{\log\FOOpdf{}{\svsm}{\pvsm}}{\fvec{\minv{\pvsm}}} 
	&= -\half[\ssiz]\tr{\wrapBracks{\fvec{\pvsm - \svsm}}}.
\end{split}
\label{eqn:deriv_gauss_loglik}
\end{equation}
\end{corollary}
\begin{proof}
Plugging in the log likelihood gives
\begin{equation*}
\drbydr{\log\FOOpdf{}{\svsm}{\pvsm}}{\fvec{\pvsm}}
 = - \half[\ssiz]\wrapBracks{\drbydr{\logdet{\pvsm}}{\fvec{\pvsm}} +
\drbydr{\trace{\minv{\pvsm}\svsm}}{\fvec{\pvsm}}},
\end{equation*}
and then standard matrix calculus gives the first 
result. \cite{magnus1999matrix,petersen2012matrix}
Proceeding similarly gives the second.
\end{proof}

This immediately gives us the Maximum Likelihood Estimator.

\begin{corollary}[MLE]
\svsm is the maximum likelihood estimator of \pvsm. \label{corollary:theta_mle}
\end{corollary}

To compute the covariance of \fvech{\pvsm}, \pvvar, in the Gaussian case, one
can compute the Fisher Information, then appeal to the fact that
\pvsm is the MLE. However, because the first element of \fvech{\pvsm}
is a deterministic $1$, the first row and column of \pvvar are all
zeros. This is an unfortunate wrinkle. One solution would be to
to compute the Fisher Information with respect to the nonredundant variables;
however, a direct brute-force approach is also possible,
and gives a slightly more general result, as in the following section.

\subsection{Distribution under elliptical returns}

\providecommand{\kurty}[1][]{\mathSUB{\kappa}{#1}}

We pursue a slightly more general result on the distribution of
\svsm, assuming that \vreti are drawn independently from 
an \emph{elliptical} distribution, with mean \pvmu, covariance
\pvsig, and `kurtosis factor' \kurty, by which we mean the
excess kurtosis of each \vreti[i] is $3 \wrapParens{\kurty - 1}$.
In the Gaussian case, $\kurty=1$.

To be concrete, we suppose $\vreti = \pvmu + \frac{a \chol{\Mtx{\Lambda}}
\vect{z}}{\norm{\vect{z}}}$, where $\vect\sim\normlaw{\vzero,\eye[n]}$,
and where $a$ is a scalar random variable independent of \vect{z}. The
covariance \pvsig is related to the matrix \Mtx{\Lambda} via
$$
\pvsig = \frac{\E{a^2}}{n}\Mtx{\Lambda}.
$$
The kurtosis parameter is then defined as
$$
\kurty = \frac{n}{n+2}\frac{\E{a^4}}{\E{a^2}}.
$$
An extension of Isserlis' Theorem to the elliptical distribution gives moments
of products of elements of \vreti.  \cite{vignat2007extension,KAN2008542}
This result is comparable to the covariance of the centered second moments
given as Equation (2.1) of Iwashita and Siotani, but is applicable to
the uncentered second moment. \cite{CJS:CJS273}

\begin{theorem}
\label{theorem:theta_covar_elliptical}
Let \svsm be the unbiased sample estimate of 
\pvsm, based on \ssiz \iid samples of \vreti, assumed to take an
elliptical distribution with kurtosis parameter \kurty.
In analogue to how \avreti are built from \vreti, define
\begin{equation}
\label{eqn:apvmu_apvsig}
	\apvmu \defeq \asvec{1,\tr{\pvmu}},\quad\mbox{and}\quad
	\apvsig \defeq \twobytwo{0}{\tr{\vzero}}{\vzero}{\pvsig}.
\end{equation}
Note that $\pvsm = \apvsig + \ogram{\apvmu}.$
Then
\begin{align}
\label{eqn:theta_covar_elliptical}
	{\ssiz}\VAR{\fvec{\svsm}} = \pvvar[0] &= \wrapParens{\kurty-1}
	\wrapBracks{\ogram{\fvec{\apvsig}} + \wrapParens{\eye + \Komm}\AkronA{\apvsig}}\\
	&\phantom{=}\,+ \wrapParens{\eye + \Komm}\wrapBracks{\AkronA{\pvsm} -
	\AkronA{\ogram{\apvmu}}}.\nonumber
\end{align}
\end{theorem}
As it is cumbersome and unenlightening, we relegate the proof to
the appendix.
The central limit theorem then gives us the following corollary.
\begin{corollary}
\label{corollary:theta_asym_var_elliptical}
Under the conditions of the previous theorem, 
asymptotically in \ssiz, 
\begin{equation}
\sqrt{\ssiz}\wrapParens{\fvech{\svsm} - \fvech{\pvsm}} 
\rightsquigarrow 
\normlaw{0,\qoform{\pvvar[0]}{\Elim}},
	\label{eqn:mvclt_isvsm_elliptical}
\end{equation}
where \pvvar[0] is defined in \eqnref{theta_covar_elliptical}.
\end{corollary}
Using \theoremref{inv_distribution} we also immediately get the following
\begin{corollary}
\label{corollary:invtheta_asym_var_elliptical}
Under the conditions of the previous theorem, 
asymptotically in \ssiz, 
\begin{equation}
\sqrt{\ssiz}\wrapParens{\fvech{\minv{\svsm}} - \fvech{\minv{\pvsm}}} 
	\rightsquigarrow \normlaw{0,\qoform{\qoform{\pvvar[0]}{\Elim}}{\Mtx{H}}},
\label{eqn:mvclt_isvsm_elliptic}
\end{equation}
where \pvvar[0] is defined in \eqnref{theta_covar_elliptical}, and
\begin{equation}
\Mtx{H} = -\EXD{\wrapParens{\AkronA{\minv{\pvsm}}}}.
\end{equation}
\end{corollary}
An uglier form of the same corollary gives the covariance
explicitly. See the appendix for the proof.
\begin{corollary}
\label{corollary:invtheta_asym_var_elliptical_explicit}
Under the conditions of the previous theorem, 
asymptotically in \ssiz, 
\begin{equation}
\sqrt{\ssiz}\wrapParens{\fvech{\minv{\svsm}} - \fvech{\minv{\pvsm}}} 
	\rightsquigarrow \normlaw{0,\Mtx{B}},
\label{eqn:mvclt_isvsm_elliptic_explicit}
\end{equation}
where
\begin{align*}
	\Mtx{B} &= \wrapParens{\kurty-1}\Elim\wrapBracks{\ogram{\fvec{\minv{\pvsm} -
	\ogram{\basev[1]}}}}
		\tr{\Elim}\\
		&\phantom{=}\,+ 2\wrapParens{\kurty-1}\Elim\Simm
		\wrapBracks{\AkronA{\wrapParens{\minv{\pvsm} - \ogram{\basev[1]}}}}
		\tr{\Simm}
		\tr{\Elim}\\
		&\phantom{=}\,+ 2\Elim\Simm
		\wrapBracks{\AkronA{\minv{\pvsm}} - \AkronA{\ogram{\basev[1]}}}
		\tr{\Simm}
		\tr{\Elim}.
\end{align*}
\end{corollary}
We are often concerned with the \txtSNR and sample \txtMP, whose joint
asymptotic distribution we can pick out from the previous corollary:
\begin{corollary}
\label{corollary:invtheta_bits_asym_var_elliptical_explicit}
Under the conditions of the previous theorem, 
asymptotically in \ssiz, 
\begin{equation}
\sqrt{\ssiz}\wrapParens{
	\twobyone{1+\ssrsqopt}{-\sportwopt} -
	\twobyone{1+\psnrsqopt}{-\pportwopt}}
	\rightsquigarrow \normlaw{0,\Mtx{C}},
\label{eqn:mvclt_isvsm_bits_elliptic_explicit}
\end{equation}
where
\begin{equation*}
\Mtx{C} = \twobytwo{%
	2\psnrsqopt\wrapParens{2+\psnrsqopt} + 3 \wrapParens{\kurty-1}\psnropt^4 }{%
	-\wrapParens{2\wrapParens{1+\psnrsqopt} + 3\wrapParens{\kurty-1}\psnrsqopt}\tr{\pportwopt}}{%
	-\wrapParens{2\wrapParens{1+\psnrsqopt} + 3\wrapParens{\kurty-1}\psnrsqopt}\pportwopt}{%
	\wrapParens{1+\kurty\psnrsqopt}\minv{\pvsig}
	+ \wrapParens{2\kurty-1}\ogram{\pportwopt} }.
\end{equation*}

Furthermore, if \Mtx{Q} is an orthogonal matrix ($\ogram{\Mtx{Q}} = \eye$) such
that 
$$
\Mtx{Q} \minv{\wrapParens{\chol{\pvsig}}}\pvmu = \psnropt\basev[1],
$$
where $\chol{\pvsig}$ is the lower triangular Cholesky factor of \pvsig, then
asymptotically in \ssiz, 
\begin{equation}
\sqrt{\ssiz}\wrapParens{
	\twobyone{1+\ssrsqopt}{\Mtx{Q}\trchol{\pvsig}\sportwopt} -
	\twobyone{1+\psnrsqopt}{\psnropt\basev[1]}}
	\rightsquigarrow \normlaw{0,\Mtx{D}},
\label{eqn:mvclt_isvsm_qbits_elliptic_explicit}
\end{equation}
where
\begin{equation*}
\Mtx{D} = \twobytwo{%
	2\psnrsqopt\wrapParens{2+\psnrsqopt} + 3 \wrapParens{\kurty-1}\psnropt^4 }{%
	\wrapParens{2\wrapParens{1+\psnrsqopt} +
	3\wrapParens{\kurty-1}\psnrsqopt}\psnropt\tr{\basev[1]}}{%
	\wrapParens{2\wrapParens{1+\psnrsqopt} +
	3\wrapParens{\kurty-1}\psnrsqopt}\psnropt\basev[1]}{%
	\wrapParens{1+\kurty\psnrsqopt}\eye
	+ \wrapParens{2\kurty-1}\psnrsqopt\ogram{\basev[1]} }.
\end{equation*}
\end{corollary}

We note that the asymptotic variance of \ssrsqopt we find here for the case of Gaussian returns 
($\kurty=1$) is consistent with the exact variance one computes from the 
non-central \flaw{} distribution, via the connection to Hotelling's \Tstat{}.
That variance (assuming, as we do here, that \svsig is estimated with \ssiz in the numerator,
not $\ssiz-1$) is 
$$
2\frac{\ssiz^2\wrapParens{\psrUL{4}{*} + 2\psnrsqopt} + \ssiz\wrapParens{\nlatf - 4\psnrsqopt} -2\nlatf}{%
	\wrapParens{\ssiz-\nlatf-2}^2\wrapParens{\ssiz-\nlatf-4}} 
	= 2\psnrsqopt\frac{2+\psnrsqopt}{\ssiz} + \bigo{\ssiz^{-2}}.
$$
Our asymptotic variance captures only the leading term, as one would expect.

The choice to rescale by \Mtx{Q} and \trchol{\pvsig} in
\eqnref{mvclt_isvsm_qbits_elliptic_explicit} is worthy of explanation.
First note that the true expected return of \sportwopt is equal to
\begin{align*}
\tr{\pvmu}\sportwopt 
	&= \tr{\pvmu}\tr{\wrapParens{\minv{\wrapParens{\chol{\pvsig}}}}} \trchol{\pvsig}\sportwopt\\
  &= \tr{\wrapParens{\minv{\wrapParens{\chol{\pvsig}}}\pvmu}} \gram{\Mtx{Q}} \trchol{\pvsig}\sportwopt 
	= \psnropt \tr{\basev[1]} \Mtx{Q} \trchol{\pvsig}\sportwopt.
\end{align*}
That is the expected return is determined entirely by the first element of
$\Mtx{Q}\trchol{\pvsig}\sportwopt$.
Now note that the volatility of \sportwopt is equal to the Euclidian norm of
that vector:
\begin{align*}
\qform{\pvsig}{\sportwopt} 
	&= \gram{\wrapParens{\trchol{\pvsig}\sportwopt}}\\
	&= \qform{\gram{\Mtx{Q}}}{\wrapParens{\trchol{\pvsig}\sportwopt}}
	= \norm{\Mtx{Q}\trchol{\pvsig}\sportwopt}.
\end{align*}
The rotation also gives a diagonal asymptotic covariance. That is, the
matrix $\wrapParens{1+\kurty\psnrsqopt}\eye +
\wrapParens{2\kurty-1}\psnrsqopt\ogram{\basev[1]}$ is diagonal, 
so we treat the errors in the vector $\Mtx{Q}\trchol{\pvsig}\sportwopt$ as 
asymptotically uncorrelated.

We can use these facts to arrive at an approximate asymptotic distribution
of the \txtSNR of the \txtMP, defined via the function
$$
\pSNR{\pportwopt} \defeq \frac{\trAB{\pportwopt}{\pvmu}}{\sqrt{\qform{\pvsig}{\pportwopt}}}.
$$
So, by the above
$$
\pSNR{\sportwopt} = \psnropt \frac{\trAB{\basev[1]}{\Mtx{Q} \trchol{\pvsig}\sportwopt}}{%
	\norm{\Mtx{Q}\trchol{\pvsig}\sportwopt}}.
$$

Asymptotically we can think of this as
$$
\pSNR{\sportwopt} = \psnropt \frac{\psnropt + \lambda_1
z_1}{\sqrt{\wrapParens{\psnropt + \lambda_1 z_1}^2 + \lambda_p^2
\wrapParens{z_2^2 + \ldots + z_p^2}}},
$$
where 
$$
\lambda_1 = \ssiz^{-\halff}\sqrt{\wrapParens{1+\kurty\psnrsqopt} + \wrapParens{2\kurty - 1}\psnrsqopt},\quad\mbox{and}\quad
\lambda_p = \ssiz^{-\halff}\sqrt{\wrapParens{1+\kurty\psnrsqopt}},
$$
and where the $z_i$ are independent standard normals.

Now consider the Tangent of Arcsine, or ``tas,'' transform defined as $\fntas{x} = x / \sqrt{1-x^2}$. 
\cite{pav2014qbounds}
Applying this transformation to the rescaled \txtSNR, one arrives at
$$
\fntas{\frac{\pSNR{\sportwopt}}{\psnropt}} =
\frac{\psnropt + \lambda_1 z_1}{\lambda_p\sqrt{z_2^2 + \ldots + z_p^2}},
$$
which looks a lot like a non-central $t$ random variable, up to scaling.
So write 
\begin{equation}
\fntas{\frac{\pSNR{\sportwopt}}{\psnropt}} 
= \frac{\lambda_1}{\lambda_p \sqrt{p-1}} \frac{\frac{\psnropt}{\lambda_1} + z_1}{\sqrt{z_2^2 + \ldots + z_p^2}/\sqrt{p-1}} 
= \frac{\lambda_1}{\lambda_p \sqrt{p-1}} t,
\label{eqn:eta_form_psnr_mp}
\end{equation}
where $t$ is a non-central $t$ random variable with $p-1$ degrees of freedom
and non-centrality parameter $\psnropt/\lambda_1$.
See \secref{ellip_sims}, however, for simulations which indicate that an
unreasonably large sample size is required for this approximation to be of any use.

We can perform one more transformation and use the delta method once again
on the map $x \mapsto \sqrt{x-1}$ to 
convert \eqnref{mvclt_isvsm_qbits_elliptic_explicit} into
\begin{equation}
\sqrt{\ssiz}\wrapParens{
	\twobyone{\ssropt}{\Mtx{Q}\trchol{\pvsig}\sportwopt} -
	\twobyone{\psnropt}{\psnropt\basev[1]}}
	\rightsquigarrow \normlaw{0,\Mtx{D_1}},
\label{eqn:mvclt_isvsm_qbits_elliptic_explicit_II}
\end{equation}
where
\begin{equation*}
\Mtx{D_1} = \twobytwo{%
	1 + \frac{2 + 3 \wrapParens{\kurty-1}}{4}\psnropt^2 }{%
	\wrapParens{\wrapParens{1+\psnrsqopt} +
	\frac{3}{2}\wrapParens{\kurty-1}\psnrsqopt}\tr{\basev[1]}}{%
	\wrapParens{2\wrapParens{1+\psnrsqopt} +
	\frac{3}{2}\wrapParens{\kurty-1}\psnrsqopt}\basev[1]}{%
	\wrapParens{1+\kurty\psnrsqopt}\eye
	+ \wrapParens{2\kurty-1}\psnrsqopt\ogram{\basev[1]} }.
\end{equation*}
Note that the variance for \psnropt given here is just Mertens' form of 
the standard error of the \txtSR, given that elliptical distributions
have zero skew and excess kurtosis of $3\wrapParens{\kurty-1}$. 
\cite{mertens2002comments}

We can take this a step further by swapping in $\ssropt$ for $\psnropt$ to arrive
at
\begin{equation}
\sqrt{\ssiz}\wrapParens{\Mtx{Q}\trchol{\pvsig}\sportwopt - \ssropt\basev[1]}
	\rightsquigarrow \normlaw{0,\Mtx{D_2}},
\label{eqn:mvclt_isvsm_qbits_elliptic_explicit_III}
\end{equation}
where
\begin{equation*}
\Mtx{D_2} = 
	\wrapParens{1+\kurty\psnrsqopt}\eye
	- \wrapParens{1 + \frac{2 + \wrapParens{\kurty-1}}{4}\psnrsqopt} \ogram{\basev[1]}.
\end{equation*}

Now perform the same transformation to a $t$ random variable to claim that
\begin{equation}
\fntas{\frac{\pSNR{\sportwopt}}{\psnropt}} 
= \frac{\lambda_1'}{\lambda_p \sqrt{p-1}} t,
\label{eqn:eta_form_psnr_mp_II}
\end{equation}
where $t$ is a non-central $t$ random variable with $p-1$ degrees of freedom
and non-centrality parameter $\ssropt/\lambda_1'$, and 
$$
\lambda_1' = \ssiz^{-\halff}\sqrt{\frac{2 + 3\wrapParens{\kurty-1}}{4}\psnrsqopt},\quad\mbox{and}\quad
\lambda_p = \ssiz^{-\halff}\sqrt{\wrapParens{1+\kurty\psnrsqopt}}.
$$
This suggests another confidence limit, namely take $t$ to be the $\typeI$ quantile
of the non-central $t$ distribution with $p-1$ degrees of freedom
and non-centrality parameter $\ssropt/\lambda_1'$, where $\lambda_1'$ and $\lambda_p$
plug in \ssropt for \psnropt wherever needed, then invert \eqnref{eta_form_psnr_mp_II}
to get a confidence limit on \pSNR{\sportwopt}. This confidence limit is also of
dubious value, however, see \secref{normal_sims}.

Note that the relation in \eqnref{eta_form_psnr_mp} requires one to know
\psnropt. To perform inference on \pSNR{\sportwopt} given the observed data,
we adapt \corollaryref{psnr_minus_ssropt_moments} to the case of elliptical
returns.

\begin{theorem}
\label{theorem:mp_snr_ci_elliptical}
Let \svsm be the unbiased sample estimate of 
\pvsm, based on \ssiz \iid samples of \vreti, assumed to take an
elliptical distribution with kurtosis parameter \kurty.
Let \sportwopt be the sample \txtMP, and \ssropt be the sample
\txtSR of that portfolio. Define the \txtSNR of \sportwopt as
\begin{equation*}
\pSNR{\pportw ; \pvsm, 0} \defeq \frac{\trAB{\pportw}{\pvmu}
}{\sqrt{\qform{\pvsig}{\pportw}}}.
\end{equation*}

Then, asymptotically in \ssiz, the difference $\pSNR{\minv{\svsm} ; \pvsm, 0} - \ssropt$ 
has the following mean and variance:
\begin{align}
\label{eqn:psnr_minus_ssropt_moments_elliptical}
\E{\pSNR{\minv{\svsm} ; \pvsm, 0} - \ssropt} &\rightsquigarrow 
\frac{ \wrapParens{\kurty \psnrsqopt + 1}\wrapParens{1 - \nlatf} 
	+ \frac{\wrapParens{3\kurty-1}}{4}\psnrsqopt + 1
	}{2\ssiz\psnropt}.\\
\VAR{\pSNR{\minv{\svsm} ; \pvsm, 0} - \ssropt} &\rightsquigarrow 
	\frac{\frac{\wrapParens{3\kurty-1}}{4} \psnrsqopt + 1}{\ssiz}.
\end{align}
	And the \emph{ratio} has the asymptotic mean and variance
\begin{align}
\label{eqn:psnr_div_ssropt_moments_elliptical}
\E{\frac{\pSNR{\minv{\svsm} ; \pvsm, 0}}{\ssropt}} &\rightsquigarrow 
1 + \frac{ \wrapParens{\kurty \psnrsqopt + 1}\wrapParens{1 - \nlatf} 
	+ 3\wrapParens{\frac{\wrapParens{3\kurty-1}}{4}\psnrsqopt + 1}
	}{2\ssiz\psnropt^2}.\\
\VAR{\frac{\pSNR{\minv{\svsm} ; \pvsm, 0}}{\ssropt}} &\rightsquigarrow 
	\frac{\frac{\wrapParens{3\kurty-1}}{4} \psnrsqopt + 1}{\ssiz\psnropt^2}.
\end{align}
\end{theorem}
We relegate the long proof to the Appendix. 
Note that a similar line of reasoning
should produce the asymptotic distribution of Hotelling's $T^2$ under 
elliptical returns, which could be compared to form given by Iwashita. \cite{IWASHITA199785}
Also of note is the result of Paulsen and S\"{o}hl \cite{paulsen2016noise}, who
show that for Gaussian returns,
$$
\E{\pSNR{\minv{\svsm} ; \pvsm, 0} - \wrapParens{\ssropt + \frac{1 - \nlatf}{\ssiz \ssropt}}} = 0.
$$

The theorem suggests the following \typeI confidence intervals for 
\pSNR{\minv{\svsm} ; \pvsm, 0}, by plugging in \ssropt for the unknown
quantity \psnropt:
\begin{equation}
\label{eqn:psnr_ci_practical_elliptical}
	\ssropt + 
\frac{ \wrapParens{\kurty \ssrsqopt + 1}\wrapParens{1 - \nlatf} 
	+ c\wrapParens{\frac{\wrapParens{3\kurty-1}}{4}\ssrsqopt  + 1}
	}{2\ssiz\ssropt} \pm Z_{1 - \typeI/2} 
	\sqrt{ \frac{\frac{\wrapParens{3\kurty-1}}{4} \ssrsqopt + 1}{\ssiz} },
\end{equation}
where one can take $c=1$ for the difference formula, and $c=3$ for the ratio
formula.  However, these confidence intervals are \emph{not} well supported
by simulations, in the sense that they require very large \ssiz to 
give near nominal coverage, \cf \secref{normal_sims}.
\subsection{Distribution under matrix normal returns}

Now we consider the case where the (augmented) returns follow a matrix normal distribution.
That is, we suppose that there is a \bby{\ssiz}{\wrapParens{1 + \nlatf}} matrix \pmmu
and symmetric positive semi-definite matrices \pmsigr and \pmsigl,
respectively of size \sbby{\wrapParens{1+\nlatf}} and \sbby{\ssiz} such that
$$
\fvec{\amreti} \sim \normlaw{\fvec{\pmmu},\pmsigr\kron\pmsigl}.
$$
This form allows us to consider deviations from the \iid assumption
by allowing \eg the \pvmu to change over time, autocorrelation in returns
and so on\footnote{We do not consider the elliptical distribution in this case, as it can
impose long term dependence among returns even when they are uncorrelated.}.

We now seek the moments of
$$
\svsm = \oneby{\ssiz}\gram{\amreti}.
$$
\begin{lemma}
\label{lemma:gram_moments}
For matrix normal returns $\fvec{\amreti} \sim \normlaw{\fvec{\pmmu},\pmsigr\kron\pmsigl}$, 
the mean and covariance of the gram are
\begin{equation}
\Eof{\gram{\amreti}} = \gram{\pmmu} + \trace{\pmsigl}\pmsigr,
\end{equation}
and 
\begin{equation}
	\VAR{\fvec{\gram{\amreti}}} = \wrapParens{\eye+\Komm}\wrapBracks{
	\pmsigr\kron\pmsigr\trace{\pmsigl^2}
	+ \wrapParens{\tr{\pmmu}\pmsigl\pmmu} \kron\pmsigr
	+ \pmsigr\kron\wrapParens{\tr{\pmmu}\pmsigl\pmmu}}.
\end{equation}
\end{lemma}

As a check we note this result is consistent with 
\theoremref{theta_covar_elliptical} in the \iid case, which corresponds to
$\pmmu=\vone\tr{\apvmu}$, $\pmsigl=\eye$, $\pmsigr=\apvsig$.

\begin{lemma}
\label{lemma:inv_gram_moments}
Let $\fvec{\amreti} \sim \normlaw{\fvec{\pmmu},\pmsigr\kron\pmsigl}$,
where $\pmsigl$ and $\pmsigr$ are rescaled such that $\trace{\pmsigl}=\ssiz$.
Let
$$
\svsm = \oneby{\ssiz}\gram{\amreti}
$$
Define 
$$
	\pvsm = \lim_{\ssiz\to\infty}\oneby{\ssiz}\gram{\pmmu} + \pmsigr.
$$
Then asymptotically in \ssiz,
\begin{align}
	\sqrt{\ssiz}\wrapParens{\minv{\svsm}-\minv{\pvsm}} 
	&\rightsquigarrow \normlaw{0,\Mtx{B}},
\end{align}
with
{\small
\begin{align*}
	\Mtx{B} &= 2\Simm\wrapBracks{
		\AkronA{\wrapParens{\minv{\pvsm}\pmsigr\minv{\pvsm}}} \trace{\pmsigl^2}}\\
	&\phantom{=}\, 
	+ 2\Simm\wrapBracks{
		\wrapParens{\minv{\pvsm}\tr{\pmmu}\pmsigl\pmmu\minv{\pvsm}} \kron \wrapParens{\minv{\pvsm}\pmsigr\minv{\pvsm}}
	  + \wrapParens{\minv{\pvsm}\pmsigr\minv{\pvsm}}\kron\wrapParens{\minv{\pvsm}\tr{\pmmu}\pmsigl\pmmu\minv{\pvsm}}
		}.
\end{align*}%
}%

In particular, letting $\psnrsqopt = \trbasev[1]\minv{\pvsm}\basev[1] - 1$, and 
$\ssrsqopt=\trbasev[1]\minv{\svsm}\basev[1] - 1$, then
\begin{align}
	\sqrt{\ssiz}\wrapParens{\ssrsqopt - \psnrsqopt} 
	&\rightsquigarrow \normlaw{0,b^2},
\end{align}
for $b^2=\wrapParens{\AkronA{\trbasev[1]}}\Mtx{B}\wrapParens{\AkronA{\basev[1]}}.$
\end{lemma}

The proof follows along the lines of \corollaryref{invtheta_bits_asym_var_elliptical_explicit} and is omitted.
This corollary could be useful in finding the asymptotic distribution of \ssrsqopt (and Hotellings \Tstat{})
under certain divergences from \iid normality. 
For example, one could impose a general autocorrelation by making $\pmsigl[i,j] = \rho^{\abs{i-j}}$ for some $\rho \in \oointerval{-1}{1}$.
One could impose heteroskedasticity structure where $\pmmu=\vect{l}\tr{\pvmu}$ and $\pmsigl=\diag{\vect{l}^\lambda}$ for some scalar lambda.
It has been shown that the \txtSR is relatively robust to deviations from assumptions under similar configurations. \cite{pav_ssc}


\subsection{Likelihood ratio test on \txtMP}

\providecommand{\lrtA}[1][i]{\mathSUB{\Mtx{A}}{#1}}
\providecommand{\lrta}[1][i]{\mathSUB{a}{#1}}


Let us again consider \reti to take a Gaussian distribution, rather than a
general elliptical distribution.
Consider the null hypothesis
\begin{equation}
\Hyp[0]: \trace{\lrtA[i]\minv{\pvsm}} = \lrta[i],\,i=1,\ldots,m.
\label{eqn:lrt_null_back}
\end{equation}
The constraints have to be sensible. 
For example, they cannot
violate the positive definiteness of 
\minv{\pvsm}, symmetry, \etc 
Without loss of generality, we can assume
that the \lrtA[i] are symmetric, since \pvsm is symmetric, and for
symmetric \Mtx{G} and square \Mtx{H}, 
$\trace{\Mtx{G}\Mtx{H}} = \trace{\Mtx{G}\half\wrapParens{\Mtx{H} +
\tr{\Mtx{H}}}}$, and so we could replace any non-symmetric \lrtA[i] with
$\half\wrapParens{\lrtA[i] + \tr{\lrtA[i]}}$.

Employing the Lagrange multiplier technique, the maximum likelihood
estimator under the null hypothesis, call it \pvsm[0], solves the
following equation
\begin{equation}
\begin{split}
0 &= \drbydr{\log\FOOpdf{}{\svsm}{\pvsm}}{\minv{\pvsm}}
- \sum_i \lambda_i
\drbydr{\trace{\lrtA[i]\minv{\pvsm}}}{\minv{\pvsm}},\notag \\
&= - \pvsm[0] + \svsm - \sum_i \lambda_i \lrtA[i],\notag.
\end{split}
\end{equation}
Thus the MLE under the null is
\begin{equation}
\pvsm[0] = \svsm - \sum_i \lambda_i \lrtA[i].
\label{eqn:lrt_mle_soln}
\end{equation}
The maximum likelihood estimator under the constraints has to be
found numerically by solving for the $\lambda_i$, subject
to the constraints in \eqnref{lrt_null_back}.

This framework slightly generalizes Dempster's 
``Covariance Selection,'' \cite{dempster1972} which
reduces to the case where each \lrta[i] is zero, and
each \lrtA[i] is a matrix of all zeros except two (symmetric) ones 
somewhere in the lower right \sbby{\nlatf} sub-matrix. In
all other respects, however, the solution here follows Dempster.

\providecommand{\vitrlam}[2]{\vectUL{\lambda}{\wrapNeParens{#1}}{#2}}
\providecommand{\sitrlam}[2]{\mathUL{\lambda}{\wrapNeParens{#1}}{#2}}
\providecommand{\vitrerr}[2]{\vectUL{\epsilon}{\wrapNeParens{#1}}{#2}}
\providecommand{\sitrerr}[2]{\mathUL{\epsilon}{\wrapNeParens{#1}}{#2}}

An iterative technique for finding the MLE based on a Newton step 
would proceed as follow.  \cite{nocedal2006numerical} 
Let \vitrlam{0}{} be some initial estimate of the
vector of $\lambda_i$. (A good initial estimate can likely be had 
by abusing the asymptotic normality 
result from \subsecref{dist_markoport}.)
The residual of the \kth{k} estimate, \vitrlam{k}{} is
\begin{equation}
\vitrerr{k}{i} \defeq 
\trace{\lrtA[i]\minv{\wrapBracks{\svsm - \sum_j \sitrlam{k}{j} \lrtA[j]}}} - \lrta[i].
\end{equation}
The Jacobian of this residual with respect to the \kth{l} element of \vitrlam{k}
is
\begin{equation}
\begin{split}
\drbydr{\vitrerr{k}{i}}{\sitrlam{k}{l}} &= 
\trace{\lrtA[i]\minv{\wrapBracks{\svsm - \sum_j \sitrlam{k}{j} \lrtA[j]}}
\lrtA[l] \minv{\wrapBracks{\svsm - \sum_j \sitrlam{k}{j} \lrtA[j]}}},\\
&= \tr{\fvec{\lrtA[i]}} \wrapParens{\AkronA{\minv{\wrapBracks{\svsm - \sum_j
\sitrlam{k}{j}\lrtA[j]}}}} \fvec{\lrtA[l]}.
\end{split}
\end{equation}

Newton's method is then the iterative scheme
\begin{equation}
\vitrlam{k+1}{} \leftarrow \vitrlam{k}{} -
\minvParens{\drbydr{\vitrerr{k}{}}{\vitrlam{k}{}}} \vitrerr{k}.
\end{equation}

When (if?) the iterative scheme converges on the optimum, plugging in
\vitrlam{k}{} into \eqnref{lrt_mle_soln} gives the MLE under the null.
The likelihood ratio test statistic is 
\begin{equation}
\begin{split}
-2\log\Lambda &\defeq
-2\log\wrapParens{\frac{\FOOlik{}{\svsm}{\pvsm[0]}}{\FOOlik{}{\svsm}{\pvsm[\mbox{unrestricted }\txtMLE]}}},\\
&= \ssiz\wrapParens{\logdet{\pvsm[0]\minv{\svsm}} + 
\trace{\wrapBracks{\minv{\pvsm[0]} - \minv{\svsm}}\svsm}},\\
&= \ssiz\wrapParens{\logdet{\pvsm[0]\minv{\svsm}} + 
\trace{\minv{\pvsm[0]}\svsm} - \wrapBracks{\nlatf + 1}},
\end{split}
\label{eqn:wilks_lambda_def}
\end{equation}
using the fact that \svsm is the unrestricted MLE, per 
\corollaryref{theta_mle}.
By Wilks' Theorem, under the null 
hypothesis, $-2\log\Lambda$ is, asymptotically in \ssiz, distributed as 
a chi-square with $m$ degrees of freedom. \cite{wilkstheorem1938}
However, most `interesting' null tests posit \pvsm to be somewhere
on the boundary of acceptable values; for such tests, asymptotic
convergence is to some other distribution. \cite{andrews2001testing}


\section{Extensions}

\label{sec:extensions}

For large samples, Wald statistics of the elements of the Markowitz
portfolio computed using the procedure outlined above tend to be 
very similar to the t-statistics produced by the procedure of 
Britten-Jones. \cite{BrittenJones1999} 
However, the technique proposed here admits a number of interesting
extensions.

The script for each of these extensions is the same: 
define, then solve, some portfolio optimization problem; 
show that the solution can be defined in terms of some transformation 
of $\minv{\pvsm}$, giving an implicit recipe for constructing the 
sample portfolio based on the same transformation of $\minv{\svsm}$; 
find the asymptotic distribution of the sample portfolio in terms of \pvvar.

\subsection{Subspace Constraint} 
\label{subsec:subspace_constraint}

Consider the \emph{constrained} portfolio optimization problem
\begin{equation}
\max_{\substack{\pportw : \zerJc \pportw = \vzero,\\
\qform{\pvsig}{\pportw} \le \Rbuj^2}}
\frac{\trAB{\pportw}{\pvmu} - \rfr}{\sqrt{\qform{\pvsig}{\pportw}}},
\label{eqn:portopt_zer}
\end{equation}
where $\zerJc$ is a \bby{\wrapParens{\nlatf - \nlatfzer}}{\nlatf}
matrix of rank $\nlatf - \nlatfzer$, 
\rfr is the disastrous rate, and $\Rbuj > 0$ is the
risk budget. Let the rows of \zerJ span the null space of the rows of
\zerJc; that is, $\zerJc \tr{\zerJ} = \mzero$, and $\ogram{\zerJ} = \eye$.
We can interpret the orthogonality constraint $\zerJc \pportw = \vzero$ as
stating that \pportw must be a linear combination of the columns of
\tr{\zerJ}, thus $\pportw = \trAB{\zerJ}{\pportx}$. The columns of
\tr{\zerJ} may be considered `baskets' of assets to which our investments
are restricted.

We can rewrite the portfolio optimization problem in terms of solving
for \pportx, but then find the asymptotic distribution of the resultant
\pportw. Note that the expected return and covariance of the portfolio
\pportx are, respectively, \trAB{\pportx}{\zerJ\pvmu} and 
\qform{\qoform{\pvsig}{\zerJ}}{\pportx}. Thus we can plug in
$\zerJ\pvmu$ and \qoform{\pvsig}{\zerJ} into 
\lemmaref{sr_optimal_portfolio} to get the following analogue.

\begin{lemma}[subspace constrained \txtSR optimal portfolio]
\label{lemma:subsp_cons_sr_optimal_portfolio}
Assuming the rows of \zerJ span the null space of the rows of \zerJc,
$\zerJ\pvmu \ne \vzero$, and \pvsig is invertible,
the portfolio optimization problem
\begin{equation}
\max_{\substack{\pportw : \zerJc \pportw = \vzero,\\
\qform{\pvsig}{\pportw} \le \Rbuj^2}} 
\frac{\trAB{\pportw}{\pvmu} - \rfr}{\sqrt{\qform{\pvsig}{\pportw}}},
\label{eqn:portopt_zer_I}
\end{equation}
for $\rfr \ge 0, \Rbuj > 0$ is solved by
\begin{equation*}
\begin{split}
\pportwoptFoo{\Rbuj,\zerJ,} 
 &\defeq c \wrapProj{\pvsig}{\zerJ}\pvmu,\\
 c &= \frac{\Rbuj}{\sqrt{\qform{\wrapProj{\pvsig}{\zerJ}}{\pvmu}}}.
\end{split}
\end{equation*}
When $\rfr > 0$ the solution is unique.
\end{lemma}

We can easily find the asymptotic distribution of
\sportwoptFoo{\Rbuj,\zerJ,}, the sample analogue of the optimal portfolio
in \lemmaref{subsp_cons_sr_optimal_portfolio}. First define the
subspace second moment.

\begin{definition}
\label{definition:subspace_second_moment}
Let \zerJt be the \bby{\wrapParens{1+\nlatfzer}}{\wrapParens{\nlatf+1}}
matrix,
$$
\zerJt \defeq \twobytwossym{1}{0}{\zerJ}.
$$
\end{definition}

Simple algebra proves the following lemma.

\begin{lemma}
The elements of $\wrapProj{\pvsm}{\zerJt}$
are 
\begin{equation*}
\wrapProj{\pvsm}{\zerJt}
=
\twobytwo{ 1 + \qform{\wrapProj{\pvsig}{\zerJ}}{\pvmu} }{
-\tr{\pvmu}\wrapProj{\pvsig}{\zerJ}}{
-\wrapProj{\pvsig}{\zerJ}\pvmu}{
\wrapProj{\pvsig}{\zerJ}}.
\end{equation*}
In particular, elements $2$ through $\nlatf+1$ of 
$-\fvech{\wrapProj{\pvsm}{\zerJt}}$ are the portfolio
$\sportwoptFoo{\Rbuj,\zerJ,}$
defined in \lemmaref{subsp_cons_sr_optimal_portfolio}, up to the scaling
constant $c$ which is the ratio of \Rbuj to the square root of the first
element of $\fvech{\wrapProj{\pvsm}{\zerJt}}$ minus one.
\end{lemma}

The asymptotic distribution of $\fvech{\wrapProj{\pvsm}{\zerJt}}$
is given by the following theorem, which is the analogue 
of \theoremref{inv_distribution}.

\begin{theorem}
\label{theorem:subzer_inv_distribution}
Let \svsm be the unbiased sample estimate of 
\pvsm, based on \ssiz \iid samples of \vreti. 
Let \zerJt 
be defined as in \definitionref{subspace_second_moment}.
Let \pvvar be the variance of $\fvech{\ogram{\avreti}}$.  
Then, asymptotically in \ssiz, 
\begin{equation}
\sqrt{\ssiz}\wrapParens{\fvech{\wrapProj{\svsm}{\zerJt}} -
\fvech{\wrapProj{\pvsm}{\zerJt}}}
\rightsquigarrow 
\normlaw{0,\qoform{\pvvar}{\Mtx{H}}},
\label{eqn:mvclt_zer_isvsm}
\end{equation}
where
\begin{equation*}
\Mtx{H} = - \EXD{\wrapParens{\AkronA{\tr{\zerJt}}} 
\wrapParens{\AkronA{\minvParens{\qoform{\pvsm}{\zerJt}}}} 
\wrapParens{\AkronA{\zerJt}}}.
\end{equation*}
\end{theorem}
\begin{proof}
By the multivariate delta method, it suffices to prove that
$$
\Mtx{H} = \dbyd{\fvech{\wrapProj{\svsm}{\zerJt}}}{\fvech{\pvsm}}.
$$
Via \lemmaref{misc_derivs}, it suffices to prove that
$$
\dbyd{\wrapProj{\pvsm}{\zerJt}}{\pvsm}
= 
- \wrapParens{\AkronA{\tr{\zerJt}}} \wrapParens{\AkronA{\minvParens{\qoform{\pvsm}{\zerJt}}}}
\wrapParens{\AkronA{\zerJt}}.
$$

A well-known fact regarding matrix manipulation \cite{magnus1999matrix} is
$$
\fvec{\Mtx{A}\Mtx{B}\Mtx{C}} = \wrapParens{\tr{\Mtx{C}}\kron\Mtx{A}}
\fvec{\Mtx{B}},\quad\mbox{therefore,}\quad
\dbyd{\Mtx{A}\Mtx{B}\Mtx{C}}{\Mtx{B}} = \tr{\Mtx{C}}\kron\Mtx{A}.
$$
Using this, and the chain rule, we have:
\begin{equation*}
\begin{split}
\dbyd{\wrapProj{\pvsm}{\zerJt}}{\pvsm}
&=
\dbyd{\wrapProj{\pvsm}{\zerJt}}{\minvParens{\qoform{\pvsm}{\zerJt}}} 
\dbyd{\minvParens{\qoform{\pvsm}{\zerJt}}}{\qoform{\pvsm}{\zerJt}}
\dbyd{\qoform{\pvsm}{\zerJt}}{\pvsm}\\
&= 
\wrapParens{\AkronA{\tr{\zerJt}}} 
\dbyd{\minvParens{\qoform{\pvsm}{\zerJt}}}{\qoform{\pvsm}{\zerJt}}
\wrapParens{\AkronA{\zerJt}}.
\end{split}
\end{equation*}
\lemmaref{deriv_vech_matrix_inverse} gives the middle term, completing the
proof.
\end{proof}
An analogue of \corollaryref{portwoptR_dist} gives the asymptotic 
distribution of $\pportwoptFoo{\Rbuj,\zerJ,}$
defined in \lemmaref{subsp_cons_sr_optimal_portfolio}.

\subsection{Hedging Constraint} 
\label{subsec:hedging_constraint}

Consider, now, the constrained portfolio optimization problem,
\begin{equation}
\max_{\substack{\pportw : \hejG\pvsig \pportw = \vzero,\\
\qform{\pvsig}{\pportw} \le \Rbuj^2}}
\frac{\trAB{\pportw}{\pvmu} - \rfr}{\sqrt{\qform{\pvsig}{\pportw}}},
\label{eqn:portopt_hej}
\end{equation}
where $\hejG$ is now a \bby{\nlatfhej}{\nlatf} matrix of 
rank \nlatfhej.
We can interpret the \hejG constraint as stating that the covariance of 
the returns of a feasible portfolio with the returns of a portfolio 
whose weights are in a given row of \hejG shall equal zero.
In the garden variety application of this problem, \hejG consists of 
\nlatfhej rows of the identity matrix;
in this case, feasible portfolios are `hedged' with respect 
to the \nlatfhej assets selected by \hejG
(although they may hold some position in the hedged assets).

\begin{lemma}[constrained \txtSR optimal portfolio]
\label{lemma:cons_sr_optimal_portfolio}
Assuming $\pvmu \ne \vzero$, and \pvsig is invertible,
the portfolio optimization problem
\begin{equation}
\max_{\substack{\pportw : \hejG\pvsig \pportw = \vzero,\\ 
\qform{\pvsig}{\pportw} \le \Rbuj^2}} 
\frac{\trAB{\pportw}{\pvmu} - \rfr}{\sqrt{\qform{\pvsig}{\pportw}}},
\label{eqn:cons_port_prob}
\end{equation}
for $\rfr \ge 0, \Rbuj > 0$ is solved by
\begin{equation*}
\begin{split}
\pportwoptFoo{\Rbuj,\hejG,} &\defeq c \wrapParens{\minv{\pvsig}{\pvmu} -
	\wrapProj{\pvsig}{\hejG}\pvmu},\\
 c &= \frac{\Rbuj}{\sqrt{\qiform{\pvsig}{\pvmu} - 
	\qform{\wrapProj{\pvsig}{\hejG}}{\pvmu}}}.
\end{split}
\end{equation*}
When $\rfr > 0$ the solution is unique.
\end{lemma}
\begin{proof}
By the Lagrange multiplier technique, the optimal portfolio
solves the following equations:
\begin{equation*}
\begin{split}
0 &= c_1 \pvmu - c_2 \pvsig \pportw - \gamma_1 \pvsig \pportw -
\pvsig\trAB{\hejG}{\vect{\gamma_2}},\\
\qform{\pvsig}{\pportw} &\le \Rbuj^2,\\
\hejG\pvsig \pportw &= \vzero,
\end{split}
\end{equation*}
where $\gamma_i$ are Lagrange multipliers, and $c_1, c_2$ are 
scalar constants.

Solving the first equation gives
$$
\pportw = c_3\wrapBracks{\minvAB{\pvsig}{\pvmu} -
\trAB{\hejG}{\vect{\gamma_2}}}.
$$
Reconciling this with the hedging equation we have
$$
\vzero 
= \hejG\pvsig \pportw 
= c_3 \hejG\pvsig \wrapBracks{\minvAB{\pvsig}{\pvmu} -
\trAB{\hejG}{\vect{\gamma_2}}},
$$ 
and therefore 
$\vect{\gamma_2} = \minvAB{\wrapParens{\qoform{\pvsig}{\hejG}}}{\hejG}\pvmu.$
Thus
$$
\pportw = c_3\wrapBracks{\minvAB{\pvsig}{\pvmu} -
\wrapProj{\pvsig}{\hejG}\pvmu}.
$$
Plugging this into the objective reduces the problem to the
univariate optimization
\begin{equation*}
\max_{c_3 :\, c_3^2 \le \fracc{\Rbuj^2}{\psnrsqoptG{\hejG}}} 
\sign{c_3} \psnroptG{\hejG} - \frac{\rfr}{\abs{c_3}\psnroptG{\hejG}},
\end{equation*}
where $\psnrsqoptG{\hejG} = \qiform{\pvsig}{\pvmu} - 
  \qform{\wrapProj{\pvsig}{\hejG}}{\pvmu}.$
The optimum
occurs for $c = \fracc{\Rbuj}{\psnroptG{\hejG}}$, moreover the optimum
is unique when $\rfr > 0$.
\end{proof}

The optimal hedged
portfolio in \lemmaref{cons_sr_optimal_portfolio} is, up to scaling,
the difference of the unconstrained optimal portfolio
from \lemmaref{sr_optimal_portfolio} and the subspace constrained
portfolio in \lemmaref{subsp_cons_sr_optimal_portfolio}. This `delta'
analogy continues for the rest of this section.

\begin{definition}[Delta Inverse Second Moment]
\label{definition:delta_inv_second_moment}
Let \hejGt be the \bby{\wrapParens{1+\nlatfhej}}{\wrapParens{\nlatf+1}}
matrix,
$$
\hejGt \defeq \twobytwossym{1}{0}{\hejG}.
$$
Define the `delta inverse second moment' as
$$
\Delhej\minv{\pvsm} \defeq \minv{\pvsm} -
\wrapProj{\pvsm}{\hejGt}.$$
\end{definition}

Simple algebra proves the following lemma.

\begin{lemma}
The elements of $\Delhej\minv{\pvsm}$
are 
\begin{equation*}
\Delhej\minv{\pvsm} =
\twobytwo{ \qiform{\pvsig}{\pvmu} - \qform{\wrapProj{\pvsig}{\hejG}}{\pvmu} }{
-\tr{\pvmu}\minv{\pvsig} + \tr{\pvmu}\wrapProj{\pvsig}{\hejG}}{
-\minv{\pvsig}\pvmu + \wrapProj{\pvsig}{\hejG}\pvmu}{
\minv{\pvsig} - \wrapProj{\pvsig}{\hejG}}.
\end{equation*}
In particular, elements $2$ through $\nlatf+1$ of 
$-\fvech{\Delhej\minv{\pvsm}}$ are the portfolio $\pportwoptFoo{\Rbuj,\hejG,}$
defined in \lemmaref{cons_sr_optimal_portfolio}, up to the scaling
constant $c$ which is the ratio of \Rbuj to the square root of the first
element of \fvech{\Delhej\minv{\pvsm}}.
\end{lemma}

The statistic 
$\qiform{\svsig}{\svmu} - \qform{\wrapProj{\svsig}{\hejG}}{\svmu}$, for the
case where \hejG is some rows of the \sbby{\nlatf} identity matrix,
was first proposed by Rao, and its distribution under Gaussian returns was 
later found by Giri. \cite{rao1952,giri1964likelihood} This test statistic
may be used for tests of portfolio \emph{spanning} for the case where a
risk-free instrument is traded.  \cite{HKspan1987,KanZhou2012}

The asymptotic distribution of $\Delhej\minv{\svsm}$ is given by the
following theorem, which is the analogue of \theoremref{inv_distribution}.

\begin{theorem}
Let \svsm be the unbiased sample estimate of 
\pvsm, based on \ssiz \iid samples of \vreti. 
Let $\Delhej\minv{\pvsm}$
be defined as in \definitionref{delta_inv_second_moment}, and similarly
define $\Delhej\minv{\svsm}$.  Let \pvvar be the variance of 
$\fvech{\ogram{\avreti}}$.  Then, asymptotically in \ssiz, 
\begin{equation}
\sqrt{\ssiz}\wrapParens{\fvech{\Delhej\minv{\svsm}} - \fvech{\Delhej\minv{\pvsm}}} 
\rightsquigarrow 
\normlaw{0,\qoform{\pvvar}{\Mtx{H}}},
\label{eqn:mvclt_hej_isvsm}
\end{equation}
where
\begin{equation*}
\Mtx{H} = - \EXD{\wrapBracks{\AkronA{\minv{\pvsm}} - 
\wrapParens{\AkronA{\tr{\hejGt}}} \wrapParens{\AkronA{\minvParens{\qoform{\pvsm}{\hejGt}}}}
\wrapParens{\AkronA{\hejGt}}}}.
\end{equation*}
\label{theorem:delhej_inv_distribution}
\end{theorem}

\begin{proof}
Minor modification of proof of \theoremref{subzer_inv_distribution}.
\end{proof}
\begin{caution}
In the hedged portfolio optimization problem considered here, 
the optimal portfolio will, in general, hold money in the 
row space of \hejG. 
For example, in 
the garden variety application, where one is hedging out exposure to 
`the market' by including a broad market ETF, and taking \hejG to be 
the corresponding row of the identity matrix, the final portfolio may 
hold some position in that broad market ETF. This is fine for an ETF,
but one may wish to hedge out exposure to an untradeable returns 
stream--the returns of an index, say.
Combining the hedging constraint of this section with the
subspace constraint of \subsecref{subspace_constraint} is
simple in the case where the rows of \hejG are spanned
by the rows of \zerJ. The more general case, however, is rather
more complicated.
\end{caution}

\subsection{Conditional Heteroskedasticity}
\label{subsec:cond_het}

The methods described above ignore `volatility clustering', and assume
homoskedasticity. \cite{stylized_facts,nelson1991,ARCH1987} 
To deal with this,
consider a strictly positive scalar random variable, \fvola[i], 
observable at the time the investment decision is
required to capture \vreti[i+1]. For reasons to be obvious later, it 
is more convenient to think of \fvola[i] as a `quietude' indicator,
or a `weight' for a weighted regression.

Two simple competing models for conditional heteroskedasticity are
\begin{align}
\label{eqn:cond_model_I}
\mbox{(constant):}\quad \Econd{\vreti[i+1]}{\fvola[i]} &=
\fvola[i]^{-1}\pvmu, & \Varcond{\vreti[i+1]}{\fvola[i]} &=
\fvola[i]^{-2} \pvsig,\\
\label{eqn:cond_model_II}
\mbox{(floating):}\quad \Econd{\vreti[i+1]}{\fvola[i]} &=
\pvmu, & \Varcond{\vreti[i+1]}{\fvola[i]} &=
\fvola[i]^{-2} \pvsig.
\end{align}
Under the model in \eqnref{cond_model_I}, the maximal \txtSR
is $\sqrt{\qiform{\pvsig}{\pvmu}}$, independent of \fvola[i];
under \eqnref{cond_model_II}, it is 
is $\fvola[i]\sqrt{\qiform{\pvsig}{\pvmu}}$. 
The model names reflect whether or not
the maximal \txtSR varies conditional on \fvola[i].

The optimal portfolio under both models is the same, as stated
in the following lemma, the proof of which follows by simply using
\lemmaref{sr_optimal_portfolio}.

\begin{lemma}[Conditional \txtSR optimal portfolio]
\label{lemma:cond_sr_optimal_portfolio_I}
Under either the model in \eqnref{cond_model_I} or \eqnref{cond_model_II}, 
conditional on observing \fvola[i], the portfolio optimization problem
\begin{equation}
\argmax_{\pportw :\, \Varcond{\trAB{\pportw}{\vreti[i+1]}}{\fvola[i]} \le \Rbuj^2} 
\frac{\Econd{\trAB{\pportw}{\vreti[i+1]}}{\fvola[i]} -
\rfr}{\sqrt{\Varcond{\trAB{\pportw}{\vreti[i+1]}}{\fvola[i]}}},
\label{eqn:cond_sr_optimal_portfolio_problem_I}
\end{equation}
for $\rfr \ge 0, \Rbuj > 0$ is solved by
\begin{equation}
\pportwopt = \frac{\fvola[i] \Rbuj}{\sqrt{\qiform{\pvsig}{\pvmu}}}
\minvAB{\pvsig}{\pvmu}.
\end{equation}
Moreover, this is the unique solution whenever $\rfr > 0$.
\end{lemma}

To perform inference on the portfolio \pportwopt from 
\lemmaref{cond_sr_optimal_portfolio_I}, under the `constant'
model of \eqnref{cond_model_I}, apply the unconditional techniques
to the sample second moment of $\fvola[i]\avreti[i+1]$. 

For
the `floating' model of \eqnref{cond_model_II}, however, some 
adjustment to the technique is required.
Define $\aavreti[i+1] \defeq \fvola[i]\avreti[i+1]$; that is,
$\aavreti[i+1] = \asvec{\fvola[i],\fvola[i]\tr{\vreti[i+1]}}$. 
Consider the second moment of \aavreti:
\begin{equation}
\pvsm[{\fvola}] \defeq \E{\ogram{\aavreti}} = 
	\twobytwo{\volavar}{\volavar\tr{\pvmu}}{\volavar\pvmu}{\pvsig +
\qoform{\volavar}{\pvmu}},\quad\mbox{where}\quad
\volavar \defeq \E{\fvola^2}.
\end{equation}
The inverse of \pvsm[{\fvola}] is
\begin{equation}
\minv{\pvsm[{\fvola}]} 
= \twobytwo{\volaivar + \qiform{\pvsig}{\pvmu}}{-\tr{\pvmu}\minv{\pvsig}}{-\minv{\pvsig}\pvmu}{\minv{\pvsig}}
\label{eqn:new_trick_inversion}
\end{equation}
Once again, the optimal portfolio (up to scaling and sign), appears in
\fvech{\minv{\pvsm[{\fvola}]}}.
Similarly, define the sample analogue:
\begin{equation}
\svsm[{\fvola}] \defeq \oneby{\ssiz}\sum_i \ogram{\aavreti[i+1]}.
\end{equation}
We can find the asymptotic distribution of \fvech{\svsm[{\fvola}]}
using the same techniques as in the unconditional case, as in the
following analogue of \theoremref{inv_distribution}:

\begin{theorem}
\label{theorem:cond_inv_distribution}
Let $\svsm[{\fvola}] \defeq \oneby{\ssiz}\sum_i \ogram{\aavreti[i+1]}$,
based on \ssiz \iid samples of \asvec{\fvola,\tr{\vreti}}.
Let \pvvar be the variance of $\fvech{\ogram{\aavreti}}$.
Then, asymptotically in \ssiz, 
\begin{equation}
\label{eqn:cond_mvclt_isvsm}
\sqrt{\ssiz}\wrapParens{\fvech{\minv{\svsm[{\fvola}]}} -
\fvech{\minv{\pvsm[{\fvola}]}}} 
\rightsquigarrow 
\normlaw{0,\qoform{\pvvar}{\Mtx{H}}},
\end{equation}
where
\begin{equation}
\Mtx{H} = -\EXD{\wrapParens{\AkronA{\minv{\pvsm[{\fvola}]}}}}.
\end{equation}
Furthermore, we may replace \pvvar in this equation with an asymptotically
consistent estimator, \svvar.
\end{theorem}

The only real difference from the unconditional case is that we cannot
automatically assume that the first row and column of \pvvar is zero
(unless \fvola is actually constant, which misses the point). Moreover,
the shortcut for estimating \pvvar under Gaussian returns is not 
valid without some patching, an exercise left for the reader.

Dependence or independence of maximal \txtSR from 
volatility is an assumption which, ideally, one could test with
data. A mixed model containing both characteristics can be written
as follows:
\begin{align}
\label{eqn:cond_model_III}
\mbox{(mixed):}\quad \Econd{\vreti[i+1]}{\fvola[i]} &=
\fvola[i]^{-1}\pvmu[0] + \pvmu[1], & \Varcond{\vreti[i+1]}{\fvola[i]} &=
\fvola[i]^{-2} \pvsig.
\end{align}
One could then test whether elements of $\pvmu[0]$ or of $\pvmu[1]$ are
zero. Analyzing this model is somewhat complicated without moving to a
more general framework, as in the sequel.

\subsection{Conditional Expectation and Heteroskedasticity}
\label{subsec:cond_ret_het}

Suppose you observe random variables $\fvola[i] > 0$, and 
\nfac-vector $\vfact[i]$ at some time prior to when the investment
decision is required to capture \vreti[i+1]. It need not be the case
that \fvola[{}] and \vfact[{}] are independent. The general model
is now
\begin{align}
\label{eqn:cond_model_IV}
\mbox{(bi-conditional):}\quad 
\Econd{\vreti[i+1]}{\fvola[i],\vfact[i]} &=
\pRegco \vfact[i], &
\Varcond{\vreti[i+1]}{\fvola[i],\vfact[i]} &=
\fvola[i]^{-2} \pvsig,
\end{align}
where \pRegco is some \bby{\nlatf}{\nfac} matrix. Without the
\fvola[i] term, these are the `predictive regression' equations
commonly used in Tactical Asset 
Allocation.  \cite{connor1997,herold2004TAA,brandt2009portfolio}

By letting
$\vfact[i] = \asvec{\fvola[i]^{-1},1}$ we recover 
the mixed model in \eqnref{cond_model_III}; the bi-conditional model
is considerably more general, however.  The conditionally-optimal portfolio
is given by the following lemma. Once again, the proof proceeds simply
by plugging in the conditional expected return and volatility into
\lemmaref{sr_optimal_portfolio}. 

\begin{lemma}[Conditional \txtSR optimal portfolio]
\label{lemma:cond_sr_optimal_portfolio_II}
Under the model in \eqnref{cond_model_IV},
conditional on observing \fvola[i] and \vfact[i], 
the portfolio optimization problem
\begin{equation}
\argmax_{\pportw :\, \Varcond{\trAB{\pportw}{\vreti[i+1]}}{\fvola[i],\vfact[i]} \le \Rbuj^2} 
\frac{\Econd{\trAB{\pportw}{\vreti[i+1]}}{\fvola[i],\vfact[i]} -
\rfr}{\sqrt{\Varcond{\trAB{\pportw}{\vreti[i+1]}}{\fvola[i],\vfact[i]}}},
\label{eqn:cond_sr_optimal_portfolio_problem_I}
\end{equation}
for $\rfr \ge 0, \Rbuj > 0$ is solved by
\begin{equation}
\pportwopt = \frac{\fvola[i]
\Rbuj}{\sqrt{\qform{\qiform{\pvsig}{\pRegco}}{\vfact[i]}}}
\minvAB{\pvsig}{\pRegco\vfact[i]}.
\end{equation}
Moreover, this is the unique solution whenever $\rfr > 0$.
\end{lemma}

\begin{caution}
It is emphatically \emph{not} the case that 
investing in the portfolio \pportwopt from 
\lemmaref{cond_sr_optimal_portfolio_II} at every time step 
is long-term \txtSR optimal. One may possibly achieve a higher long-term
\txtSR by down-levering at times when the conditional \txtSR is low.
The optimal long term investment strategy falls under the 
rubric of `multiperiod portfolio choice', and is an area of active
research. \cite{mulvey2003advantages,fabozzi2007robust,brandt2009portfolio}
\end{caution}

The matrix \minvAB{\pvsig}{\pRegco} is the generalization of the
Markowitz portfolio: it is the multiplier for a model under which
the optimal portfolio is linear in the features \vfact[i] (up to
scaling to satisfy the risk budget). We can think of this matrix
as the `Markowitz coefficient'. If an entire column of 
\minvAB{\pvsig}{\pRegco} is zero, it suggests that the 
corresponding element of \vfact[{}] can be ignored in investment
decisions; if an entire row of \minvAB{\pvsig}{\pRegco} is zero,
it suggests the corresponding instrument delivers no return or
hedging benefit.

Conditional on observing \vfact[i] and \fvola[i], the maximal
achievable squared \txtSNR is
\begin{equation*}
\acondb{\psnrsqopt}{\fvola[i],\vfact[i]} \defeq
\wrapParens{\frac{\Econd{\trAB{\pportwopt}{\vreti[i+1]}}{\fvola[i],\vfact[i]}}{\sqrt{\Varcond{\trAB{\pportwopt}{\vreti[i+1]}}{\fvola[i],\vfact[i]}}}}^2
= \qform{\qiform{\pvsig}{\pRegco}}{\vfact[i]}.
\end{equation*}
This is independent of \fvola[i], but depends on \vfact[i].
The unconditional expected value of the maximal squared
\txtSNR is thus
\begin{equation*}
\begin{split}
\Eof[{\vfact[{}]}]{\psnrsqopt} 
&\defeq \Eof[{\vfact[{}]}]{\trace{\qform{\qiform{\pvsig}{\pRegco}}{\vfact}}},\\
&= \Eof[{\vfact[{}]}]{\trace{\wrapParens{\qiform{\pvsig}{\pRegco}}\ogram{\vfact}}},\\
&= \trace{\wrapParens{\qiform{\pvsig}{\pRegco}}\Eof[{\vfact[{}]}]{\ogram{\vfact}}},\\
&= \trace{\wrapParens{\qiform{\pvsig}{\pRegco}}\pfacsig}.
\end{split}
\end{equation*}
This quantity is the \emph{Hotelling-Lawley trace}, typically used
to test the so-called Multivariate General 
Linear Hypothesis.  \cite{Rencher2002,Muller1984143}
See \secref{MGLH}.


To perform inference on the Markowitz coefficient, we can proceed
exactly as above. Let
\begin{equation}
\aavreti[i+1] \defeq \asvec{\fvola[i]\tr{\vfact[i]},\fvola[i]\tr{\vreti[i+1]}}.
\end{equation}
Consider the second moment of \aavreti:
\begin{equation}
\pvsm[\sfactsym] \defeq \E{\ogram{\aavreti}} = 
	\twobytwo{\pfacsig}{\pfacsig\tr{\pRegco}}{\pRegco\pfacsig}{\pvsig +
\qoform{\pfacsig}{\pRegco}},\quad\mbox{where}\quad
\pfacsig \defeq \E{\fvola[{}]^2\ogram{\vfact[{}]}}.
\label{eqn:cond_pvsm_def}
\end{equation}
The inverse of \pvsm[\sfactsym] is
\begin{equation}
\minv{\pvsm[\sfactsym]}
= \twobytwo{\minv{\pfacsig} +
\qiform{\pvsig}{\pRegco}}{-\tr{\pRegco}\minv{\pvsig}}{-\minv{\pvsig}\pRegco}{\minv{\pvsig}}
\label{eqn:new_new_trick_inversion}
\end{equation}
Once again, the Markowitz coefficient (up to scaling and sign), appears in
\fvech{\minv{\pvsm[\sfactsym]}}.

The following theorem is an analogue of, and shares a proof with,
\theoremref{inv_distribution}.

\begin{theorem}
\label{theorem:cond_inv_distribution_II}
Let $\svsm[\sfactsym] \defeq \oneby{\ssiz}\sum_i \ogram{\aavreti[i+1]}$,
based on \ssiz \iid samples of \asvec{\fvola,\tr{\vfact[{}]},\tr{\vreti}},
where 
$$
\aavreti[i+1] \defeq \asvec{\fvola[i]\tr{\vfact[i]},\fvola[i]\tr{\vreti[i+1]}}.
$$
Let \pvvar be the variance of $\fvech{\ogram{\aavreti}}$.
Then, asymptotically in \ssiz, 
\begin{equation}
\sqrt{\ssiz}\wrapParens{\fvech{\minv{\svsm[\sfactsym]}} -
\fvech{\minv{\pvsm[\sfactsym]}}} 
\rightsquigarrow 
\normlaw{0,\qoform{\pvvar}{\Mtx{H}}},
\label{eqn:cond_mvclt_isvsm_II}
\end{equation}
where
\begin{equation}
\Mtx{H} = -\EXD{\wrapParens{\AkronA{\minv{\pvsm[\sfactsym]}}}}.
\end{equation}
Furthermore, we may replace \pvvar in this equation with an asymptotically
consistent estimator, \svvar.
\end{theorem}
\subsection{Conditional Expectation and Heteroskedasticity with Hedging Constraint}

A little work allows us to combine the conditional model of
\subsecref{cond_ret_het} with the hedging constraint of 
\subsecref{hedging_constraint}. Suppose returns follow the model
of \eqnref{cond_model_IV}. To prove the following lemma, simply 
plug in $\pRegco\vfact[i]$ for \pvmu, and 
$\fvola[i]^{-2} \pvsig$ for \pvsig into \lemmaref{cons_sr_optimal_portfolio}.

\begin{lemma}[Hedged Conditional \txtSR optimal portfolio]
\label{lemma:hej_cond_sr_optimal_portfolio}
Let $\hejG$ be a given \bby{\nlatfhej}{\nlatf} matrix of 
rank \nlatfhej. Under the model in \eqnref{cond_model_IV},
conditional on observing \fvola[i] and \vfact[i], 
the portfolio optimization problem
\begin{equation}
\argmax_{\substack{\pportw :\,\hejG\pvsig \pportw = \vzero,\\
 \Varcond{\trAB{\pportw}{\vreti[i+1]}}{\fvola[i],\vfact[i]} \le \Rbuj^2}}
\frac{\Econd{\trAB{\pportw}{\vreti[i+1]}}{\fvola[i],\vfact[i]} -
\rfr}{\sqrt{\Varcond{\trAB{\pportw}{\vreti[i+1]}}{\fvola[i],\vfact[i]}}},
\label{eqn:hej_cond_sr_optimal_portfolio_problem_I}
\end{equation}
for $\rfr \ge 0, \Rbuj > 0$ is solved by
\begin{equation*}
\begin{split}
\pportwoptFoo{\Rbuj,\hejG,} &\defeq c \wrapParens{\minv{\pvsig}{\pRegco} -
	\wrapProj{\pvsig}{\hejG}\pRegco}\vfact[i],\\
 c &= \frac{\fvola[i]\Rbuj}{\sqrt{\qiform{\pvsig}{\wrapParens{\pRegco\vfact[i]}} - 
	\qform{\wrapProj{\pvsig}{\hejG}}{\wrapParens{\pRegco\vfact[i]}}}}.
\end{split}
\end{equation*}
Moreover, this is the unique solution whenever $\rfr > 0$.
\end{lemma}
The same cautions regarding multiperiod portfolio choice apply to the
above lemma. Results analogous to those of \subsecref{hedging_constraint} 
follow, but with one minor modification to the analogue of 
\definitionref{delta_inv_second_moment}.

\begin{lemma}
Let \hejGt now be the 
\bby{\wrapParens{\nfac+\nlatfhej}}{\wrapParens{\nfac+\nlatf}} matrix,
$$
\hejGt \defeq \twobytwossym{\eye_{\nfac}}{0}{\hejG},
$$
where the upper right corner is the \sbby{\nfac} identity matrix.
Define the `delta inverse second moment' as
\begin{equation}
\label{eqn:cond_del_hej_def}
\Delhej\minv{\pvsm[\sfactsym]} \defeq \minv{\pvsm[\sfactsym]} -
\wrapProj{\pvsm[\sfactsym]}{\hejGt},
\end{equation}
where $\pvsm[\sfactsym]$ is defined as in \eqnref{cond_pvsm_def}.
The elements of $\Delhej\minv{\pvsm[\sfactsym]}$
are  
\begin{equation*}
\Delhej\minv{\pvsm} =
\twobytwo{ \qiform{\pvsig}{\pRegco} - \qform{\wrapProj{\pvsig}{\hejG}}{\pRegco} }{
-\tr{\pRegco}\minv{\pvsig} + \tr{\pRegco}\wrapProj{\pvsig}{\hejG}}{
-\minv{\pvsig}\pRegco + \wrapProj{\pvsig}{\hejG}\pRegco}{
\minv{\pvsig} - \wrapProj{\pvsig}{\hejG}}.
\end{equation*}
In particular, the Markowitz coefficient from
\lemmaref{hej_cond_sr_optimal_portfolio} appears in the lower left
corner of 
$-\fvech{\Delhej\minv{\pvsm[\sfactsym]}}$, and the
denominator of the constant $c$ from 
\lemmaref{hej_cond_sr_optimal_portfolio} depends on a quadratic form
of \vfact[i] with the upper right left corner of
$\fvech{\Delhej\minv{\pvsm[\sfactsym]}}$.
\end{lemma}

\begin{theorem}
\label{theorem:hej_cond_inv_distribution_II}
Let $\svsm[\sfactsym] \defeq \oneby{\ssiz}\sum_i \ogram{\aavreti[i+1]}$,
based on \ssiz \iid samples of \asvec{\fvola,\tr{\vfact[{}]},\tr{\vreti}},
where 
$$
\aavreti[i+1] \defeq \asvec{\fvola[i]\tr{\vfact[i]},\fvola[i]\tr{\vreti[i+1]}}.
$$
Let \pvvar be the variance of $\fvech{\ogram{\aavreti}}$.
Define 
$\Delhej\minv{\pvsm[\sfactsym]}$ as in 
\eqnref{cond_del_hej_def} for given
\bby{\wrapParens{\nfac+\nlatfhej}}{\wrapParens{\nlatf+\nfac}} matrix \hejGt.

Then, asymptotically in \ssiz, 
\begin{equation}
\sqrt{\ssiz}\wrapParens{\fvech{\Delhej\minv{\svsm[\sfactsym]}} -
\fvech{\Delhej\minv{\pvsm[\sfactsym]}}} 
\rightsquigarrow 
\normlaw{0,\qoform{\pvvar}{\Mtx{H}}},
\label{eqn:hej_cond_mvclt_isvsm_II}
\end{equation}
where
\begin{equation*}
\Mtx{H} = - \EXD{\wrapBracks{\AkronA{\minv{\pvsm[\sfactsym]}} - 
\wrapParens{\AkronA{\tr{\hejGt}}}
\wrapParens{\AkronA{\minvParens{\qoform{\pvsm[\sfactsym]}{\hejGt}}}}
\wrapParens{\AkronA{\hejGt}}}}.
\end{equation*}
Furthermore, we may replace \pvvar in this equation with an asymptotically
consistent estimator, \svvar.
\end{theorem}

\subsection{Conditional Expectation and Multivariate Heteroskedasticity}
\label{subsec:cond_ret_gen_het}

Here we extend the model from \subsecref{cond_ret_het} to accept
multiple heteroskedasticity `features'. 
First note that if we redefined \vfact[i] to be $\vfact[i]\fvola[i]$,
we could rewrite the model in \eqnref{cond_model_IV}
as
\begin{align*}
\Econd{\vreti[i+1]\fvola[i]}{\fvola[i],\vfact[i]} &=
\pRegco \vfact[i], &
\Varcond{\vreti[i+1]\fvola[i]}{\fvola[i],\vfact[i]} &=
\pvsig,
\end{align*}
This can be generalized to vector-valued \fvvola[] by means of the
\emph{flattening trick}. \cite{JOFI:JOFI1055}

Suppose you observe the state variables \nvol-vector $\fvvola[i]$, 
and \nfac-vector $\vfact[i]$ at some time prior to when the investment
decision is required to capture \vreti[i+1]. It need not be the case
that \fvvola[{}] and \vfact[{}] are independent. For sane 
interpretation of the model, it makes sense that all elements of
\fvvola[] are positive.
The general model is now
\begin{align}
\label{eqn:cond_model_V}
\Econd{\fvec{\vreti[i+1]\tr{\fvvola[i]}}}{\fvvola[i],\vfact[i]} &=
\pRegco \vfact[i], &
\Varcond{\fvec{\vreti[i+1]\tr{\fvvola[i]}}}{\fvvola[i],\vfact[i]} &=
\pvsig,
\end{align}
where \pRegco is some \bby{\wrapParens{\nlatf\nvol}}{\nfac} matrix, 
and \pvsig is now a \bby{\wrapParens{\nlatf\nvol}}{\wrapParens{\nlatf\nvol}}
matrix.

Conditional on observing \fvvola[i], a portfolio on the \nfac
assets, \sportw, can be expressed as the portfolio 
\fvec{\sportw\trminv{\fvvola[i]}} on the $\wrapParens{\nlatf\nvol}$ assets
whose returns are the vector \fvec{\vreti[i+1]\tr{\fvvola[i]}}; here
$\trminv{\fvvola[i]}$ refers to the element-wise, or Hadamard, inverse
of $\tr{\fvvola[i]}$.
Thus we may perform portfolio conditional optimization on the
enlarged space of $\wrapParens{\nlatf\nvol}$ assets, 
and then, conditional on \fvvola[i], impose a subspace constraint 
requiring the portfolio to be spanned by the column space of
$\tr{\wrapParens{\eye[\nlatf]\krov\fvvola[i]}}$, where $\krov$ is
used to mean the Kronecker product with $\minv{\fvvola[i]}$, the
Hadamard inverse of \fvvola[i].

We can then combine the results of \subsecref{subspace_constraint}
and \subsecref{cond_ret_het} to solve the portfolio optimization
problem, and perform inference on that portfolio. The following
is the analogue of \lemmaref{cond_sr_optimal_portfolio_II} combined
with \lemmaref{subsp_cons_sr_optimal_portfolio}.

\begin{lemma}[Conditional \txtSR optimal portfolio]
\label{lemma:cond_sr_optimal_portfolio_III}
Suppose returns follow the 
model in \eqnref{cond_model_V}, and suppose
\fvvola[i] and \vfact[i] have been observed.
Let
$\zerJ = \tr{\wrapParens{\eye[\nlatf]\krov\fvvola[i]}},$ and
suppose $\zerJ\pRegco\vfact[i]$ is not all zeros. 
Then the portfolio optimization problem
\begin{equation}
\argmax_{\pportw :\, \Varcond{\trAB{\pportw}{\vreti[i+1]}}{\fvvola[i],\vfact[i]} \le \Rbuj^2} 
\frac{\Econd{\trAB{\pportw}{\vreti[i+1]}}{\fvvola[i],\vfact[i]} -
\rfr}{\sqrt{\Varcond{\trAB{\pportw}{\vreti[i+1]}}{\fvvola[i],\vfact[i]}}},
\label{eqn:cond_sr_optimal_portfolio_problem_III}
\end{equation}
for $\rfr \ge 0, \Rbuj > 0$ is solved by
\begin{equation*}
\begin{split}
\pportwopt &= c \wrapProj{\pvsig}{\zerJ}\pRegco\vfact[i],\\
 c &=
\frac{\Rbuj}{\sqrt{\qform{\wrapProj{\pvsig}{\zerJ}}{\wrapParens{\pRegco\vfact[i]}}}}.
\end{split}
\end{equation*}
Moreover, the solution is unique whenever $\rfr > 0$. This portfolio
achieves the maximal objective of
\begin{equation*}
\sqrt{\qform{\wrapProj{\pvsig}{\zerJ}}{\wrapParens{\pRegco\vfact[i]}}} -
\frac{\rfr}{\Rbuj}.
\end{equation*}
\end{lemma}

The distribution of the sample analogue of the portfolio in 
\lemmaref{cond_sr_optimal_portfolio_III} is given essentially
by \theoremref{subzer_inv_distribution}, applied to the case of
conditional expected returns.

\renewcommand{\MGLHT}[1][]{\MtxUL{T}{}{#1}}
\providecommand{\mglhL}[0]{\MATHIT{\Mtx{G}_1}}
\providecommand{\mglhR}[0]{\MATHIT{\Mtx{G}_2}}

\section{Constrained Estimation}

\providecommand{\cnB}{\MATHIT{\mathfrak{D}}}
\providecommand{\cnW}{\MATHIT{\mathfrak{W}}}
\providecommand{\cnd}{\MATHIT{\mathfrak{b}}}
\providecommand{\cnne}{\mathSUB{n}{c}}
\providecommand{\cnnv}{\mathSUB{n}{v}}

Now consider the case where the population parameter
\pvsm[\sfactsym] is known or suspected, \emph{a priori}, 
to satisfy some constraints.
One then wishes to impose the same constraints on the sample
estimate prior to constructing the \txtMP, imposing a hedge, \etc

To avoid the possibility that the constrained estimate is not
positive definite or the need for cone programming
to find the estimate, here we assume 
the constraint can be expressed in
terms of the (lower) \emph{Cholesky factor} of \pvsm[\sfactsym]. 
Note that this takes the form
\begin{equation}
\label{eqn:chol_of_pvsmf}
\chol{\pvsm[\sfactsym]} 
= \twobytwo{\chol{\pfacsig}}{\mzero}{\pRegco\chol{\pfacsig}}{\chol{\pvsig}},
\end{equation}
as can be confirmed by multiplying the above by its transpose.

\subsection{Linear constraints}
Now consider equality constraints of the form
$\MGLHA \pRegco = \MGLHT$ for conformable matrices 
$\MGLHA, \MGLHT$, a less general form of the
Multivariate General Linear Hypothesis, of which more in the sequel. 
Via equalities of this form, one can constrain the mean of certain
assets to be zero (for example, assets to be hedged out),
or force certain elements of \vfact[i] to have no marginal
predictive ability on certain elements of \vreti[i+1].
When this constraint is satisfied, note that
\begin{equation*}
{\onebytwo{-\MGLHT}{\MGLHA}}\chol{\pvsm[\sfactsym]}
\twobyone{\eye}{\mzero} =
\mzero,
\end{equation*}
which can be rewritten as
\begin{equation*}
\begin{split}
\mzero 
&= \wrapParens{\onebytwo{\eye}{\mzero}\kron
	\onebytwo{-\MGLHT}{\MGLHA}}\fvec{\chol{\pvsm[\sfactsym]}},\\
&= \wrapParens{\onebytwo{\eye}{\mzero}\kron
	\onebytwo{-\MGLHT}{\MGLHA}}\tr{\Elim}\fvech{\chol{\pvsm[\sfactsym]}}.
\end{split}
\end{equation*}

This motivates the imposition of \cnne equality 
constraints of the form
\begin{equation}
\label{eqn:consest_def}
\cnB\fvech{\chol{\pvsm[\sfactsym]}} = \cnd,
\end{equation}
where \cnB is some \bby{\cnne}{\cnnv} matrix and \cnd is a \cnne-vector,
where $\cnnv = \wrapParens{\nlatf + \nfac + 1}\wrapParens{\nlatf + \nfac}/2$
is the number of elements in \fvech{\chol{\pvsm[\sfactsym]}}.

Now consider the optimization problem
\begin{equation}
\label{eqn:cons_optimization}
\min_{z :\, \cnB z = \cnd} \qform{\cnW}{\wrapParens{z - \fvech{\chol{\svsm[\sfactsym]}}}},
\end{equation}
where \cnW is some symmetric positive definite \sbby{\cnnv}
`weighting' matrix, the identity in the garden variety application.
The solution to this problem can easily be identified via the 
Lagrange multiplier technique to be
\begin{equation*}
\begin{split}
z_{*} &= \fvech{\chol{\svsm[\sfactsym]}} + 
\minv{\cnW}\tr{\cnB}\minvParens{\qiform{\cnW}{\cnB}}\wrapParens{\cnd -
\cnB\fvech{\chol{\svsm[\sfactsym]}}},\\
&= \minv{\cnW}\tr{\cnB}\minvParens{\qiform{\cnW}{\cnB}}\cnd 
+ \wrapBracks{\eye - \minv{\cnW}\wrapProj{\minv{\cnW}}{\cnB}}
\fvech{\chol{\svsm[\sfactsym]}}.
\end{split}
\end{equation*}

Define \svsmc[\sfactsym] to be the \sbby{\cnnv} matrix whose 
Cholesky factor solves minimization problem \ref{eqn:cons_optimization}:
\begin{equation*}
\svsmc[\sfactsym]\defeq\ogram{\wrapParens{\fivec{z_{*}}}}.
\end{equation*}
When the population parameter satisfies the constraints,
this sample estimate is asymptotically unbiased.

\begin{theorem}
\label{theorem:hej_cons_est_distribution}
Suppose $\cnB\fvech{\chol{\pvsm[\sfactsym]}} = \cnd$ for given
\bby{\cnne}{\cnnv} matrix \cnB and \cnne-vector \cnd. 
Let \cnW be a symmetric, positive definite \sbby{\cnnv} matrix.
Let $\svsm[\sfactsym] \defeq \oneby{\ssiz}\sum_i \ogram{\avreti[i+1]}$,
based on \ssiz \iid samples of \asvec{\tr{\vfact[{}]},\tr{\vreti}},
where 
$$
\avreti[i+1] \defeq \asvec{\tr{\vfact[i]},\tr{\vreti[i+1]}}.
$$
Let \pvvar be the variance of $\fvech{\ogram{\avreti}}$.
Define \svsmc[\sfactsym] such that
\begin{multline*}
\fvech{\chol{\svsmc[\sfactsym]}}
= \minv{\cnW}\tr{\cnB}\minvParens{\qiform{\cnW}{\cnB}}\cnd 
+ \\ \wrapBracks{\eye - \minv{\cnW}\wrapProj{\minv{\cnW}}{\cnB}}
\fvech{\chol{\svsm[\sfactsym]}}.
\end{multline*}

Then, asymptotically in \ssiz, 
\begin{equation}
\sqrt{\ssiz}\wrapParens{\fvech{\svsmc[\sfactsym]}
- \fvech{\pvsm[\sfactsym]}} 
\rightsquigarrow 
\normlaw{0,\qoform{\pvvar}{\Mtx{H}}},
\label{eqn:hej_cons_mvclt}
\end{equation}
where $\Mtx{H} = \Mtx{H}_1 \Mtx{H}_2 \Mtx{H}_3$ defined as
\begin{equation*}
\begin{split}
\Mtx{H}_1 &= 
\Elim\wrapParens{\eye + \Komm}\wrapParens{\chol{\pvsm[\sfactsym]} \kron \eye},\\
\Mtx{H}_2 &= \tr{\Elim} \wrapBracks{\eye -
\minv{\cnW}\wrapProj{\minv{\cnW}}{\cnB}},\\
\Mtx{H}_3 &= 
\minv{\wrapParens{%
\qoform{\wrapParens{\eye + \Komm}\wrapParens{\chol{\pvsm[\sfactsym]} \kron
\eye}}{\Elim}}},\\
\end{split}
\end{equation*}
where \Komm is the Commutation matrix.

Furthermore, we may replace \pvvar in this equation with an asymptotically
consistent estimator, \svvar.
\end{theorem}
\begin{proof}
\providecommand{\eff}[2]{\funcitL{f}{#1}{#2}}
Define the functions \eff{1}{\Mtx{X}}, \eff{2}{\Mtx{X}}, \eff{3}{\Mtx{X}} as follows:
\begin{equation*}
\begin{split}
\eff{1}{\Mtx{X}} &= \fvech{\ogram{\Mtx{X}}},\\
\eff{2}{\Mtx{X}} &= \fitril{\minv{\cnW}\tr{\cnB}\minvParens{\qiform{\cnW}{\cnB}}\cnd 
+ \wrapBracks{\eye - \minv{\cnW}\wrapProj{\minv{\cnW}}{\cnB}}\Mtx{X}},\\
\eff{3}{\Mtx{X}} &= \fvech{\chol{\Mtx{X}}},\\
\end{split}
\end{equation*}
where $\fitril{\Mtx{X}}$ is the function that takes a conformable vector to a
\emph{lower triangular} matrix. We have then defined 
\fvech{\svsmc[\sfactsym]} as $\eff{1}{\eff{2}{\eff{3}{\svsm[\sfactsym]}}}$. 
By the central limit theorem, and the matrix chain rule, it suffices to
show that $\Mtx{H}_1$ is the derivative of $\eff{1}{\cdot}$ evaluated at
$\eff{2}{\eff{3}{\pvsm[\sfactsym]}}$, that $\Mtx{H}_2$ is the derivative
of $\eff{2}{\cdot}$ evaluated at 
$\eff{3}{\pvsm[\sfactsym]}$, and $\Mtx{H}_3$ is the derivative of
$\eff{3}{\cdot}$ evaluated at \pvsm[\sfactsym]. 

These are established by \eqnref{mtx_ogram_rule} of 
\lemmaref{more_misc_derivs}, and \lemmaref{cholesky_deriv}, and by the
assumption that
$\cnB\fvech{\chol{\pvsm[\sfactsym]}} = \cnd$, which implies that
$\eff{2}{\eff{3}{\pvsm[\sfactsym]}} = \chol{\pvsm[\sfactsym]}$.

\end{proof}

The choice of \cnW is non-trivial. Armed with 
\theoremref{hej_cons_est_distribution} and knowledge of 
\pvsm and \pvvar, one would attempt to minimize the covariance
of the estimator \svsmc[\sfactsym]. Since these are
unknown population parameters, one would have to estimate them
somehow. \cite{Duncan1983,Cragg1983}

\nocite{petersen2012matrix}

The linear equality constraint can be generalized to a \emph{single} half-space 
inequality constraint. That is,
\begin{equation}
\label{eqn:consest_ineq_def}
\cnB\fvech{\chol{\pvsm[\sfactsym]}} \le \cnd,
\end{equation}
for \cnB a \bby{1}{\cnnv} matrix and \cnd a scalar.  \cite{nla.cat-vn3800977,tang1994}
However, the general case of multiple inequality constraints is much more
difficult, and remains an open question.


\subsection{Rank constraints}
\nocite{Izenman1975248}

\begingroup  
\providecommand{\rankr}{\MATHIT{r}}
\providecommand{\aaf}[3]{\mathUL{a}{\wrapParens{#1}}{#2#3}}
\providecommand{\Wul}[2]{\vectUL{W}{#1}{#2}}
\providecommand{\Wogram}[2]{\Wul{}{#1}\Wul{\trsym}{#2}}
\providecommand{\Vul}[2]{\vectUL{V}{#1}{#2}}
\providecommand{\Vogram}[2]{\Vul{}{#1}\Vul{\trsym}{#2}}
\providecommand{\oneproj}[3]{\MATHIT{\mathUL{f}{\wrapNeParens{#2}}{#1}\wrapParens{#3}}}
\providecommand{\haproj}[3]{\MATHIT{\mathUL{\pi}{\wrapNeParens{#2}}{#1}\wrapParens{#3}}}

\providecommand{\evalfunc}[2]{\MATHIT{\mathUL{v}{}{#1}\wrapParens{#2}}}
\providecommand{\evecfunc}[2]{\MATHIT{\mathUL{V}{}{#1}\wrapParens{#2}}}

Another plausible type of constraint is a rank constraint.  Here
it is suspected \emph{a priori} that the \sbby{\wrapParens{\nlatf+\nfac}}
matrix \pvsm[\sfactsym] has rank $\rankr < \nlatf + \nfac$. 
One sane response of a portfolio manager with this belief is to
project \svsm[\sfactsym] to a rank \rankr matrix, take the pseudoinverse,
and use the (negative) corner sub-matrix as the Markowitz coefficient. 
(\cf \lemmaref{hej_cond_sr_optimal_portfolio}) Here we consider the
asymptotic distribution of this reduced rank Markowitz coefficient.


To find the asymptotic distribution of sample estimates of \pvsm[\sfactsym] 
with a rank constraint, the derivative of the reduced rank decomposition is
needed. \cite{petersen2012matrix,Izenman1975248}
\begin{lemma}
Let \Mtx{A} be a real \sbby{J} symmetric matrix with rank $\rankr\le J$.
Let \evalfunc{j}{\Mtx{A}} be the function that returns the \kth{j} eigenvalue
of \Mtx{A}, and similarly let \evecfunc{j}{\Mtx{A}} compute the corresponding
eigenvector. Then
\begin{align}
\dbyd{\evalfunc{j}{\Mtx{A}}}{\vech{\Mtx{A}}} &= 
	\tr{\vech{\ogram{\evecfunc{j}{\Mtx{A}}}}} \Dupp,\\
\dbyd{\evecfunc{j}{\Mtx{A}}}{\vech{\Mtx{A}}} &= 
\pinv{\wrapParens{\evalfunc{j}{\Mtx{A}} \eye - \Mtx{A}}} 
\wrapParens{\tr{\evecfunc{j}{\Mtx{A}}} \kron \eye} \Dupp.
\end{align}
\end{lemma}
\begin{proof}
The derivatives are known. \cite[equation (67)-(68)]{petersen2012matrix} The
form here follows from algebra and \lemmaref{misc_derivs}.
\end{proof}
From these, the derivative of the \sbby{\rankr} diagonal matrix with 
diagonal
$$
\asvec{\evalfunc{1}{\Mtx{A}}^p,\evalfunc{2}{\Mtx{A}}^p,\ldots,\evalfunc{\rankr}{\Mtx{A}}^p}
$$
can be computed with respect to \vech{\Mtx{A}} for arbitrary non-zero $p$.
Similarly the derivative of the matrix whose columns are 
$\evalfunc{1}{\Mtx{A}},\evalfunc{2}{\Mtx{A}},\ldots,\evalfunc{\rankr}{\Mtx{A}}$
can be computed with respect to \vech{\Mtx{A}}. From these the derivative of
the pseudo-inverse of the rank \rankr approximation to \Mtx{A} can be computed
with respect to \vech{\Mtx{A}}. By the delta method, then, an asymptotic normal
distribution of the pseudo-inverse of the rank \rankr approximation to \Mtx{A}
can be established. The formula for the variance is best left unwritten, since
it would be too complex to be enlightening, and is best constructed by
automatic differentiation anyway.

\endgroup 

\section{The multivariate general linear hypothesis}
\label{sec:MGLH}

Dropping the conditional heteroskedasticity term from
\eqnref{cond_model_IV}, we have the model
\begin{align*}
\Econd{\vreti[i+1]}{\vfact[i]} &= \pRegco \vfact[i], &
\Varcond{\vreti[i+1]}{\vfact[i]} &= \pvsig,
\end{align*}
The unknowns \pRegco and \pvsig can be estimated by multivariate
multiple linear regression. 
Testing for significance of the elements of \pRegco is via the 
\emph{Multivariate General Linear Hypothesis} 
(MGLH).
\cite{Muller1984143,shieh2003,shieh2005,obrienshieh,timm2002applied,yanagihara2001}
The MGLH can be posed as
\begin{equation}
\label{eqn:MGLH_def}
\Hyp[0]: \MGLHA \pRegco \MGLHC = \MGLHT,
\end{equation}
for \bby{\MGLHa}{\nlatf} matrix \MGLHA, 
\bby{\nfac}{\MGLHc} matrix \MGLHC,
and \bby{\MGLHa}{\MGLHc} matrix \MGLHT. 
We require \MGLHA and \MGLHC to have full rank, and
$\MGLHa\le\nlatf$ and $\MGLHc\le\nfac.$ 
In the garden-variety application one tests whether \pRegco is all
zero by letting \MGLHA and \MGLHC be identity matrices, and \MGLHT 
a matrix of all zeros.

Testing the MGLH proceeds by one of four test statistics, each
defined in terms of two matrices, the model variance matrix, \sMGLHH,
and the error variance matrix, \sMGLHE, defined as
\begin{equation}
\label{eqn:mglh_HE_def}
\sMGLHH \defeq
\qoiform{\wrapParens{\qiform{\sfacsig}{\MGLHC}}}{\wrapParens{\MGLHA\sRegco\MGLHC - \MGLHT}},
\quad
\sMGLHE \defeq \qform{\svsig}{\MGLHA},
\end{equation}
where $\sfacsig = \oneby{\ssiz}\sum_i \ogram{\vfact[i]}$.
Note that typically in non-finance applications, the regressors are
deterministic and controlled by the experimenter 
(giving rise to the term `design matrix'). In this case, it is
assumed that \sfacsig estimates the population analogue, \pfacsig,
without error, though some work has been done for the case of
`random explanatory variables.' \cite{shieh2005}

The four test statistics for the MGLH are:
\begin{align}
\label{eqn:sHLT_def}
\mbox{Hotelling-Lawley trace:}\quad 
\sHLT &\defeq \trace{\minvAB{\sMGLHE}{\sMGLHH}} = 
\trace{\eye[\MGLHa] + \minvAB{\sMGLHE}{\sMGLHH}} - \MGLHa,&\\
\label{eqn:sPBT_def}
\mbox{Pillai-Bartlett trace:}\quad 
\sPBT &\defeq \trace{\minvParens{\eye[\MGLHa] + \minvAB{\sMGLHE}{\sMGLHH}}},&\\
\label{eqn:sWILK_def}
\mbox{Wilk's LRT:}\quad 
\sWILK &\defeq \det{\minvParens{\eye[\MGLHa] + \minvAB{\sMGLHE}{\sMGLHH}}},&\\
\label{eqn:sRLR_def}
\mbox{Roy's largest root:}\quad 
\sRLR &\defeq \fmax{\feig{\minvAB{\sMGLHE}{\sMGLHH}}},&\\
&= \fmax{\feig{\eye[\MGLHa] + \minvAB{\sMGLHE}{\sMGLHH}}} - 1.&\nonumber
\end{align}
Of these four, Roy's largest root has historically been the least
well understood.  \cite{johnstone2009} Each of these can
be described as some function of the eigenvalues of
the matrix 
$\eye[\MGLHa] + \minvAB{\sMGLHE}{\sMGLHH}$.
Under the null hypothesis, \Hyp[0], the matrix \sMGLHH `should' be
all zeros, in a sense that will be made precise later, and thus
the Hotelling Lawley trace and Roy's largest root `should' equal zero, 
the Pillai Bartlett trace `should' equal \MGLHa, and Wilk's LRT
`should' equal $1$.

One can describe the MGLH tests statistics in terms of the
asymptotic expansions of the matrix \pvsm[\sfactsym] given in the
previous sections.  As in 
\subsecref{cond_ret_het}, let 
\begin{equation}
\avreti[i+1] \defeq \asvec{\tr{\vfact[i]},\tr{\vreti[i+1]}}.
\end{equation}
The second moment of \avreti is
\begin{equation}
\pvsm[\sfactsym] \defeq \E{\ogram{\avreti}} = 
	\twobytwo{\pfacsig}{\pfacsig\tr{\pRegco}}{\pRegco\pfacsig}{\pvsig +
\qoform{\pfacsig}{\pRegco}}.
\label{eqn:cond_pvsm_def_II}
\end{equation}

We can express the MGLH statistics in terms of the product of two matrices
defined in terms of \pvsm[\sfactsym].
Let \mglhM be the 
\bby{\wrapParens{\nfac+\nlatf}}{\wrapParens{\MGLHc+\MGLHa}} matrix
\begin{equation*}
\mglhM \defeq \twobytwossym{\eye[\nfac]}{0}{\MGLHA}.
\end{equation*}
Linear algebra confirms that
\begin{equation}
\minvParens{\qform{\pvsm[\sfactsym]}{\mglhM}} 
= \twobytwo{\minv{\sfacsig} + 
\qiform{\wrapParens{\qoform{\svsig}{\MGLHA}}}{\wrapParens{\MGLHA\sRegco}}}{%
-\tr{\sRegco}\tr{\MGLHA}\minvParens{\qoform{\svsig}{\MGLHA}}}{%
-\minvAB{\wrapParens{\qoform{\svsig}{\MGLHA}}}{\MGLHA\sRegco}}{%
\minvParens{\qoform{\svsig}{\MGLHA}}}.
\label{eqn:mglh_magic_inversion}
\end{equation}
Thus
\begin{multline}
\label{eqn:mglhR_def}
\mglhR\defeq
\onebytwo{\tr{\MGLHC}}{\tr{\MGLHT}} \minvParens{\qform{\pvsm[\sfactsym]}{\mglhM}} 
\twobyone{\MGLHC}{\MGLHT}
=\\
\qiform{\sfacsig}{\MGLHC} +
\qiform{\wrapParens{\qoform{\svsig}{\MGLHA}}}{\wrapParens{\MGLHA\sRegco\MGLHC -
\MGLHT}}.
\end{multline}
Now define 
\begin{equation}
\label{eqn:mglhL_def}
\mglhL\defeq
\minvParens{\qiform{\wrapParens{\onebytwo{\tr{\eye[\nfac]}}{\tr{\mzero[\nlatf]}}
\pvsm[\sfactsym] \twobyone{\eye[\nfac]}{\mzero[\nlatf]}}}{\MGLHC}} =
\minvParens{\qiform{\sfacsig}{\MGLHC}}.
\end{equation}

Thus
\begin{equation}
\mglhL\mglhR = 
\eye[\MGLHc] + 
\minvParens{\qiform{\sfacsig}{\MGLHC}}\qiform{\wrapParens{\qoform{\svsig}{\MGLHA}}}{\wrapParens{\MGLHA\sRegco\MGLHC - \MGLHT}}.
\end{equation}
This matrix is `morally equivalent'\footnote{To quote my advisor, Noel
Walkington.} to the matrix $\fvech{\eye[\MGLHa] +
\minvAB{\sMGLHE}{\sMGLHH}}$, in that they have the same eigenvalues.
This holds because $\feig{\Mtx{A}\Mtx{B}} = \feig{\Mtx{B}\Mtx{A}}$. 
\cite[equation (280)]{petersen2012matrix} Taking into account that they
have diferent sizes (one is \sbby{\MGLHa}, the other \sbby{\MGLHc}),
the MGLH statistics can be expressed as:
\begin{align*}
\sHLT &= \trace{\mglhL\mglhR} - \MGLHc,&\quad
\sPBT &= \trace{\minvParens{\mglhL\mglhR}} + \MGLHa - \MGLHc,&\\
\sWILK &= \det{\minvParens{\mglhL\mglhR}},&\quad
\sRLR &= \fmax{\feig{\mglhL\mglhR}} - 1.&
\end{align*}

To find the asymptotic distribution of $\mglhL\mglhR$, 
and of the MGLH test statistics,
the derivatives of the matrices above with respect to 
\pvsm[\sfactsym] need to be found. In practice this would
certainly be better achieved through automatic 
differentation.  \cite{rall1981automatic}
For concreteness, however, the derivatives are given here.

It must also be noted that the
straightforward application of the delta method results in asymptotic
\emph{normal} approximations for the MGLH statistics. By their very
nature, however, these statistics look much more like (non-central)
Chi-square or F statistics.  \cite{obrienshieh} 
Further study is warranted on this matter, perhaps using Hall's approach.  \cite{hall1983chisquare}

\begin{lemma}[Some derivatives]
\label{lemma:Z_derivatives}
Define 
\begin{equation}
\begin{split}
\mglhL &\defeq
\minvParens{\qiform{\wrapParens{\onebytwo{\tr{\eye[\nfac]}}{\tr{\mzero[\nlatf]}}
\pvsm[\sfactsym] \twobyone{\eye[\nfac]}{\mzero[\nlatf]}}}{\MGLHC}},\\
\mglhR &\defeq \wrapParens{
\onebytwo{\tr{\MGLHC}}{\tr{\MGLHT}} \minvParens{\qform{\pvsm[\sfactsym]}{\mglhM}} 
\twobyone{\MGLHC}{\MGLHT}}.
\end{split}
\end{equation}
Then
\begin{equation*}
\begin{split}
\dbyd{\mglhL}{\pvsm[\sfactsym]} &= 
\fdinvwrap{\MGLHC}{%
\minv{\wrapParens{\onebytwo{\tr{\eye[\nfac]}}{\tr{\mzero[\nlatf]}}
\pvsm[\sfactsym] \twobyone{\eye[\nfac]}{\mzero[\nlatf]}}}}
\fdinvwrap{\twobyone{\eye[\nfac]}{\mzero[\nlatf]}}{\pvsm[\sfactsym]},\\
\dbyd{\mglhR}{\pvsm[\sfactsym]} &= 
\wrapParens{\AkronA{\onebytwo{\tr{\MGLHC}}{\tr{\MGLHT}}}}
\fdinvwrap{\mglhM}{\pvsm[\sfactsym]},\\
\dbyd{\minv{\mglhL}}{\pvsm[\sfactsym]} &= 
\wrapParens{\AkronA{\tr{\MGLHC}}}
\fdinvwrap{\twobyone{\eye[\nfac]}{\mzero[\nlatf]}}{\pvsm[\sfactsym]},\\
\dbyd{\minv{\mglhR}}{\pvsm[\sfactsym]} &= 
\fdinvwrap{\twobyone{\MGLHC}{\MGLHT}}{%
\minvParens{\qform{\pvsm[\sfactsym]}{\mglhM}}}
\fdinvwrap{\mglhM}{\pvsm[\sfactsym]},
\end{split}
\end{equation*}
where we define
\begin{equation*}
\fdinvwrap{\Mtx{J}}{\Mtx{X}} 
\defeq - \wrapParens{\AkronA{\minvParens{\qform{\Mtx{X}}{\Mtx{J}}}}}\wrapParens{\AkronA{\tr{\Mtx{J}}}}.
\end{equation*}
\end{lemma}
\begin{proof}
These follow from \lemmaref{more_misc_derivs} and the chain rule.
\end{proof}
\providecommand{\Qderiv}[1]{\MtxUL{Q}{}{#1}}

\begin{lemma}[MGLH derivatives]
\label{lemma:MGLH_derivatives}
Define the population analogues of the MGLH statistics as
\begin{align}
\pHLT &\defeq \trace{\eye[\MGLHa] + \minvAB{\pMGLHE}{\pMGLHH}} - \MGLHa,&
\pPBT &\defeq \trace{\minvParens{\eye[\MGLHa] + \minvAB{\pMGLHE}{\pMGLHH}}},&\\
\pWILK &\defeq \det{\minvParens{\eye[\MGLHa] + \minvAB{\pMGLHE}{\pMGLHH}}},&
\pRLR &\defeq \fmax{\feig{\eye[\MGLHa] + \minvAB{\pMGLHE}{\pMGLHH}}} - 1,&
\end{align}
where 
$$
\pMGLHH \defeq
\qoiform{\wrapParens{\qiform{\pfacsig}{\MGLHC}}}{\wrapParens{\MGLHA\pRegco\MGLHC - \MGLHT}},
\quad
\pMGLHE \defeq \qform{\pvsig}{\MGLHA}.
$$

Define $\pvsm[\sfactsym ]$ as in \eqnref{cond_pvsm_def_II}. Let 
$$\Qderiv{\pHLT}\defeq\dbyd{\pHLT}{{\pvsm[\sfactsym ]}},$$
and similarly define
$\Qderiv{\pPBT}, \Qderiv{\pWILK}, \Qderiv{\pRLR}.$
Then
\begin{align*}
\Qderiv{\pHLT} &=
\tr{\fvec{\tr{\mglhL}}}\dbyd{\mglhR}{\pvsm[\sfactsym]} + 
\tr{\fvec{\mglhR}}\dbyd{\tr{\mglhL}}{\pvsm[\sfactsym]},&\\
\Qderiv{\pPBT} &= 
\tr{\fvec{\trminv{\mglhL}}}\dbyd{\minv{\mglhR}}{\pvsm[\sfactsym]} + 
\tr{\fvec{\minv{\mglhR}}}\dbyd{\trminv{\mglhL}}{\pvsm[\sfactsym]},&\\
\Qderiv{\pWILK} &= 
\det{\mglhL\mglhR}^{-1}
\wrapParens{ \tr{\fvec{\trminv{\mglhL}}}\dbyd{\mglhL}{\pvsm[\sfactsym]} +
\tr{\fvec{\trminv{\mglhR}}}\dbyd{\mglhR}{\pvsm[\sfactsym]}},&\\
\Qderiv{\pRLR} &= 
\wrapParens{\AkronA{\tr{\eigvec[1]}}}
\wrapBracks{%
\wrapParens{\eye\kron\mglhL}\dbyd{\mglhR}{\pvsm[\sfactsym]} +
\wrapParens{\tr{\mglhR}\kron\eye}\dbyd{\mglhL}{\pvsm[\sfactsym]}},&
\end{align*}
where \eigvec[1] is the leading eigenvector of $\mglhL\mglhR$, normalized
by $\gram{\eigvec[1]} = 1$, and where
where \mglhR and \mglhL are defined in \eqnref{mglhR_def} and
\eqnref{mglhL_def}, and their derivatives with respect to 
$\pvsm[\sfactsym]$ are given in \lemmaref{Z_derivatives}.
\end{lemma}
\begin{proof}
These follow from the definition of the MGLH statistics, and 
Equations\nobreakspace\ref{eqn:mtx_trace_deriv}, 
\ref{eqn:mtx_inv_prod_det_deriv}, and \ref{eqn:mtx_eig_deriv}
of \lemmaref{more_misc_derivs}.
\end{proof}
\begin{theorem}
\label{theorem:mglh_distribution}
Let $\svsm[\sfactsym ] \defeq \oneby{\ssiz}\sum_i \ogram{\avreti[i+1]}$,
based on \ssiz \iid samples of \asvec{\tr{\vfact[{}]},\tr{\vreti}},
where 
$$
\avreti[i+1] \defeq \asvec{\tr{\vfact[i]},\tr{\vreti[i+1]}}.
$$
Let \pvvar be the variance of $\fvech{\ogram{\avreti}}$.
Define \sMGLHH and \sMGLHE as in \eqnref{mglh_HE_def}, for
given \bby{\MGLHa}{\nlatf} matrix \MGLHA, 
\bby{\nfac}{\MGLHc} matrix \MGLHC,
and \bby{\MGLHa}{\MGLHc} matrix \MGLHT. 

Define the MGLH test statistics, \sHLT, \sPBT, \sWILK, and
\sRLR as in Equations\nobreakspace\ref{eqn:sHLT_def} through
\ref{eqn:sRLR_def}, and let \pHLT, \pPBT, \pWILK and \pRLR be
their population analogues. 

Then, asymptotically in \ssiz, 
\begin{align*}
\sqrt{\ssiz}\wrapParens{\sHLT - \pHLT}
&\rightsquigarrow 
\normlaw{0,\qoform{\pvvar}{\wrapParens{\Qderiv{\pHLT}\Dupp}}},&\\
\sqrt{\ssiz}\wrapParens{\sPBT - \pPBT}
&\rightsquigarrow 
\normlaw{0,\qoform{\pvvar}{\wrapParens{\Qderiv{\pPBT}\Dupp}}},&\\
\sqrt{\ssiz}\wrapParens{\sWILK - \pWILK}
&\rightsquigarrow 
\normlaw{0,\qoform{\pvvar}{\wrapParens{\Qderiv{\pWILK}\Dupp}}},&\\
\sqrt{\ssiz}\wrapParens{\sRLR - \pRLR}
&\rightsquigarrow 
\normlaw{0,\qoform{\pvvar}{\wrapParens{\Qderiv{\pRLR}\Dupp}}},&
\end{align*}
where \Qderiv{\pHLT}, \Qderiv{\pPBT}, \Qderiv{\pWILK}, \Qderiv{\pRLR} are
given in \lemmaref{MGLH_derivatives}.
Furthermore, we may replace \pvvar in this equation with an asymptotically
consistent estimator, \svvar.
\end{theorem}
\begin{proof}
This follows from the delta method and
\lemmaref{MGLH_derivatives}.
\end{proof}

\section{Examples}

\subsection{Random Data}



Empirically, the marginal Wald test for zero weighting in the 
\txtMP based on the approximation of 
\theoremref{inv_distribution} are nearly
identical to the \tstat-statistics produced by the procedure
of Britten-Jones. \cite{BrittenJones1999} Here $1024$
days of Gaussian returns for $5$ assets with mean zero and
some randomly generated covariance are randomly generated. 
The procedure of Britten-Jones is applied marginally on each
asset. The Wald statistic is also computed via 
\theoremref{inv_distribution} by `plugging in' the sample 
estimate, \svsm, to estimate the standard error. The two test
values for the $5$ assets are presented in 
\tabref{simple_null}, and match very well.

The value of the asymptotic approach is that it admits the
generalizations of \secref{extensions}, and allows robust
estimation of \pvvar.  \cite{Zeileis:2004:JSSOBK:v11i10}

\begin{table}[ht]
\centering
\begin{tabular}{rrrrrr}
  \hline
 & rets1 & rets2 & rets3 & rets4 & rets5 \\ 
  \hline
Britten.Jones & 0.4950 & 0.0479 & 1.2077 & -0.4544 & -1.4636 \\ 
  Wald & 0.4965 & 0.0479 & 1.2107 & -0.4573 & -1.4635 \\ 
   \hline
\end{tabular}
\caption{The \tstat statistics of Britten-Jones \cite{BrittenJones1999} are presented, along with the Wald statistics from plugging in \svsm for \pvsm in \theoremref{inv_distribution}, for 1024 days of Gaussian returns of 5 assets with zero mean. Statistics are presented with 4 significant digits to illustrate the difference in values of the two methods, not because these statistics are worthy of such accuracy.} 
\label{tab:simple_null}
\end{table}

\subsubsection{Normal Returns}
\label{sec:normal_sims}


We test the confidence intervals of \theoremref{mp_snr_ci_elliptical} and \eqnref{eta_form_psnr_mp_II} 
for random data. We draw returns from the multivariate normal distribution, for which $\kurty=1$.
We fix the number of days of data, \ssiz, the number
of assets, \nlatf, and the optimal \txtSNR, \psnropt, and perform 
$\ensuremath{2.5\times 10^{4}}$ simulations. 
We then let \ssiz vary from 100 to $\ensuremath{1.28\times 10^{4}}$ days; 
we let \nlatf vary from 2 to 16;
we let \psnropt vary from 0.5 to 2 in
`annualized' units (per square root year), where we assume
252 days per year. We compute the lower confidence limits on
$\pSNR{\pportw ; \pvsm, 0}$, the \txtSNR of the sample \txtMP
based on the difference and ratio forms from the theorem.
The confidence limtis are computed \emph{very} optimistically, 
by using the \emph{actual} \psnropt in the expressions for the
mean and variance of 
${\pSNR{\minv{\svsm} ; \pvsm, 0} - \ssropt}$ and
${\pSNR{\minv{\svsm} ; \pvsm, 0} / \ssropt}$. 
For the `TAS' form of confidence limit, we use \ssropt once in the
non-centrality parameter, but otherwise use the actual \psnropt when computing
parameters.
Thus while this does not test the confidence limits in the way they would 
be practically used (\eg \eqnref{psnr_ci_practical_elliptical}), the
results are sufficiently discouraging even with this bit of clairvoyance
to recommend against their general use.


We compute the lower 0.05 confidence limit based on the difference
and ratio forms and the `TAS' transform. 
We then compute the empirical type I rate. 
These are plotted against \ssiz in \figref{badci_ci_plots}.
We show facet columns for \psnropt, and facet rows for \nlatf. The confidence
intervals fail to achieve nominal coverage except perhaps for the largest
values of \ssiz, though these are much larger than would be used in practice.

\begin{knitrout}\small
\definecolor{shadecolor}{rgb}{0.969, 0.969, 0.969}\color{fgcolor}\begin{figure}
\includegraphics[width=\maxwidth]{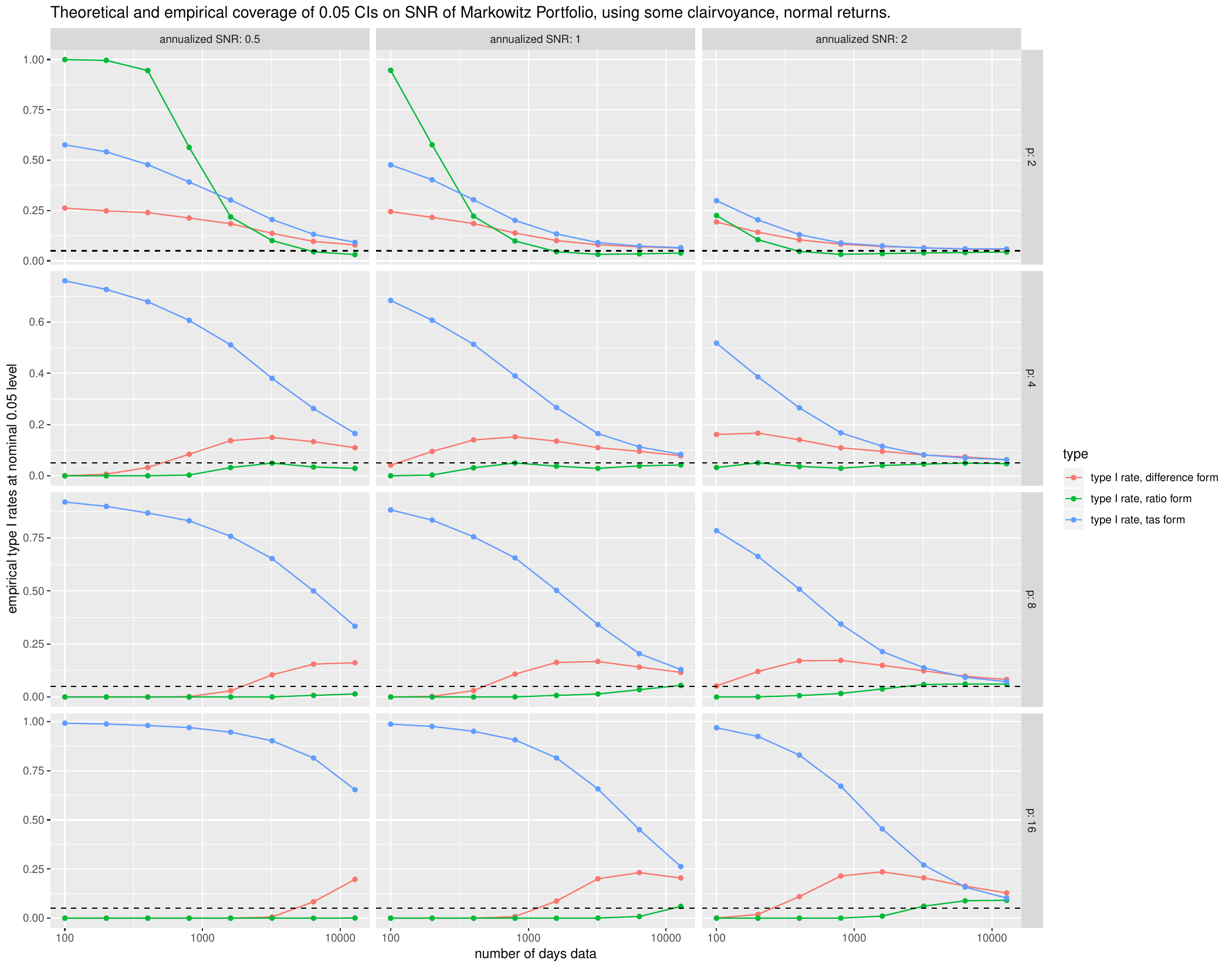} \caption{The empirical type I rate of three different one-sided confidence intervals for $\pSNR{\minv{\svsm} ; \pvsm, 0}$, the \txtSNR of the \txtMP are shown, where the nominal type I rate is $0.05$. The daily returns are drawn from multivariate normal distribution with varying \psnropt, \ssiz, and \nlatf.  The confidence intervals generally fail to achieve the nominal rate except for unrealistically large values of \ssiz.}\label{fig:badci_ci_plots}
\end{figure}

\end{knitrout}



As a check, we also compare the empirical mean of the \txtSNR of the \txtMP
from our experiments
with the theoretical asymptotic value from 
\theoremref{mp_snr_ci_elliptical}, namely
$$
\frac{ \wrapParens{\kurty \psnrsqopt + 1}\wrapParens{1 - \nlatf}}{2\ssiz\psnropt}.
$$
We plot the empirical and theoretical means in 
\figref{badci_mean_mp_snr_plots}, again versus \ssiz with 
facet columns for \psnropt, and facet rows for \nlatf. The theoretical asymptotic
value gives a good approximation for larger sample sizes, but `only' around
6 years of daily data are required. Note that the theoretical value gets
worse as $\ssiz\psnropt \to 0$, as one would expect: the \txtSNR of \emph{any} portfolio
must be no greater than \psnropt in absolute value, but the theoretical value goes to
$-\infty$.

\begin{knitrout}\small
\definecolor{shadecolor}{rgb}{0.969, 0.969, 0.969}\color{fgcolor}\begin{figure}
\includegraphics[width=\maxwidth]{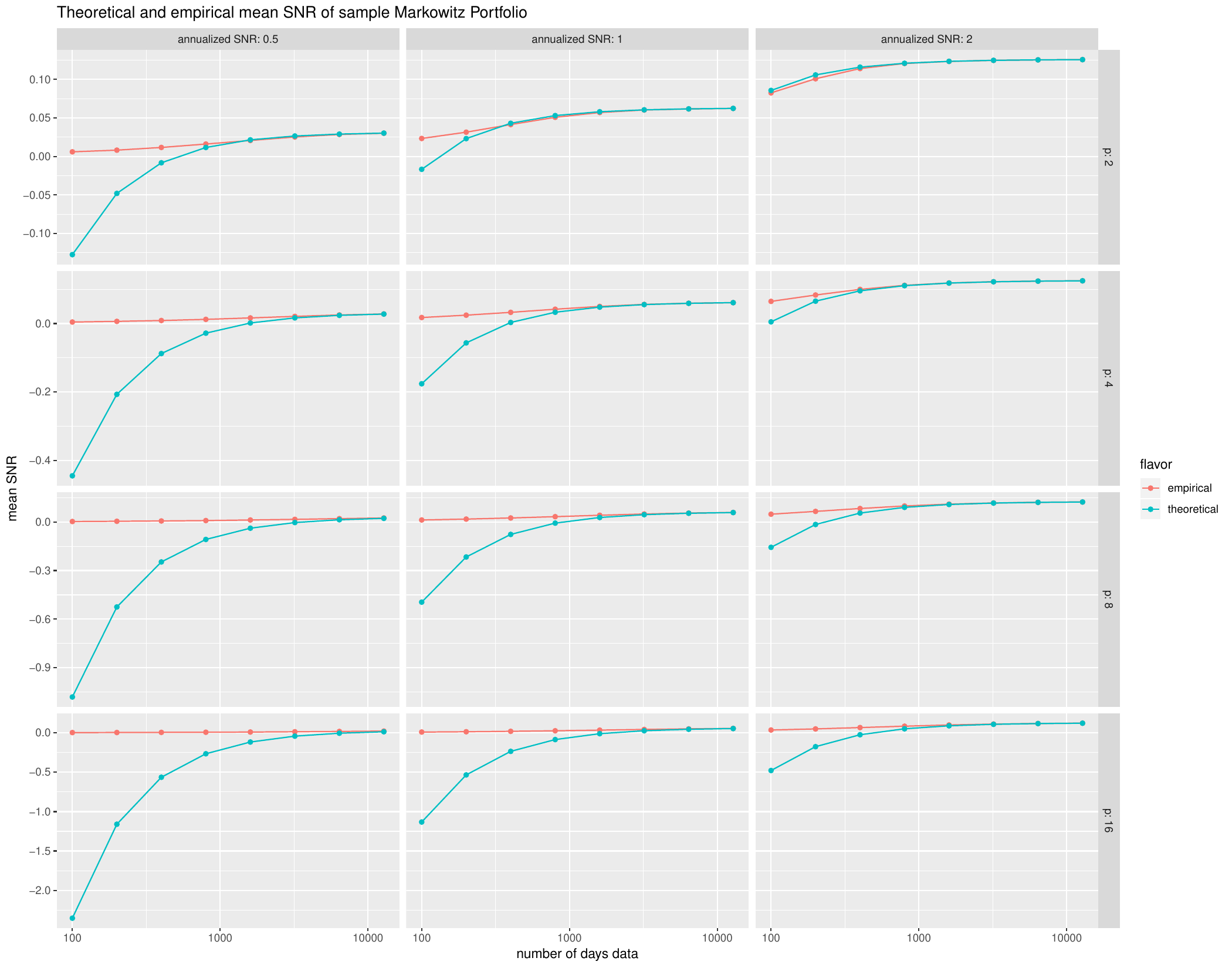} \caption[The empirical and theoretical asymptotic mean value of the \txtSNR of the \txtMP are shown versus sample size, \ssiz, for varying \psnropt and \nlatf]{The empirical and theoretical asymptotic mean value of the \txtSNR of the \txtMP are shown versus sample size, \ssiz, for varying \psnropt and \nlatf.}\label{fig:badci_mean_mp_snr_plots}
\end{figure}

\end{knitrout}

\subsubsection{Elliptical Returns}
\label{sec:ellip_sims}


We now test \eqnref{mvclt_isvsm_qbits_elliptic_explicit} via simulation.
We fix the number of days of data, \ssiz, the number
of assets, \nlatf, and the optimal \txtSNR, \psnropt, and 
the kurtosis factor, \kurty, and perform $\ensuremath{10^{4}}$ simulations. 
We let \ssiz vary from 100 to $1600$ days; 
we let \nlatf vary from 4 to 16;
we let \kurty vary from 1 to 16;
we let \psnropt vary from 0.5 to 2 in
`annualized' units (per square root year), where we assume
252 days per year.  When $\kurty=1$ we draw from a multivariate
normal distribution; when $\kurty > 1$, we draw from a multivariate shifted
$t$ distribution.

For each simulation, we collect the first and last elements of the vector
$$\sqrt{\ssiz} \wrapParens{ {\Mtx{Q}\trchol{\pvsig}\sportwopt} - {\psnropt\basev[1]} } \Mtx{D}^{-1/2}.$$
By \eqnref{mvclt_isvsm_qbits_elliptic_explicit} these should
be asymptotically distributed as a standard normal.
In \figref{markocov_firstv_qq_plots} we give Q-Q plots of the first
element of this vector versus quantiles of the standard normal
for each setting of the simulation parameters. 
Rather than present the full Q-Q plot
(each facet would contain 10,000 points),
we take evenly spaced points between $-4$ and $4$, then convert them
into percentage points of the standard normal. 
We then find the empirical quantiles at those levels and plot the empirical
quantiles against the selected theoretical.
This allows us to also plot the pointwise confidence bands, which should be
very small for the sample size, except at the periphery.
Similarly, 
in \figref{markocov_lastv_qq_plots} we give the same kind of subsampled
Q-Q plots of the last element of this vector.

With some exception, the Q-Q plots show a fair degree of support
for normality when $\ssiz / \nlatf$ is reasonably large. For example,
for $\nlatf=4$, a sample of $\ssiz=400$ is apparently sufficient
to get near normality for the last element of the vector. The first
element of the vector, which one suspects is dependent on \psnropt,
appears to suffer when \psnropt is larger.


\begin{knitrout}\small
\definecolor{shadecolor}{rgb}{0.969, 0.969, 0.969}\color{fgcolor}\begin{figure}
\includegraphics[width=\maxwidth]{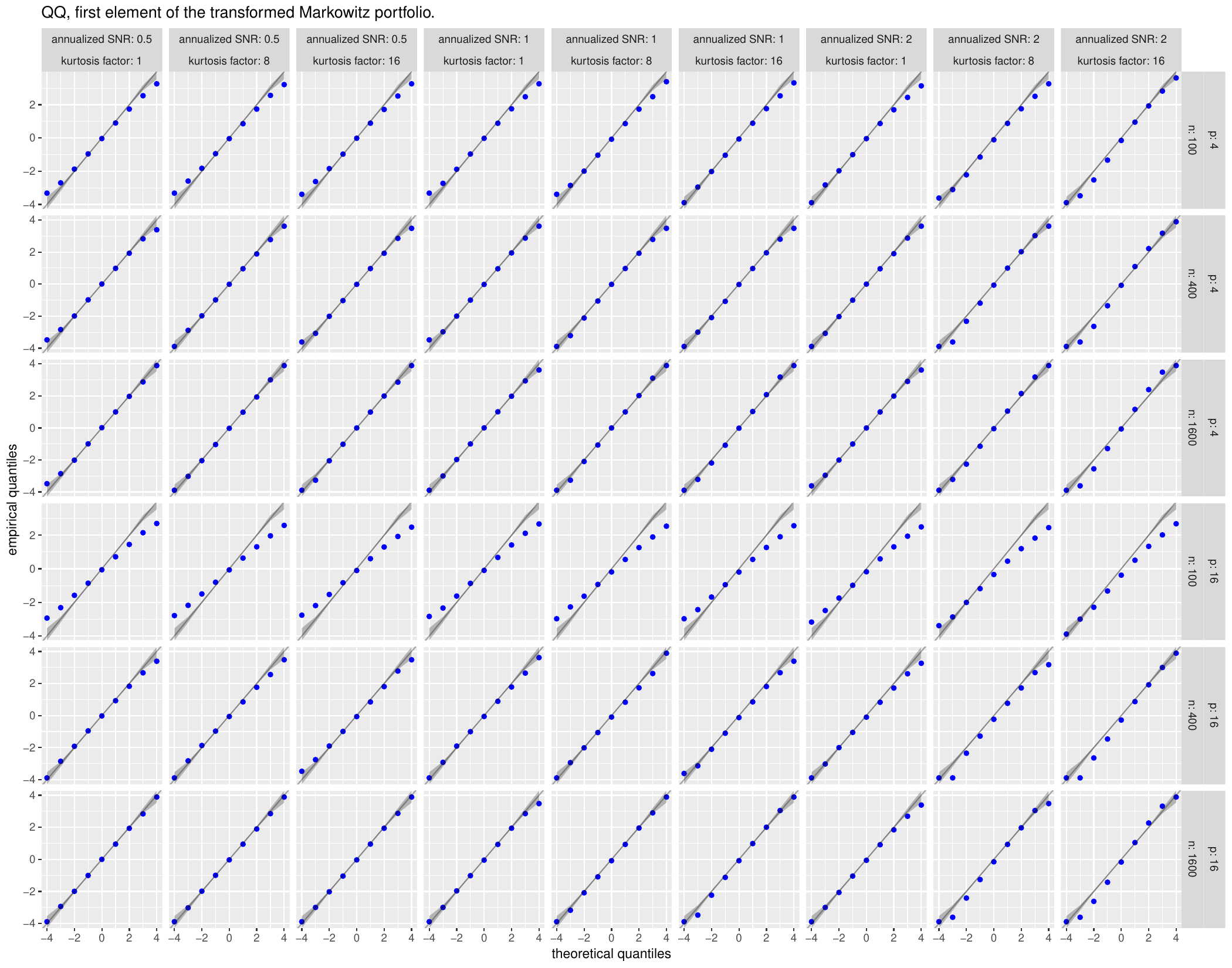} \caption[Q-Q plots of the first element of the transformed \txtMP, transformed to an approximate standard normal, versus normal quantiles are shown for varying \nlatf, \ssiz, \psnropt, and \kurty]{Q-Q plots of the first element of the transformed \txtMP, transformed to an approximate standard normal, versus normal quantiles are shown for varying \nlatf, \ssiz, \psnropt, and \kurty. Evenly spaced theoretical quantiles are used to compute the empirical quantiles, allowing us to plot 0.95\% confidence bands. }\label{fig:markocov_firstv_qq_plots}
\end{figure}

\end{knitrout}

\begin{knitrout}\small
\definecolor{shadecolor}{rgb}{0.969, 0.969, 0.969}\color{fgcolor}\begin{figure}
\includegraphics[width=\maxwidth]{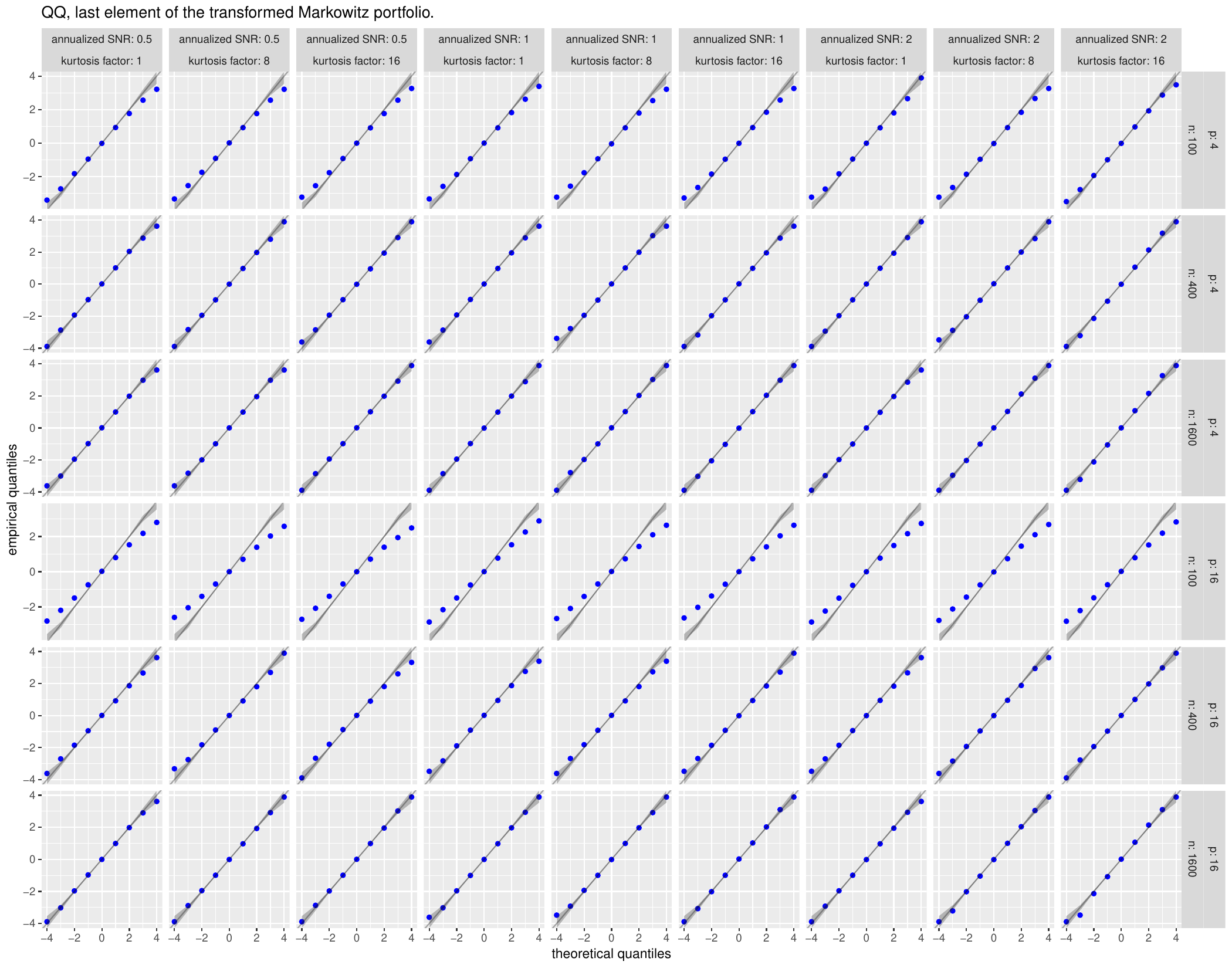} \caption[Q-Q plots of the \emph{last} element of the transformed \txtMP, transformed to an approximate standard normal, versus normal quantiles are shown for varying \nlatf, \ssiz, \psnropt, and \kurty]{Q-Q plots of the \emph{last} element of the transformed \txtMP, transformed to an approximate standard normal, versus normal quantiles are shown for varying \nlatf, \ssiz, \psnropt, and \kurty. Evenly spaced theoretical quantiles are used to compute the empirical quantiles, allowing us to plot 0.95\% confidence bands. }\label{fig:markocov_lastv_qq_plots}
\end{figure}

\end{knitrout}




We also check \eqnref{eta_form_psnr_mp} via a smaller set of simulations.
We fix $\psnropt=2$ in annualized units,
$\nlatf=6$, and $\kurty=4$. We then test two different
sample sizes: $\ssiz=2,520$ and $\ssiz=25,200$, corresponding
to 10 and 100 years of daily data, at a rate of $252$ days per year.
For a given simulation we draw the appropriate number of days of independent
returns from a shifted multivariate $t$ distribution. We compute the sample \txtMP
and compute its \txtSNR, $\pSNR{\sportwopt}$. We perform this simulation
50,000 times for each setting of \ssiz.
We construct theoretical approximate quantiles of
$\pSNR{\sportwopt}$ using \eqnref{eta_form_psnr_mp}, then plot the empirical quantiles
versus the theoretical quantiles in \figref{snr_qq_plots_for_100}.
It appears that more than 10 years of daily data (but fewer than 100 years worth)
are required for this approximation to be any good.
When the approximation breaks down, it tends to underestimate the true \txtSNR.


\begin{knitrout}\small
\definecolor{shadecolor}{rgb}{0.969, 0.969, 0.969}\color{fgcolor}\begin{figure}
\includegraphics[width=\maxwidth]{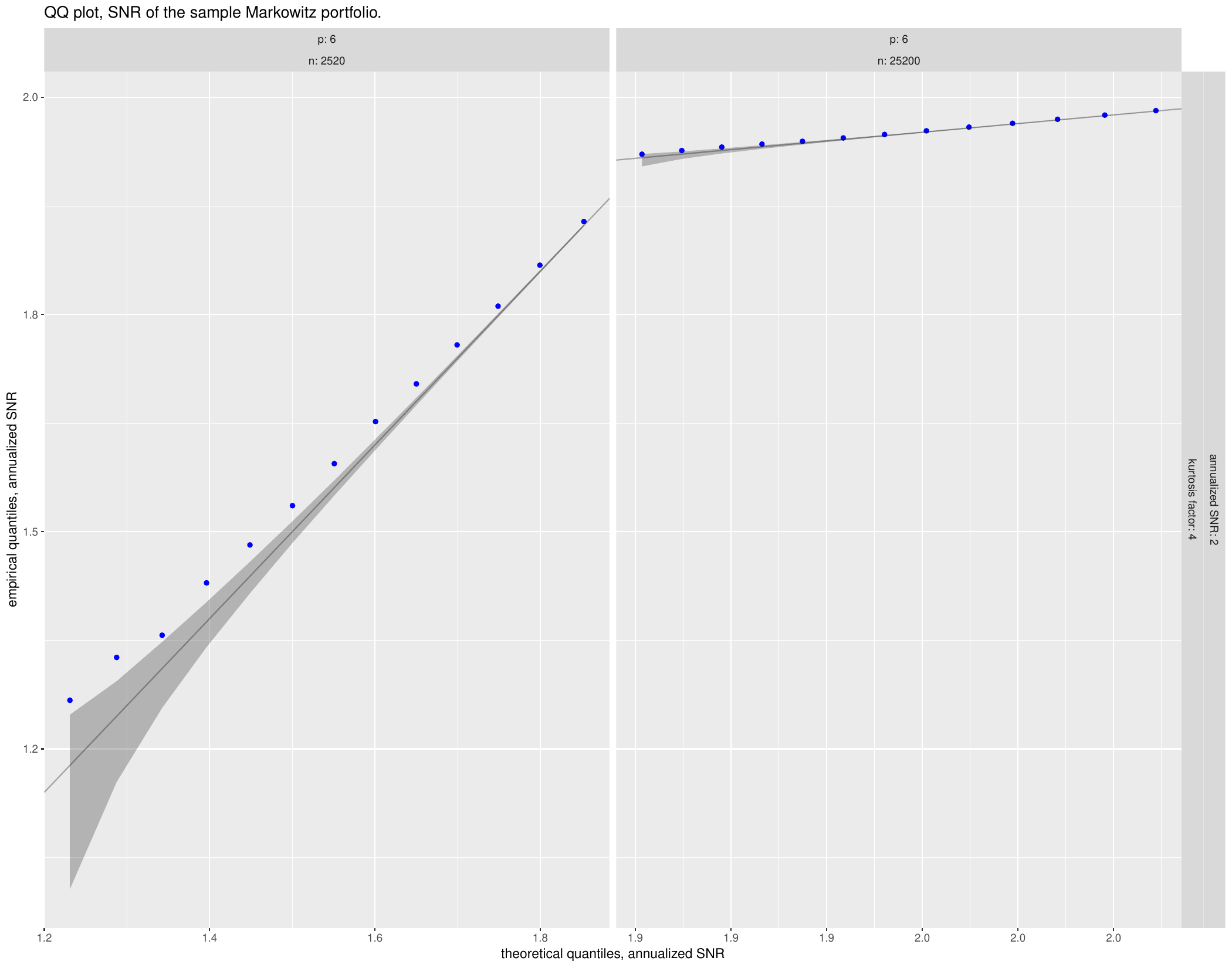} \caption[Q-Q plots of \pSNR{\sportwopt} are given versus the theoretical values from \eqnref{eta_form_psnr_mp}, for $\nlatf=6$, $\kurty=4$, $\psnropt=2$ and for 10 and 100 years of daily data at a rate of 252 days per year]{Q-Q plots of \pSNR{\sportwopt} are given versus the theoretical values from \eqnref{eta_form_psnr_mp}, for $\nlatf=6$, $\kurty=4$, $\psnropt=2$ and for 10 and 100 years of daily data at a rate of 252 days per year. We plot only selected theoretical quantiles and the corresponding empirical quantiles, with 95\% confidence bands. Note that the two facets have different $x$ limits and different aspect ratios. The approximation appears to underestimate the true \txtSNR for small sample sizes. }\label{fig:snr_qq_plots_for_100}
\end{figure}

\end{knitrout}

\subsection{Fama French Three Factor Portfolios}

The monthly returns of the Fama French 3 factor portfolios
from Jul 1926 to Jul 2013
were downloaded from \emph{Quandl}. \cite{Fama_French_1992,Quandl}
The returns excess the risk-free rate (which is also given in
the same data) are computed. The procedure of Britten-Jones
is applied to get marginal \tstat statistics for each of the three
assets.  The marginal Wald statistics are also computed, first using
the vanilla estimator of \pvvar, then using a robust (HAC) estimator
via the errors via the \Rpackage{sandwich}
package.  \cite{Zeileis:2004:JSSOBK:v11i10} 
These are presented in \tabref{simple_ff3}.
The Wald statistics are slightly less optimistic than the
Britten-Jones \tstat-statistics for the long MKT and short SMB
positions. This is amplified when the HAC estimator is used.

\begin{table}[ht]
\centering
\begin{tabular}{rrrr}
  \hline
 & MKT & HML & SMB \\ 
  \hline
Britten.Jones & 4.10 & 0.30 & -1.97 \\ 
   \hline
Wald & 3.86 & 0.31 & -1.92 \\ 
  Wald.HAC & 3.51 & 0.27 & -1.78 \\ 
   \hline
\end{tabular}
\caption{The marginal \tstat statistics of Britten-Jones \cite{BrittenJones1999} procedure, along with the Wald statistics from plugging in \svsm for \pvsm in \theoremref{inv_distribution}, with `vanilla' and HAC estimators of \pvvar, are shown for 1045 months of excess returns of the three Fama-French portfolios.} 
\label{tab:simple_ff3}
\end{table}

\subsubsection{Incorporating conditional heteroskedasticity}

A rolling estimate of general market volatility is computed
by taking the 11 month FIR mean of the median absolute return of 
the three portfolios, delayed by one month. The model
of `constant maximal \txtSR' (\ie \eqnref{cond_model_I}) is
assumed, where \fvola[i] is the inverse
of this estimated volatility. This is equivalent to 
dividing the returns by the estimated volatility, then applying
the unconditional estimator.

The marginal Wald statistics are presented in \tabref{wt_ff3},
and are more `confident' in the long MKT position, and short
SMB position, with little evidence to support a long or short
position in HML.

\begin{table}[ht]
\centering
\begin{tabular}{rrrr}
  \hline
 & MKT & HML & SMB \\ 
  \hline
HAC.wald.wt & 4.51 & 0.62 & -2.93 \\ 
   \hline
\end{tabular}
\caption{The marginal Wald statistics computed from plugging in \svsm for \pvsm in \theoremref{inv_distribution}, with a HAC estimator of \pvvar, are shown for 1045 months of excess returns of the three Fama-French portfolios.  To adjust for heteroskedasticity, returns are divided by a lagged volatility estimate based on 11 previous months of returns of the three portfolios.} 
\label{tab:wt_ff3}
\end{table}


\subsubsection{Conditional expectation}

The Shiller cyclically adjusted price/earnings data (CAPE) are
downloaded from \emph{Quandl}. \cite{Quandl} The CAPE data are
delayed by a month so that they qualify as features which could
be used in the investment decision. That is, we are testing a
model where CAPE are used to predict returns, not 
`explain' them, contemporaneously. The CAPE data are centered
by subtracting the mean. The marginal Wald statistics, 
computed using a HAC estimator for \pvvar, for the
6 elements of the Markowitz coefficient matrix are presented in 
\tabref{wald_cape_ff}, and indicate a significant 
unconditional long MKT position; when CAPE is above the long term 
average value of 17.54, decreasing the position in 
MKT is warranted. 


\begin{table}[ht]
\centering
\begin{tabular}{rrrr}
  \hline
 & MKT & HML & SMB \\ 
  \hline
Intercept & 2.22 & 0.59 & -1.33 \\ 
  CAPE & -2.46 & -1.13 & -0.70 \\ 
   \hline
\end{tabular}
\caption{The marginal Wald statistics of the Markowitz coefficient computed from plugging in \svsm for \pvsm in \theoremref{cond_inv_distribution_II}, with a HAC estimator of \pvvar, are shown for 1045 months of excess returns of the three Fama-French portfolios.  The Shiller CAPE data are delayed by a month, and centered. No adjustments for conditional heteroskedasticity are performed.} 
\label{tab:wald_cape_ff}
\end{table}

\begin{table}[ht]
\centering
\begin{tabular}{rrrr}
  \hline
 & MKT & HML & SMB \\ 
  \hline
Intercept & 3.49 & 0.27 & -1.77 \\ 
  del.CAPE & 1.19 & 0.03 & 2.52 \\ 
   \hline
\end{tabular}
\caption{The marginal Wald statistics of the Markowitz coefficient computed from plugging in \svsm for \pvsm in \theoremref{cond_inv_distribution_II}, with a HAC estimator of \pvvar, are shown for 1045 months of excess returns of the three Fama-French portfolios.  The Shiller CAPE data are delayed by a month, and the first difference is computed. No adjustments for conditional heteroskedasticity are performed.} 
\label{tab:wald_dcape_ff}
\end{table}

The CAPE data change at a very low frequency. It is possible that the
changes in the CAPE data are predictive of future returns. The Wald
statistics of the Markowitz coefficient using the first difference in
monthly CAPE dataa, delayed by a month, are presented in 
\tabref{wald_dcape_ff}. These suggest a long unconditional position
in MKT, with the differences in CAPE providing a `timing' signal for
an unconditional short SMB position.


\subsubsection{Attribution of error}

\theoremref{inv_distribution} gives the asymptotic distribution
of $\minv{\svsm}$, which contains the (negative) \txtMP and the
precision matrix. This allows one to estimate the amount of error
in the \txtMP which is attributable to mis-estimation of the covariance.
The remainder one can attribute to mis-estimation of the mean vector, which,
is typically implicated as the leading effect. \cite{chopra1993effect}

The computation is performed as follows: the estimated covariance of
$\fvech{\minv{\svsm}}$ is turned into a correlation matrix in the usual
way\footnote{That is, by a Hadamard divide of the rank one matrix of the 
outer product of the diagonal. Or, more practically, by the \Rlang function
\Rfunction{cov2cor}.}. This gives a correlation matrix, call it \Mtx{R},
some of the elements of which correspond to the negative \txtMP, and
some to the precision matrix. For a single element of the \txtMP, let 
\vect{r} be a sub-column of \Mtx{R} consisting of the column corresponding
to that element of the \txtMP and the rows of the precision matrix. And
let $\Mtx{R}_{\pvsig}$ be the sub-matrix of \Mtx{R} corresponding to the
precision matrix.  The multiple correlation coefficient is then
$\qform{\Mtx{R}_{\pvsig}^{-1}}{\vect{r}}$. This is an `R-squared' number
between zero and one, estimating the proportion of variance in that
element of the \txtMP `explained' by error in the precision matrix.

\begin{table}[ht]
\centering
\begin{tabular}{rrr}
  \hline
 & vanilla & weighted \\ 
  \hline
MKT & 41 \% & 32 \% \\ 
  HML & 11 \% & 6.8 \% \\ 
  SMB & 29 \% & 13 \% \\ 
   \hline
\end{tabular}
\caption{Estimated multiple correlation coefficients for the three elements of the \txtMP are presented. These are the percentage of error attributable to mis-estimation of the precision matrix.  HAC estimators for \pvvar are used for both.  In the `vanilla' column, no adjustments are made for conditional heteroskedasticity, while in the `weighted' column, returns are divided by the volatility estimate.} 
\label{tab:ffret_multcor}
\end{table}

Here, for each of the members of the vanilla \txtMP on the 
3 assets, this squared
coefficient of multiple correlation, is expressed as percents 
in \tabref{ffret_multcor}. A HAC estimator for \pvvar is used.
In the `weighted' column, the returns are divided by the 
rolling estimate of volatility described above, assuming a 
model of `constant maximal \txtSR'.
We can claim, then, that approximately 41 percent of the 
error in the MKT position is due to mis-estimation of the precision
matrix. 







\nocite{markowitz1952portfolio,markowitz1999early,markowitz2012foundations}
\bibliographystyle{plainnat}
\bibliography{SharpeR,rauto}

\appendix

\section{Matrix Derivatives}

\begin{lemma}[Derivatives]
\label{lemma:more_misc_derivs}
Given conformable, symmetric, matrices \Mtx{X}, \Mtx{Y}, \Mtx{Z},
and constant matrix \Mtx{J}, define
\begin{equation*}
\fdinvwrap{\Mtx{J}}{\Mtx{X}} 
\defeq - \wrapParens{\AkronA{\minvParens{\qform{\Mtx{X}}{\Mtx{J}}}}}\wrapParens{\AkronA{\tr{\Mtx{J}}}}.
\end{equation*}
Then
\begin{align}
\label{eqn:mtx_qinv_rule}
\dbyd{\minvParens{\qform{\Mtx{X}}{\Mtx{J}}}}{\Mtx{Z}} 
&= \fdinvwrap{\Mtx{J}}{\Mtx{X}}\dbyd{\Mtx{X}}{\Mtx{Z}}.&\\
\label{eqn:mtx_prod_rule}
\dbyd{\Mtx{X}\Mtx{Y}}{\Mtx{Z}} &=
\wrapParens{\eye\kron\Mtx{X}}\dbyd{\Mtx{Y}}{\Mtx{Z}} +
\wrapParens{\tr{\Mtx{Y}}\kron\eye}\dbyd{\Mtx{X}}{\Mtx{Z}}.&\\
\label{eqn:mtx_ogram_rule}
\dbyd{\ogram{\Mtx{X}}}{\Mtx{Z}} &=
\wrapParens{\eye + \Komm}\wrapParens{\Mtx{X} \kron
\eye}\dbyd{\Mtx{X}}{\Mtx{Z}}.&\\
\label{eqn:mtx_trace_deriv}
\dbyd{\trace{\Mtx{X}\Mtx{Y}}}{\Mtx{Z}} &=
\tr{\fvec{\tr{\Mtx{X}}}}\dbyd{\Mtx{Y}}{\Mtx{Z}} + 
\tr{\fvec{\Mtx{Y}}}\dbyd{\tr{\Mtx{X}}}{\Mtx{Z}}.&\\
\label{eqn:mtx_det_deriv}
\dbyd{\det{\Mtx{X}}}{\Mtx{Z}} &=
\det{\Mtx{X}}
\tr{\fvec{\trminv{\Mtx{X}}}}\dbyd{\Mtx{X}}{\Mtx{Z}}.&\\
\label{eqn:mtx_prod_det_deriv}
\dbyd{\det{\Mtx{X}\Mtx{Y}}}{\Mtx{Z}} &=
\det{\Mtx{X}\Mtx{Y}}
\wrapParens{ \tr{\fvec{\trminv{\Mtx{X}}}}\dbyd{\Mtx{X}}{\Mtx{Z}} +
\tr{\fvec{\trminv{\Mtx{Y}}}}\dbyd{\Mtx{Y}}{\Mtx{Z}}}.&\\
\label{eqn:mtx_inv_prod_det_deriv}
\dbyd{\det{\minvParens{\Mtx{X}\Mtx{Y}}}}{\Mtx{Z}} &=
\det{\Mtx{X}\Mtx{Y}}^{-1}
\wrapParens{ \tr{\fvec{\trminv{\Mtx{X}}}}\dbyd{\Mtx{X}}{\Mtx{Z}} +
\tr{\fvec{\trminv{\Mtx{Y}}}}\dbyd{\Mtx{Y}}{\Mtx{Z}}}.&
\end{align}
Here \Komm is the 'commutation matrix.'

Let \eigval[j] be the \kth{j} eigenvalue of \Mtx{X}, with corresponding
eigenvector \eigvec[j], normalized so that $\gram{\eigvec[j]} = 1$. Then
\begin{align}
\label{eqn:mtx_eig_deriv}
\dbyd{\eigval[j]}{\Mtx{Z}} &=
\wrapParens{\AkronA{\tr{\eigvec[j]}}}\dbyd{\Mtx{X}}{\Mtx{Z}}.&
\end{align}
\end{lemma}
\begin{proof}
For \eqnref{mtx_qinv_rule}, write
$$
\dbyd{\minvParens{\qform{\Mtx{X}}{\Mtx{J}}}}{\Mtx{Z}} =
\dbyd{\minvParens{\qform{\Mtx{X}}{\Mtx{J}}}}{\wrapParens{\qform{\Mtx{X}}{\Mtx{J}}}}
\dbyd{\wrapParens{\qform{\Mtx{X}}{\Mtx{J}}}}{\Mtx{Z}}.
$$
\lemmaref{deriv_vech_matrix_inverse} gives the derivative on the left;
to get the derivative on the right, note that 
$\fvec{\qform{\Mtx{X}}{\Mtx{J}}} =
\wrapParens{\AkronA{\tr{\Mtx{J}}}}\fvec{\Mtx{X}}$, then use linearity of
the derivative.

For \eqnref{mtx_prod_rule}, write $\fvec{\Mtx{X}\Mtx{Y}} =
\wrapParens{\tr{\Mtx{Y}}\kron\Mtx{X}}\fvec{\eye}$. Then consider
the derivative of $\fvec{\Mtx{X}\Mtx{Y}}$ with respect to any \emph{scalar}
$z$ :
\begin{equation*}
\dbyd{\fvec{\Mtx{X}\Mtx{Y}}}{z} 
= \dbyd{\wrapParens{\tr{\Mtx{Y}}\kron\Mtx{X}}\fvec{\eye}}{z}
= \fivec{\dbyd{\wrapParens{\tr{\Mtx{Y}}\kron\Mtx{X}}}{z}} \fvec{\eye},
\end{equation*}
where $\fivec{\cdot}$ is the inverse of $\fvec{\cdot}$. That is,
$\fivec{\fvec{\cdot}}$ is the identity over 
square matrices. (This wrinkle is needed because we have defined
derivatives of matrices to be the derivative of their vectorization.)

For \eqnref{mtx_ogram_rule}, by \eqnref{mtx_prod_rule},
\begin{equation*}
\begin{split}
\dbyd{\fvec{\ogram{\Mtx{X}}}}{\fvec{\Mtx{Z}}} &= 
\wrapParens{\eye\kron\Mtx{X}}\dbyd{\fvec{\tr{\Mtx{X}}}}{\fvec{\Mtx{Z}}} +
\wrapParens{{\Mtx{X}}\kron\eye}\dbyd{\fvec{\Mtx{X}}}{\fvec{\Mtx{Z}}},\\
&= \wrapParens{\eye\kron\Mtx{X}}\Komm \dbyd{\fvec{\Mtx{X}}}{\fvec{\Mtx{Z}}} + 
\wrapParens{{\Mtx{X}}\kron\eye}\dbyd{\fvec{\Mtx{X}}}{\fvec{\Mtx{Z}}}.\\
\end{split}
\end{equation*}
Now let \Mtx{A} be any conformable square matrix. We have:
\begin{multline*}
\wrapParens{\eye\kron\Mtx{X}}\Komm\fvec{\Mtx{A}}
= \wrapParens{\eye\kron\Mtx{X}}\fvec{\tr{\Mtx{A}}}
= \fvec{\Mtx{X}\tr{\Mtx{A}}}
= \\
\Komm \fvec{\Mtx{A}\tr{\Mtx{X}}}
= \Komm \wrapParens{\Mtx{X}\kron\eye} \fvec{\Mtx{A}}.
\end{multline*}
Because \Mtx{A} was arbitrary, we have
$\wrapParens{\eye\kron\Mtx{X}}\Komm = \Komm\wrapParens{\Mtx{X}\kron\eye},$
and the result follows.

Using the product rule for Kronecker products \cite{petersen2012matrix},
then using the vector identity again we have
\begin{equation*}
\begin{split}
\dbyd{\fvec{\Mtx{X}\Mtx{Y}}}{z} 
&= \wrapParens{\fivec{\dbyd{\tr{\Mtx{Y}}}{z}}\kron\Mtx{X} +
\tr{\Mtx{Y}} \kron \fivec{\dbyd{\Mtx{X}}{z}}}\fvec{\eye},\\
&= 
\wrapParens{\eye\kron\Mtx{X}} \dbyd{{\Mtx{Y}}}{z}  +
\wrapParens{\tr{\Mtx{Y}}\kron\eye}\dbyd{\Mtx{X}}{z}.
\end{split}
\end{equation*}
Then apply this result to every element of \fvec{\Mtx{Z}} to get
the result.

For \eqnref{mtx_trace_deriv}, write
$$
\trace{\Mtx{X}\Mtx{Y}} = \tr{\fvec{\tr{\Mtx{X}}}}\fvec{\Mtx{Y}},
$$
then use the product rule.

For \eqnref{mtx_det_deriv}, first consider the derivative of 
$\det{\Mtx{X}}$ with respect to a scalar $z$. This is known
to take form:  \cite{petersen2012matrix}
$$
\dbyd{\det{\Mtx{X}}}{z} =
\det{X}\trace{\minv{\Mtx{X}}\fivec{\dbyd{\Mtx{X}}{z}}},
$$
where the $\fivec{\cdot}$ is here because of how we 
have defined derivatives of matrices. Rewrite the trace
as the dot product of two vectors:
$$
\dbyd{\det{\Mtx{X}}}{z} =
\det{X}\tr{\fvec{\trminv{\Mtx{X}}}} \dbyd{\Mtx{X}}{z}.
$$
Using this to compute the derivative with respect to
each element of $\fvec{\Mtx{Z}}$ gives the result.
\eqnref{mtx_prod_det_deriv} follows from the scalar product
rule since $\det{\Mtx{X}\Mtx{Y}} =
\det{\Mtx{X}}\det{\Mtx{Y}}$.
\eqnref{mtx_inv_prod_det_deriv} then follows, using the scalar chain
rule.

For \eqnref{mtx_eig_deriv}, the derivative of the \kth{j} eigenvalue
of matrix \Mtx{X} with respect to a scalar $z$ is known 
to be: \cite[equation (67)]{petersen2012matrix}
$$
\dbyd{\eigval[j]}{z} = \qform{\fivec{\dbyd{\Mtx{X}}{z}}}{\eigvec[j]}.
$$
Take the vectorization of this scalar, and rewrite it in Kronecker
form:
$$
\dbyd{\eigval[j]}{z} = \wrapParens{\AkronA{\tr{\eigvec[j]}}}\dbyd{\Mtx{X}}{z}.
$$
Use this to compute the derivative of \eigval[j] with respect to each
element of $\fvec{\Mtx{Z}}$.

\end{proof}

\begin{lemma}[Cholesky Derivatives]
\label{lemma:cholesky_deriv}
Let \Mtx{X} be a symmetric positive definite matrix. Let 
\Mtx{Y} be its lower triangular Cholesky factor. That is, 
\Mtx{Y} is the lower triangular matrix such that $\ogram{\Mtx{Y}} = \Mtx{X}$.
Then 
\begin{align}
\label{eqn:mtx_chol_rule}
\dbyd{\fvech{\Mtx{Y}}}{\fvech{\Mtx{X}}} &= \minv{\wrapParens{%
\qoform{\wrapParens{\eye + \Komm}\wrapParens{\Mtx{Y} \kron \eye}}{\Elim}}},&
\end{align}
where \Komm is the 'commutation matrix'.
\cite{magnus1999matrix}
\end{lemma}
\begin{proof}
By \eqnref{mtx_ogram_rule} of \lemmaref{more_misc_derivs}, 
\begin{equation*}
\dbyd{\fvec{\ogram{\Mtx{Y}}}}{\fvec{\Mtx{Y}}} = 
\wrapParens{\eye + \Komm}\wrapParens{\Mtx{Y} \kron \eye}.
\end{equation*}
By the chain rule, for lower triangular matrix \Mtx{Y}, we have
\begin{equation*}
\begin{split}
\dbyd{\fvech{\ogram{\Mtx{Y}}}}{\fvech{\Mtx{Y}}} &= 
\dbyd{\fvech{\ogram{\Mtx{Y}}}}{\fvec{\ogram{\Mtx{Y}}}} 
\dbyd{\fvec{\ogram{\Mtx{Y}}}}{\fvec{\Mtx{Y}}} 
\dbyd{\fvec{\Mtx{Y}}}{\fvech{\Mtx{Y}}},\\
&= \Elim \wrapParens{\eye + \Komm}\wrapParens{\Mtx{Y} \kron \eye}
\tr{\Elim}.
\end{split}
\end{equation*}

The result now follows since
\begin{equation*}
\dbyd{\fvech{\Mtx{Y}}}{\fvech{\Mtx{X}}} 
= \dbyd{\fvech{\Mtx{Y}}}{\fvech{\ogram{\Mtx{Y}}}} 
= \minv{\wrapParens{\dbyd{\fvech{\ogram{\Mtx{Y}}}}{\fvech{\Mtx{Y}}}}}. 
\end{equation*}
\end{proof}


\section{Proofs}

Here we give proofs of \theoremref{theta_covar_elliptical} and
\corollaryref{invtheta_asym_var_elliptical_explicit}.

\begin{proof}[Proof of \theoremref{theta_covar_elliptical}]
We note that in the Gaussian case ($\kurty=1$) this result is proved by Magnus and Neudecker, up to a rearrangement in the terms \cite[Theorem 4.4 (ii)]{magnus1979commutation}.

First we suppose that the first element of \avreti, rather than
being a deterministic 1, is a random variable with mean $1$, 
no covariance with the elements of \reti, and a variance of $\epsilon$. 
We assume the random first element is such that \avreti is elliptically distributed
with mean \apvmu and covariance \apvsig.
After finding the variance of \ogram{\avreti}, we will take 
$\epsilon \to 0.$ 
Below we will use \pvsm to mean $\apvsig + \ogram{\apvmu}$, which converges
to our usual definition as $\epsilon \to 0$.

\providecommand{\mycent}[1]{\avreti[#1] - \apvmu[#1]}
\providecommand{\wmycent}[1]{\wrapParens{\mycent{#1}}}

An extension of Isserlis' theorem gives the moments of centered
elements of \avreti.  \cite{vignat2007extension,KAN2008542}
In particular, the first four centered moments are
\begin{align}
	\E{\mycent{i}} &= 0,\nonumber\\
	\E{\wmycent{i}\wmycent{j}} &= \apvsig[i,j],\nonumber\\
	\E{\wmycent{i}\wmycent{j}\wmycent{k}} &= 0,\nonumber\\
	\E{\wmycent{i}\wmycent{j}\wmycent{k}\wmycent{l}} &=
	\kurty\wrapBracks{\apvsig[i,j]\apvsig[k,l] + \apvsig[i,k]\apvsig[j,l] + 
	\apvsig[i,l]\apvsig[j,k]}.\nonumber
\end{align}
It is a tedious exercise to compute the raw uncentered moments:
\begin{align}
	\E{\avreti[i]} &= \apvmu[i],\nonumber\\
	\E{\avreti[i]\avreti[j]} &= \pvsm[i,j],\nonumber\\
	\E{\avreti[i]\avreti[j]\avreti[k]} &= \apvmu[i]\apvmu[j]\apvmu[k] 
	+ \apvmu[i]\apvsig[j,k] 
	+ \apvmu[j]\apvsig[i,k] 
	+ \apvmu[k]\apvsig[i,j],\nonumber\\
	\E{\avreti[i]\avreti[j]\avreti[k]\avreti[l]} &=
	\kurty\wrapBracks{\apvsig[i,j]\apvsig[k,l] + \apvsig[i,k]\apvsig[j,l] + 
	\apvsig[i,l]\apvsig[j,k]}\nonumber\\
	&\phantom{=}\,
	+\apvmu[i]\apvmu[j] \apvsig[k,l] 
	+\apvmu[i]\apvmu[k] \apvsig[j,l] 
	+\apvmu[i]\apvmu[l] \apvsig[j,k]\nonumber\\ 
	&\phantom{=}\,
	+\apvmu[j]\apvmu[k] \apvsig[i,l] 
	+\apvmu[j]\apvmu[l] \apvsig[i,k] 
	+\apvmu[k]\apvmu[l] \apvsig[i,j] \nonumber\\
	&\phantom{=}\,
	+\apvmu[i]\apvmu[j]\apvmu[k]\apvmu[l].\nonumber
\end{align}

We now want to compute the covariance of $\avreti[i]\avreti[j]$ with
$\avreti[k]\avreti[l]$.  We have
\begin{align*}
\COV{\avreti[i]\avreti[j]}{\avreti[k]\avreti[l]} 
&= \E{\avreti[i]\avreti[j]\avreti[k]\avreti[l]} -
	\E{\avreti[i]\avreti[j]}\E{\avreti[k]\avreti[l]},\\
&=\kurty\wrapBracks{\apvsig[i,j]\apvsig[k,l] + \apvsig[i,k]\apvsig[j,l] + 
	\apvsig[i,l]\apvsig[j,k]}\nonumber\\
	&\phantom{=}\,
	+\apvmu[i]\apvmu[j] \apvsig[k,l] 
	+\apvmu[i]\apvmu[k] \apvsig[j,l] 
	+\apvmu[i]\apvmu[l] \apvsig[j,k]\nonumber\\ 
	&\phantom{=}\,
	+\apvmu[j]\apvmu[k] \apvsig[i,l] 
	+\apvmu[j]\apvmu[l] \apvsig[i,k] 
	+\apvmu[k]\apvmu[l] \apvsig[i,j] \nonumber\\
	&\phantom{=}\,
	+\apvmu[i]\apvmu[j]\apvmu[k]\apvmu[l].\nonumber\\
	&\phantom{=}\,
	-\wrapParens{\apvmu[i]\apvmu[j] + \apvsig[i,j]}
	\wrapParens{\apvmu[k]\apvmu[l] + \apvsig[k,l]},\\
&=\kurty\wrapBracks{\apvsig[i,j]\apvsig[k,l] + \apvsig[i,k]\apvsig[j,l] + 
	\apvsig[i,l]\apvsig[j,k]}\nonumber\\
	&\phantom{=}\,
	+\apvmu[i]\apvmu[k] \apvsig[j,l] 
	+\apvmu[i]\apvmu[l] \apvsig[j,k]\nonumber\\ 
	&\phantom{=}\,
	+\apvmu[j]\apvmu[k] \apvsig[i,l] 
	+\apvmu[j]\apvmu[l] \apvsig[i,k] \nonumber\\
	&\phantom{=}\,
	-\apvsig[i,j]\apvsig[k,l],\\
&=\wrapParens{\kurty-1}\wrapBracks{\apvsig[i,j]\apvsig[k,l] + \apvsig[i,k]\apvsig[j,l] + 
	\apvsig[i,l]\apvsig[j,k]}\nonumber\\
	&\phantom{=}\,
	+\apvsig[i,k]\apvsig[j,l] 
	+\apvsig[i,l]\apvsig[j,k]\nonumber\\
	&\phantom{=}\,
	+\apvmu[i]\apvmu[k] \apvsig[j,l] 
	+\apvmu[i]\apvmu[l] \apvsig[j,k]\nonumber\\ 
	&\phantom{=}\,
	+\apvmu[j]\apvmu[k] \apvsig[i,l] 
	+\apvmu[j]\apvmu[l] \apvsig[i,k],\nonumber\\
&=\wrapParens{\kurty-1}\wrapBracks{\apvsig[i,j]\apvsig[k,l] + \apvsig[i,k]\apvsig[j,l] + 
	\apvsig[i,l]\apvsig[j,k]}\nonumber\\
	&\phantom{=}\,
	+\svsm[i,k]\apvsig[j,l] 
	+\svsm[i,l]\apvsig[j,k]\nonumber\\
	&\phantom{=}\,
	+\apvmu[j]\apvmu[k] \apvsig[i,l] 
	+\apvmu[j]\apvmu[l] \apvsig[i,k],\nonumber\\
&=\wrapParens{\kurty-1}\wrapBracks{\apvsig[i,j]\apvsig[k,l] + \apvsig[i,k]\apvsig[j,l] + 
	\apvsig[i,l]\apvsig[j,k]}\nonumber\\
	&\phantom{=}\,
	+\svsm[i,k]\apvsig[j,l] 
	+\svsm[i,l]\apvsig[j,k]\nonumber\\
	&\phantom{=}\,
	+\apvmu[j]\apvmu[k] \svsm[i,l] - \apvmu[i]\apvmu[j]\apvmu[k]\apvmu[l]
	+\apvmu[j]\apvmu[l] \svsm[i,k] - \apvmu[i]\apvmu[j]\apvmu[k]\apvmu[l],\nonumber\\
&=\wrapParens{\kurty-1}\wrapBracks{\apvsig[i,j]\apvsig[k,l] + \apvsig[i,k]\apvsig[j,l] + 
	\apvsig[i,l]\apvsig[j,k]}\nonumber\\
	&\phantom{=}\,
	+\svsm[i,k]\svsm[j,l] 
	+\svsm[i,l]\svsm[j,k] - 2 \apvmu[i]\apvmu[j]\apvmu[k]\apvmu[l].\nonumber
\end{align*}

Now we need only translate this scalar result into the vector result in the
theorem.  If \avreti is $m$-dimensional, then the variance-covariance matrix
of $\fvec{\ogram{\avreti}}$ is $\sbby{m^2}$ whose \kth{ij,kl} element is given
above. The term \apvsig[i,j]\apvsig[k,l] is the \kth{ij,kl} element of
$\ogram{\fvec{\apvsig}}$. 
The term \apvsig[i,k]\apvsig[j,l] is the \kth{ij,kl} element of
\AkronA{\apvsig}.
The term \apvsig[j,k]\apvsig[i,l] is the \kth{ji,kl} element of
\AkronA{\apvsig}, and thus is the \kth{ij,kl} element of
$\Komm\wrapParens{\AkronA{\apvsig}}$. We can similarly identify
the terms $\svsm[i,k]\svsm[j,l]$, $\svsm[i,l]\svsm[j,k]$, and
$\apvmu[i]\apvmu[j]\apvmu[k]\apvmu[l]$, and thus
\begin{align*}
	\VAR{\fvec{\ogram{\avreti}}}_{ij,kl} 
	&= \COV{\avreti[i]\avreti[j]}{\avreti[k]\avreti[l]},\\
	&=\wrapParens{\kurty-1}\wrapBracks{\ogram{\fvec{\apvsig}}
	+ \wrapParens{\eye + \Komm}\AkronA{\apvsig}}_{ij,kl}\nonumber\\
	&\phantom{=}\,
	+ \wrapParens{\eye + \Komm}\wrapParens{\AkronA{\svsm} -
	\AkronA{\ogram{\apvmu}}}_{ij,kl}.\nonumber
\end{align*}

\end{proof}

\begin{proof}[Proof of \corollaryref{invtheta_asym_var_elliptical_explicit}]
From the previous corollary, it suffices to prove the identity of 
\Mtx{B}.  Define $\Simm = \half \wrapParens{\eye + \Komm}$. 
Then note that \eqnref{theta_covar_elliptical} becomes
\begin{align*}
	\pvvar[0] &= \wrapParens{\kurty-1}
	\wrapBracks{\Simm \fvec{\apvsig}\kron\tr{\fvec{\apvsig}} + 2\Simm\AkronA{\apvsig}}\\
	&\phantom{=}\,+ 2\Simm \wrapBracks{\AkronA{\pvsm} -
	\AkronA{\ogram{\apvmu}}}.
\end{align*}
By Lemma 3.5 of Magnus and Neudecker, note that $\Dupp\Elim\Simm = \Simm$.  \cite{magnus1980elimination}
So consider
\begin{align*}
\Mtx{B} &= \qoform{{\qoform{\pvvar[0]}{\Elim}}}{\wrapParens{\EXD{\wrapParens{\AkronA{\minv{\pvsm}}}}}},\\
&= \wrapParens{\kurty-1}\Elim\wrapParens{\AkronA{\minv{\pvsm}}} \Simm 
		\wrapBracks{\ogram{\fvec{\apvsig}}}
		\tr{\Elim}
		\tr{\wrapParens{\EXD{\wrapParens{\AkronA{\minv{\pvsm}}}}}}\\
		&\phantom{=}\,+ 2\wrapParens{\kurty-1}\Elim\wrapParens{\AkronA{\minv{\pvsm}}} \Simm 
		\wrapBracks{\AkronA{\apvsig}}
		\tr{\Elim}
		\tr{\wrapParens{\EXD{\wrapParens{\AkronA{\minv{\pvsm}}}}}}\\
		&\phantom{=}\,+ 2\Elim\wrapParens{\AkronA{\minv{\pvsm}}} \Simm 
		\wrapBracks{\AkronA{\pvsm} - \AkronA{\ogram{\apvmu}}}
		\tr{\Elim}
		\tr{\wrapParens{\EXD{\wrapParens{\AkronA{\minv{\pvsm}}}}}}.
\end{align*}
Now by Lemma 2.1 of Magnus and Neudecker, 
$\Simm\wrapParens{\AkronA{\Mtx{A}}} = \wrapParens{\AkronA{\Mtx{A}}} \Simm = \Simm \wrapParens{\AkronA{\Mtx{A}}} \Simm$,
so we can slide the \Simm matrix around. Also note that because
\apvsig is symmetric, $\Simm \ogram{\fvec{\apvsig}} = \ogram{\fvec{\apvsig}}$.
Then
\begin{align*}
\Mtx{B} &= \wrapParens{\kurty-1}\Elim\wrapParens{\AkronA{\minv{\pvsm}}} 
		\wrapBracks{\ogram{\fvec{\apvsig}}}
		\tr{\Elim}
		\tr{\wrapParens{\EXD{\wrapParens{\AkronA{\minv{\pvsm}}}}}}\\
		&\phantom{=}\,+ 2\wrapParens{\kurty-1}\Elim\Simm\wrapParens{\AkronA{\minv{\pvsm}}}
		\wrapBracks{\AkronA{\apvsig}}
		\tr{\Elim}
		\tr{\wrapParens{\EXD{\wrapParens{\AkronA{\minv{\pvsm}}}}}}\\
		&\phantom{=}\,+ 2\Elim\Simm\wrapParens{\AkronA{\minv{\pvsm}}}
		\wrapBracks{\AkronA{\pvsm} - \AkronA{\ogram{\apvmu}}}
		\tr{\Elim}
		\tr{\wrapParens{\EXD{\wrapParens{\AkronA{\minv{\pvsm}}}}}},\\
  &= \wrapParens{\kurty-1}\Elim\wrapBracks{\wrapParens{\AkronA{\minv{\pvsm}}} 
		\ogram{\fvec{\apvsig}}}
		\tr{\Simm}
		\tr{\Elim}
		\tr{\wrapParens{\EXD{\wrapParens{\AkronA{\minv{\pvsm}}}}}}\\
		&\phantom{=}\,+ 2\wrapParens{\kurty-1}\Elim
		\wrapBracks{\AkronA{\minv{\pvsm}\apvsig}}
		\tr{\Simm}
		\tr{\Elim}
		\tr{\wrapParens{\EXD{\wrapParens{\AkronA{\minv{\pvsm}}}}}}\\
		&\phantom{=}\,+ 2\Elim
		\wrapBracks{\eye - \AkronA{\basev[1]\tr{\apvmu}}}
		\tr{\Simm}
		\tr{\Elim}
		\tr{\wrapParens{\EXD{\wrapParens{\AkronA{\minv{\pvsm}}}}}},
\end{align*}
where we have slided the \Simm matrix to the left, and multiplied by
\AkronA{\minv{\pvsm}}. Now work on the right sides, shifting \Simm,
collapsing $\tr{\Simm}\tr{\Elim}\tr{\Dupp}$ to $\tr{\Simm}$, and
multiplying by \AkronA{\minv{\pvsm}} from the right:
\begin{align*}
	\Mtx{B}
	&= \wrapParens{\kurty-1}\Elim\wrapBracks{\fvec{\minv{\pvsm}\apvsig\minv{\pvsm}}\tr{\fvec{\apvsig}}}
  	\tr{\Simm}
		\tr{\Elim}
		\tr{\Dupp}
		\wrapParens{\AkronA{\minv{\pvsm}}}
		\tr{\Elim}\\
		&\phantom{=}\,+ 2\wrapParens{\kurty-1}\Elim\Simm
		\wrapBracks{\AkronA{\minv{\pvsm}\apvsig}}
		\tr{\Simm}
		\tr{\Elim}
		\tr{\Dupp}
		\wrapParens{\AkronA{\minv{\pvsm}}}
		\tr{\Elim}\\
		&\phantom{=}\,+ 2\Elim\Simm
		\wrapBracks{\eye - \AkronA{\basev[1]\tr{\apvmu}}}
		\tr{\Simm}
		\tr{\Elim}
		\tr{\Dupp}
		\wrapParens{\AkronA{\minv{\pvsm}}}
		\tr{\Elim},\\
	&= \wrapParens{\kurty-1}\Elim\wrapBracks{\fvec{\minv{\pvsm}\apvsig\minv{\pvsm}}\tr{\fvec{\apvsig}}}
		\wrapParens{\AkronA{\minv{\pvsm}}}
  	\tr{\Simm}
		\tr{\Elim}\\
		&\phantom{=}\,+ 2\wrapParens{\kurty-1}\Elim\Simm
		\wrapBracks{\AkronA{\minv{\pvsm}\apvsig}}
		\wrapParens{\AkronA{\minv{\pvsm}}}
		\tr{\Simm}
		\tr{\Elim}\\
		&\phantom{=}\,+ 2\Elim\Simm
		\wrapBracks{\eye - \AkronA{\basev[1]\tr{\apvmu}}}
		\wrapParens{\AkronA{\minv{\pvsm}}}
		\tr{\Simm}
		\tr{\Elim},\\
	&= \wrapParens{\kurty-1}\Elim\wrapBracks{\ogram{\fvec{\minv{\pvsm}\apvsig\minv{\pvsm}}}}
  	\tr{\Simm}
		\tr{\Elim}\\
		&\phantom{=}\,+ 2\wrapParens{\kurty-1}\Elim\Simm
		\wrapBracks{\AkronA{\minv{\pvsm}\apvsig\minv{\pvsm}}}
		\tr{\Simm}
		\tr{\Elim}\\
		&\phantom{=}\,+ 2\Elim\Simm
		\wrapBracks{\AkronA{\minv{\pvsm}} - \AkronA{\ogram{\basev[1]}}}
		\tr{\Simm}
		\tr{\Elim}.
\end{align*}
Now note that $\minv{\pvsm}\apvsig\minv{\pvsm} = \minv{\pvsm} -
\ogram{\basev[1]}$. Then
\begin{align*}
	\Mtx{B} &= \wrapParens{\kurty-1}\Elim\wrapBracks{\ogram{\fvec{\minv{\pvsm} -
	\ogram{\basev[1]}}}}
  	\tr{\Simm}
		\tr{\Elim}\\
		&\phantom{=}\,+ 2\wrapParens{\kurty-1}\Elim\Simm
		\wrapBracks{\AkronA{\wrapParens{\minv{\pvsm} - \ogram{\basev[1]}}}}
		\tr{\Simm}
		\tr{\Elim}\\
		&\phantom{=}\,+ 2\Elim\Simm
		\wrapBracks{\AkronA{\minv{\pvsm}} - \AkronA{\ogram{\basev[1]}}}
		\tr{\Simm}
		\tr{\Elim}.
\end{align*}
\end{proof}


\begin{proof}[Proof of \theoremref{mp_snr_ci_elliptical}]
Starting from \corollaryref{psnr_minus_ssropt_moments}, take
$\pvvar = \qoform{\pvvar[0]}{\Elim}$, with \pvvar[0] defined in
\theoremref{theta_covar_elliptical}. In 
\eqnref{shorter_HFH_form} we give the abbreviated form for
\qform{\Mtx{F}}{\Mtx{H}}.

Because of the symmetry of \pvvar[0], we can drop the \Elim and \Dupp terms
(but not the \Simm matrix). We decompose \pvvar into four terms based
on the terms of \pvvar[0].
\begin{align*}
\trace{ \qform{\Mtx{F}}{\Mtx{H}}\pvvar } 
	&= a_1 + a_2 + a_3 - a_4.
\end{align*}

Now we need some technical results. The first is that for conformable matrices
\begin{equation*}
	\trace{\wrapParens{\Mtx{A}\kron\Mtx{B}}\ogram{\fvec{\Mtx{X}}}} =
	\trace{\qform{\wrapParens{\Mtx{A}\kron\Mtx{B}}}{\fvec{\Mtx{X}}}} =
	\trace{\tr{\Mtx{X}}\Mtx{B}\Mtx{X}\tr{\Mtx{A}}}.
\end{equation*}
The other concerns traces with \Simm,
where $2\Simm = \eye + \Komm$, and 
where \Komm is the commutation matrix.
Let \Mtx{A} and \Mtx{B} be square matrices of the same size, then
\begin{equation}
\trace{2 \Simm \wrapParens{\Mtx{A} \kron \Mtx{B}}} =
\trace{\Mtx{A}}\trace{\Mtx{B}} + \trace{\Mtx{A}\Mtx{B}}.
\end{equation}
The proof is simple: since $\trace{\Mtx{A}\kron\Mtx{B}} =
\trace{\Mtx{A}}\trace{\Mtx{B}}$ is a known result, we need only focus on the 
\Komm term. Note then that
\begin{align*}
\trace{\Komm \wrapParens{\Mtx{A} \kron \Mtx{B}}} 
	&= \sum_{i,j} \wrapParens{\Komm \wrapParens{\Mtx{A} \kron \Mtx{B}}}_{ij,ij} = \sum_{i,j} \wrapParens{\wrapParens{\Mtx{A} \kron \Mtx{B}}}_{ji,ij},\\
	&= \sum_{i,j} \Mtx{B}_{j,i} \Mtx{A}_{i,j} = \sum_{i} \wrapParens{\Mtx{A}\Mtx{B}}_{i,i} = \trace{\Mtx{A}\Mtx{B}}.
\end{align*}
This is Magnus and Neudecker Theorem 3.1 (xiii). \cite{magnus1979commutation}
We will also implicitly use the fact that one can rotate terms in a trace
operator. We then tackle the trace terms one at a time.
\begin{align*}
	a_1 &= \wrapParens{\kurty-1} \trace{ \qform{\Mtx{F}}{\Mtx{H}}
	\ogram{\fvec{\apvsig}}},\\
	&= \frac{\kurty-1}{\psnropt} \trace{ \wrapParens{%
		\wrapParens{\minv{\pvsm}\ogram{\basev[1]}\minv{\pvsm}} \kron
		\twobytwo{0}{\tr{\vzero}}{\vzero}{\frac{\minv{\pvsig}\ogram{\pvmu}\minv{\pvsig}}{\psnrsqopt}
		- \minv{\pvsig}} } \ogram{\fvec{\apvsig}}},\\
	&= \frac{\kurty-1}{\psnropt} \trace{ \tr{\apvsig} 
		\twobytwo{0}{\tr{\vzero}}{\vzero}{\frac{\minv{\pvsig}\ogram{\pvmu}\minv{\pvsig}}{\psnrsqopt}
		- \minv{\pvsig}} \apvsig
		\wrapParens{\minv{\pvsm}\ogram{\basev[1]}\minv{\pvsm}}}.
\end{align*}
But note that
\begin{align*}
	\apvsig {\twobytwo{0}{\tr{\vzero}}{\vzero}{\frac{\minv{\pvsig}\ogram{\pvmu}\minv{\pvsig}}{\psnrsqopt}
		- \minv{\pvsig}} } \apvsig \minv{\pvsm}\basev[1]
		&= {\twobytwo{0}{\tr{\vzero}}{\vzero}{\frac{\ogram{\pvmu}}{\psnrsqopt}
		- \pvsig} } \minv{\pvsm}\basev[1],\\
		&= {\twobytwo{0}{\tr{\vzero}}{\vzero}{\frac{\ogram{\pvmu}\minv{\pvsig}}{\psnrsqopt}
		- \eye} }\basev[1] = \twobyone{0}{\vzero},
\end{align*}
thus $a_1 = 0$.

Now consider
\begin{align*}
	a_2 &= \wrapParens{\kurty-1} \trace{ \qform{\Mtx{F}}{\Mtx{H}}
	2 \Simm \wrapParens{\AkronA{\apvsig}} },\\
	&= \wrapParens{\kurty-1} \trace{ 2\Simm \wrapParens{\AkronA{\apvsig}}
	\qform{\Mtx{F}}{\Mtx{H}} }.
\end{align*}
Note that
\begin{align*}
	\wrapParens{\AkronA{\apvsig}} \qform{\Mtx{F}}{\Mtx{H}} 
	&= \oneby{\psnropt}\wrapParens{\AkronA{\apvsig}} \wrapParens{%
		\wrapParens{\minv{\pvsm}\ogram{\basev[1]}\minv{\pvsm}} \kron
		\twobytwo{0}{\tr{\vzero}}{\vzero}{\frac{\minv{\pvsig}\ogram{\pvmu}\minv{\pvsig}}{\psnrsqopt}
		- \minv{\pvsig}} },\\
	&= \oneby{\psnropt}\wrapParens{\apvsig \minv{\pvsm}\ogram{\basev[1]}\minv{\pvsm}} \kron
		\twobytwo{0}{\tr{\vzero}}{\vzero}{\frac{\ogram{\pvmu}\minv{\pvsig}}{\psnrsqopt}
		- \eye},\\
	&= \oneby{\psnropt}\wrapParens{\twobyone{0}{-\pvmu} \asrowvec{1+\psnrsqopt,-\pportwopt}} \kron
		\twobytwo{0}{\tr{\vzero}}{\vzero}{\frac{\ogram{\pvmu}\minv{\pvsig}}{\psnrsqopt}
		- \eye}.
\end{align*}
So
\begin{align*}
	a_2 &= \frac{\kurty-1}{\psnropt} \trace{ \twobyone{0}{-\pvmu}
	\asrowvec{1+\psnrsqopt,-\pportwopt}}\trace{\twobytwo{0}{\tr{\vzero}}{\vzero}{\frac{\ogram{\pvmu}\minv{\pvsig}}{\psnrsqopt}
		- \eye}},\\ 
	&\phantom{=}\,
	+ \frac{\kurty-1}{\psnropt} \trace{ \twobyone{0}{-\pvmu}
	\asrowvec{1+\psnrsqopt,-\pportwopt} \twobytwo{0}{\tr{\vzero}}{\vzero}{\frac{\ogram{\pvmu}\minv{\pvsig}}{\psnrsqopt}
		- \eye}},\\
	&= \frac{\kurty-1}{\psnropt} \psnrsqopt\wrapParens{1 - \nlatf}
	+ \frac{\kurty-1}{\psnropt} \trace{ \twobyone{0}{-\pvmu}
	\asrowvec{0,\tr{\vzero}} },\\
	&= \wrapParens{\kurty-1}\psnropt\wrapParens{1 - \nlatf}.
\end{align*}

Now consider
\begin{align*}
	a_3 &= \trace{ \qform{\Mtx{F}}{\Mtx{H}}
	2 \Simm \wrapParens{\AkronA{\pvsm}} }.
\end{align*}
Proceeding as before we have
\begin{align*}
	\wrapParens{\AkronA{\pvsm}} \qform{\Mtx{F}}{\Mtx{H}} 
	&=\oneby{\psnropt}\wrapParens{\AkronA{\pvsm}} \wrapParens{%
		\wrapParens{\minv{\pvsm}\ogram{\basev[1]}\minv{\pvsm}} \kron
		\twobytwo{0}{\tr{\vzero}}{\vzero}{\frac{\minv{\pvsig}\ogram{\pvmu}\minv{\pvsig}}{\psnrsqopt}
		- \minv{\pvsig}} },\\
	&= \oneby{\psnropt}\wrapParens{\basev[1] \asrowvec{1+\psnrsqopt,-\pportwopt}} \kron
		\twobytwo{0}{\tr{\vzero}}{\vzero}{\frac{\ogram{\pvmu}\minv{\pvsig}}{\psnrsqopt}
		- \eye}.
\end{align*}
Then 
\begin{align*}
	a_3 &= \oneby{\psnropt} \trace{ \basev[1] \asrowvec{1+\psnrsqopt,-\pportwopt}}\trace{\twobytwo{0}{\tr{\vzero}}{\vzero}{\frac{\ogram{\pvmu}\minv{\pvsig}}{\psnrsqopt}
		- \eye}},\\ 
	&\phantom{=}\,
	+ \oneby{\psnropt} \trace{ \basev[1]
	\asrowvec{1+\psnrsqopt,-\pportwopt} \twobytwo{0}{\tr{\vzero}}{\vzero}{\frac{\ogram{\pvmu}\minv{\pvsig}}{\psnrsqopt}
		- \eye}},\\
	&= \oneby{\psnropt}\wrapParens{1+\psnrsqopt}\wrapParens{1 - \nlatf}
	+ \oneby{\psnropt} \trace{ \basev[1] \asrowvec{0,\tr{\vzero}} },\\
	&= \oneby{\psnropt}\wrapParens{1+\psnrsqopt}\wrapParens{1 - \nlatf}.
\end{align*}

Finally 
\begin{align*}
	a_4 &= \trace{ \qform{\Mtx{F}}{\Mtx{H}}
	2 \Simm \wrapParens{\AkronA{\ogram{\apvmu}}} }.
\end{align*}
We have
\begin{align*}
	\wrapParens{\AkronA{\ogram{\apvmu}}} \qform{\Mtx{F}}{\Mtx{H}} 
	&=\oneby{\psnropt}
		\wrapParens{\ogram{\apvmu}\minv{\pvsm}\ogram{\basev[1]}\minv{\pvsm}} \kron
		\wrapParens{\ogram{\apvmu}\twobytwo{0}{\tr{\vzero}}{\vzero}{\frac{\minv{\pvsig}\ogram{\pvmu}\minv{\pvsig}}{\psnrsqopt} - \minv{\pvsig}}},\\
	&=\oneby{\psnropt}
		\wrapParens{\ogram{\apvmu}\minv{\pvsm}\ogram{\basev[1]}\minv{\pvsm}} \kron
		\mzero = \mzero,
\end{align*}
because $\trAB{\pvmu}{\wrapParens{\frac{\minv{\pvsig}\ogram{\pvmu}\minv{\pvsig}}{\psnrsqopt}
- \minv{\pvsig}}} = \vzero$. Thus $a_4 = 0$.

Putting them together we have
\begin{align*}
\trace{ \qform{\Mtx{F}}{\Mtx{H}}\pvvar[0] } &=
	\frac{\kurty \psnrsqopt + 1}{\psnropt}\wrapParens{1 - \nlatf}.
\end{align*}

Now consider the term with \ogram{\vect{h}} in
\corollaryref{psnr_minus_ssropt_moments}. Define the four terms as
\begin{align*}
	\trace{ \ogram{\vect{h}} \pvvar[0] } 
	&= b_1 + b_2 + b_3 - b_4.
\end{align*}
Again, we will casually drop the \Elim and \Dupp as needed.

From \theoremref{psnr_and_ssropt_moments}, 
$$
\vect{h} = - \AkronA{\wrapParens{\tr{\basev[1]}\minv{\pvsm}}} \Dupp.
$$
Continuing in order, we have:
{\footnotesize
\begin{align*}
	b_1 &= \wrapParens{\kurty-1} \trace{ \ogram{\wrapParens{\AkronA{\wrapParens{\tr{\basev[1]}\minv{\pvsm}}} }}
	\ogram{\fvec{\apvsig}}},\\
	&= \wrapParens{\kurty-1} \trace{
		\wrapParens{\tr{\basev[1]}\minv{\pvsm}\apvsig\minv{\pvsm}\basev[1]}
		\wrapParens{\tr{\basev[1]}\minv{\pvsm}\apvsig\minv{\pvsm}\basev[1]} },\\
	&= \wrapParens{\kurty-1} \trace{
		\asrowvec{0 -\tr{\pvmu}}\twobyone{1+\psnrsqopt}{-\pportwopt} 
		\asrowvec{0 -\tr{\pvmu}}\twobyone{1+\psnrsqopt}{-\pportwopt} } =
	\wrapParens{\kurty-1} \psnropt^4.\\
	b_2 
	&= \wrapParens{\kurty-1} \trace{ 2\Simm \wrapParens{\AkronA{\apvsig}}
	\wrapParens{\AkronA{\wrapParens{\ogram{\twobyone{1+\psnrsqopt}{-\pportwopt}}}}} },\\
	&= \wrapParens{\kurty-1} \trace{ 2\Simm 
	\wrapParens{\AkronA{\wrapParens{{\twobyone{0}{-\pvmu}}\tr{\twobyone{1+\psnrsqopt}{-\pportwopt}} }}} },\\
	&= \wrapParens{\kurty-1} \psnropt^4  +
	\wrapParens{\kurty-1}
	\trace{\twobyone{0}{-\pvmu}\tr{\twobyone{1+\psnrsqopt}{-\pportwopt}}
	\twobyone{0}{-\pvmu}\tr{\twobyone{1+\psnrsqopt}{-\pportwopt}} },\\
	&= 2 \wrapParens{\kurty-1} \psnropt^4.\\
	b_3 
	&= \trace{ 2\Simm \wrapParens{\AkronA{\pvsm}}
	\wrapParens{\AkronA{\wrapParens{\ogram{\twobyone{1+\psnrsqopt}{-\pportwopt}}}}} },\\
	&= \trace{ 2\Simm 
	\wrapParens{\AkronA{\wrapParens{\basev[1]\tr{\twobyone{1+\psnrsqopt}{-\pportwopt}} }}} },\\
	&= \wrapParens{1 + \psnrsqopt}^2 +
	\trace{\basev[1]\tr{\twobyone{1+\psnrsqopt}{-\pportwopt}}
	\basev[1]\tr{\twobyone{1+\psnrsqopt}{-\pportwopt}} },\\
	&= 2 \wrapParens{1 + \psnrsqopt}^2.\\
	b_4 
	&= \trace{ 2\Simm \wrapParens{\AkronA{\ogram{\apvmu}}} 
	\wrapParens{\AkronA{\wrapParens{\ogram{\twobyone{1+\psnrsqopt}{-\pportwopt}}}}} },\\
	&= \trace{ 2\Simm 
	\wrapParens{\AkronA{\wrapParens{\apvmu \tr{\twobyone{1+\psnrsqopt}{-\pportwopt}} }}} },\\
	&= 1 + 
	\trace{\apvmu\tr{\twobyone{1+\psnrsqopt}{-\pportwopt}}
	\apvmu\tr{\twobyone{1+\psnrsqopt}{-\pportwopt}} } = 2.
\end{align*}}

Collecting terms, 
\begin{align*}
	\trace{ \ogram{\vect{h}} \pvvar[0] } 
	&= 3 \wrapParens{\kurty-1} \psnropt^4 + 2 \wrapParens{1 + \psnrsqopt}^2 - 2 
	= 3 \wrapParens{\kurty-1} \psnropt^4 + 4 \psnrsqopt + 2 \psnropt^4,\\
	&= \wrapParens{3\kurty-1} \psnropt^4 + 4 \psnrsqopt.
\end{align*}

\end{proof}

The proof of \theoremref{mp_snr_ci_elliptical} is rather tedious;
in \secref{confirming_theorem}, we confirm that it is correct.

\begin{proof}[Proof of \lemmaref{gram_moments}]
We first note that 
\begin{align*}
\gram{\amreti}_{i,j} &= \sum_t \trbasev[t]\amreti\basev[i] \trbasev[t]\amreti\basev[j],\\
	&= \sum_t \trbasev[j]\kron\trbasev[t]\kron\trbasev[i]\kron\trbasev[t] \fvec{\fvec{\amreti}\tr{\fvec{\amreti}}}.
\end{align*}
We then note that
$$
\Eof{\fvec{\amreti}\tr{\fvec{\amreti}}} = \fvec{\pmmu}\tr{\fvec{\pmmu}} + \pmsigr\kron\pmsigl.
$$
Then 
\begin{align*}
	\Eof{\gram{\amreti}_{i,j}} 
	&= \sum_t \trbasev[j]\kron\trbasev[t]\kron\trbasev[i]\kron\trbasev[t] \wrapBracks{
		\fvec{\fvec{\pmmu}\tr{\fvec{\pmmu}}} + \fvec{\pmsigr\kron\pmsigl}
		},\\
		&= \wrapParens{\gram{\pmmu}}_{i,j} + 
	\sum_t \fvec{\trbasev[i]\kron\trbasev[t] \wrapParens{\pmsigr\kron\pmsigl} \basev[j]\kron\basev[t]},\\
		&= \wrapParens{\gram{\pmmu}}_{i,j} + 
	\sum_t \pmsigr_{i,j} \pmsigl_{t,t}.
\end{align*}
This establishes the mean.

Now for the variance. By \eqnref{theta_covar_elliptical} with $\kurty=1$ (and with $\ssiz=1$ since we are
	vectorizing the \amreti), letting $\vect{y} = \fvec{\fvec{\amreti}\tr{\fvec{\amreti}}}$,
{\small
\begin{align*}
\VAR{\vect{y}} &= 
	\wrapParens{\eye + \Komm}\wrapBracks{\AkronA{\wrapParens{\fvec{\pmmu}\tr{\fvec{\pmmu}} + \pmsigr\kron\pmsigl}}}\\
	&\phantom{=}\,- 
	\wrapParens{\eye + \Komm}\wrapBracks{\AkronA{\ogram{\fvec{\pmmu}}}},\\
	&=2\Simm \wrapBracks{\AkronA{\pmsigr\kron\pmsigl} 
	+ \wrapParens{\fvec{\pmmu}\tr{\fvec{\pmmu}}}\kron\pmsigr\kron\pmsigl
	+ \pmsigr\kron\pmsigl\kron\wrapParens{\fvec{\pmmu}\tr{\fvec{\pmmu}}}}.
\end{align*}%
}%
Let $\vect{z} = \fvec{\gram{\amreti}}$. Then
\begin{align*}
\VAR{\vect{z}}_{ij,kl} 
	&= \sum_{t,s} \trbasev[j]\kron\trbasev[t]\kron\trbasev[i]\kron\trbasev[t] \VAR{\vect{y}}
	\basev[l]\kron\basev[s]\kron\basev[k]\kron\basev[s],\\
	&= \sum_{t,s} \pmsigr[j,l]\pmsigl[t,s]\pmsigr[i,k]\pmsigl[t,s]
	+ \pmmu[t,j] \pmmu[s,l]\pmsigr[i,k]\pmsigl[t,s]
	+ \pmmu[t,i] \pmmu[s,k]\pmsigr[j,l]\pmsigl[t,s]\\
	&\phantom{=}\,
	+ \sum_{t,s} \pmsigr[i,l]\pmsigl[t,s]\pmsigr[j,k]\pmsigl[t,s]
	+ \pmmu[t,i] \pmmu[s,l]\pmsigr[j,k]\pmsigl[t,s]
	+ \pmmu[t,j] \pmmu[s,k]\pmsigr[i,l]\pmsigl[t,s]\\  
	&= \pmsigr[j,l]\pmsigr[i,k]\trace{\pmsigl\pmsigl}
	+ \wrapParens{\tr{\pmmu}\pmsigl\pmmu}_{j,l}\pmsigr[i,k]
	+ \wrapParens{\tr{\pmmu}\pmsigl\pmmu}_{i,k}\pmsigr[j,l]\\
	&\phantom{=}\,
	+ \pmsigr[i,l]\pmsigr[j,k]\trace{\pmsigl\pmsigl}   
	+ \wrapParens{\tr{\pmmu}\pmsigl\pmmu}_{i,l}\pmsigr[j,k]
	+ \wrapParens{\tr{\pmmu}\pmsigl\pmmu}_{j,k}\pmsigr[i,l].
\end{align*}
\end{proof}


\section{Confirming \theoremref{mp_snr_ci_elliptical}}
\label{sec:confirming_theorem}

Here we `confirm' \theoremref{mp_snr_ci_elliptical}, by which we mean we 
generate random population values of \pvmu and \pvsig, then compute 
\pvvar[0] from \theoremref{theta_covar_elliptical}, and use it with
to generate the estimate values of
$\E{\pSNR{\minv{\svsm} ; \pvsm, 0} - \ssropt}$ and
$\VAR{\pSNR{\minv{\svsm} ; \pvsm, 0} - \ssropt}$ from 
\corollaryref{psnr_minus_ssropt_moments}. We then compare these to the
forms given in \theoremref{mp_snr_ci_elliptical}, finding that
they are equal to machine precision. Note that this does \emph{nothing}
to confirm \theoremref{theta_covar_elliptical} nor 
\corollaryref{psnr_minus_ssropt_moments}, rather it gives some assurance
that \theoremref{mp_snr_ci_elliptical} is consistent with those two
results.

\begin{kframe}
\begin{alltt}
\hlkwd{library}\hlstd{(matrixcalc)}

\hlcom{# linear algebra utilities}
\hlcom{# quadratic form x' A x}
\hlstd{qform} \hlkwb{<-} \hlkwa{function}\hlstd{(}\hlkwc{A}\hlstd{,}\hlkwc{x}\hlstd{) \{} \hlkwd{t}\hlstd{(x)} \hlopt{%*%} \hlstd{(A} \hlopt{%*%} \hlstd{x) \}}
\hlcom{# quadratic form x A x'}
\hlstd{qoform} \hlkwb{<-} \hlkwa{function}\hlstd{(}\hlkwc{A}\hlstd{,}\hlkwc{x}\hlstd{) \{} \hlkwd{qform}\hlstd{(A,}\hlkwd{t}\hlstd{(x)) \}}
\hlcom{# outer gram: x x'}
\hlstd{ogram} \hlkwb{<-} \hlkwa{function}\hlstd{(}\hlkwc{x}\hlstd{) \{ x} \hlopt{%*%} \hlkwd{t}\hlstd{(x) \}}
\hlcom{# A kron A}
\hlstd{AkronA} \hlkwb{<-} \hlkwa{function}\hlstd{(}\hlkwc{A}\hlstd{) \{} \hlkwd{kronecker}\hlstd{(A,A,}\hlkwc{FUN}\hlstd{=}\hlstr{'*'}\hlstd{) \}}
\hlcom{# matrix trace}
\hlstd{matrace} \hlkwb{<-} \hlkwa{function}\hlstd{(}\hlkwc{A}\hlstd{) \{ matrixcalc}\hlopt{::}\hlkwd{matrix.trace}\hlstd{(A) \}}

\hlcom{# duplication, elimination, commutation, symmetrizer}
\hlcom{# commutation matrix;}
\hlcom{# is p^2 x p^2}
\hlstd{Comm} \hlkwb{<-} \hlkwa{function}\hlstd{(}\hlkwc{p}\hlstd{) \{}
  \hlstd{Ko} \hlkwb{<-} \hlkwd{diag}\hlstd{(p}\hlopt{^}\hlnum{2}\hlstd{)}
  \hlstd{dummy} \hlkwb{<-} \hlkwd{diag}\hlstd{(p)}
  \hlstd{newidx} \hlkwb{<-} \hlstd{(}\hlkwd{row}\hlstd{(dummy)} \hlopt{-} \hlnum{1}\hlstd{)} \hlopt{*} \hlkwd{ncol}\hlstd{(dummy)} \hlopt{+} \hlkwd{col}\hlstd{(dummy)}
  \hlstd{Ko[newidx,,}\hlkwc{drop}\hlstd{=}\hlnum{FALSE}\hlstd{]}
\hlstd{\}}
\hlcom{# Symmetrizing matrix, N}
\hlcom{# is p^2 x p^2}
\hlstd{Symm} \hlkwb{<-} \hlkwa{function}\hlstd{(}\hlkwc{p}\hlstd{) \{} \hlnum{0.5} \hlopt{*} \hlstd{(}\hlkwd{Comm}\hlstd{(p)} \hlopt{+} \hlkwd{diag}\hlstd{(p}\hlopt{^}\hlnum{2}\hlstd{)) \}}
\hlcom{# Duplication & Elimination matrices}
\hlstd{Dupp} \hlkwb{<-} \hlkwa{function}\hlstd{(}\hlkwc{p}\hlstd{) \{ matrixcalc}\hlopt{::}\hlkwd{duplication.matrix}\hlstd{(}\hlkwc{n}\hlstd{=p) \}}
\hlstd{Elim} \hlkwb{<-} \hlkwa{function}\hlstd{(}\hlkwc{p}\hlstd{) \{ matrixcalc}\hlopt{::}\hlkwd{elimination.matrix}\hlstd{(}\hlkwc{n}\hlstd{=p) \}}
\hlcom{# vector function}
\hlstd{fvec} \hlkwb{<-} \hlkwa{function}\hlstd{(}\hlkwc{x}\hlstd{) \{}
  \hlkwd{dim}\hlstd{(x)} \hlkwb{<-} \hlkwd{c}\hlstd{(}\hlkwd{length}\hlstd{(x),}\hlnum{1}\hlstd{)}
  \hlstd{x}
\hlstd{\}}

\hlcom{# compute Theta from mu, Sigma}
\hlstd{make_Theta} \hlkwb{<-} \hlkwa{function}\hlstd{(}\hlkwc{mu}\hlstd{,}\hlkwc{Sigma}\hlstd{) \{}
  \hlkwd{stopifnot}\hlstd{(}\hlkwd{nrow}\hlstd{(Sigma)} \hlopt{==} \hlkwd{ncol}\hlstd{(Sigma),}
            \hlkwd{nrow}\hlstd{(Sigma)} \hlopt{==} \hlkwd{length}\hlstd{(mu))}
  \hlstd{mu_twid} \hlkwb{<-} \hlkwd{c}\hlstd{(}\hlnum{1}\hlstd{,mu)}
  \hlstd{Sg_twid} \hlkwb{<-} \hlkwd{cbind}\hlstd{(}\hlnum{0}\hlstd{,}\hlkwd{rbind}\hlstd{(}\hlnum{0}\hlstd{,Sigma))}
  \hlstd{Theta}   \hlkwb{<-} \hlstd{Sg_twid} \hlopt{+} \hlkwd{ogram}\hlstd{(mu_twid)}
  \hlstd{Theta_i} \hlkwb{<-} \hlkwd{solve}\hlstd{(Theta)}
  \hlstd{zeta_sq} \hlkwb{<-} \hlstd{Theta_i[}\hlnum{1}\hlstd{,}\hlnum{1}\hlstd{]} \hlopt{-} \hlnum{1}
  \hlkwd{list}\hlstd{(}\hlkwc{pp1}\hlstd{=}\hlkwd{nrow}\hlstd{(Theta),}
       \hlkwc{p}\hlstd{=}\hlkwd{nrow}\hlstd{(Theta)}\hlopt{-}\hlnum{1}\hlstd{,}
       \hlkwc{mu_twid}\hlstd{=mu_twid,}
       \hlkwc{Sg_twid}\hlstd{=Sg_twid,}
       \hlkwc{Theta}\hlstd{=Theta,}
       \hlkwc{Theta_i}\hlstd{=Theta_i,}
       \hlkwc{zeta_sq}\hlstd{=zeta_sq,}
       \hlkwc{zeta}\hlstd{=}\hlkwd{sqrt}\hlstd{(zeta_sq))}
\hlstd{\}}

\hlcom{# construct four parts of Omega_0 from the Theorem;}
\hlstd{Omega_bits} \hlkwb{<-} \hlkwa{function}\hlstd{(}\hlkwc{mu}\hlstd{,}\hlkwc{Sigma}\hlstd{,}\hlkwc{kurtf}\hlstd{=}\hlnum{1}\hlstd{) \{}
  \hlstd{tvals} \hlkwb{<-} \hlkwd{make_Theta}\hlstd{(mu,Sigma)}

  \hlstd{Nmat} \hlkwb{<-} \hlkwd{Symm}\hlstd{(tvals}\hlopt{$}\hlstd{pp1)}
  \hlstd{P1} \hlkwb{<-} \hlstd{(kurtf}\hlopt{-}\hlnum{1}\hlstd{)} \hlopt{*} \hlkwd{ogram}\hlstd{(}\hlkwd{fvec}\hlstd{(tvals}\hlopt{$}\hlstd{Sg_twid))}
  \hlstd{P2} \hlkwb{<-} \hlstd{(kurtf}\hlopt{-}\hlnum{1}\hlstd{)} \hlopt{*} \hlnum{2} \hlopt{*} \hlstd{Nmat} \hlopt{%*%} \hlkwd{AkronA}\hlstd{(tvals}\hlopt{$}\hlstd{Sg_twid)}
  \hlstd{P3} \hlkwb{<-} \hlnum{2} \hlopt{*} \hlstd{Nmat}  \hlopt{%*%} \hlkwd{AkronA}\hlstd{(tvals}\hlopt{$}\hlstd{Theta)}
  \hlstd{P4} \hlkwb{<-} \hlnum{2} \hlopt{*} \hlstd{Nmat}  \hlopt{%*%} \hlkwd{AkronA}\hlstd{(}\hlkwd{ogram}\hlstd{(tvals}\hlopt{$}\hlstd{mu_twid))}
  \hlkwd{list}\hlstd{(}\hlkwc{P1}\hlstd{=P1,}\hlkwc{P2}\hlstd{=P2,}\hlkwc{P3}\hlstd{=P3,}\hlkwc{P4}\hlstd{=P4)}
\hlstd{\}}
\hlstd{Omega_0} \hlkwb{<-} \hlkwa{function}\hlstd{(}\hlkwc{mu}\hlstd{,}\hlkwc{Sigma}\hlstd{,}\hlkwc{kurtf}\hlstd{=}\hlnum{1}\hlstd{) \{}
  \hlstd{obits} \hlkwb{<-} \hlkwd{Omega_bits}\hlstd{(}\hlkwc{mu}\hlstd{=mu,}\hlkwc{Sigma}\hlstd{=Sigma,}\hlkwc{kurtf}\hlstd{=kurtf)}
  \hlstd{obits}\hlopt{$}\hlstd{P1} \hlopt{+} \hlstd{obits}\hlopt{$}\hlstd{P2} \hlopt{+} \hlstd{obits}\hlopt{$}\hlstd{P3} \hlopt{-} \hlstd{obits}\hlopt{$}\hlstd{P4}
\hlstd{\}}
\hlcom{# construct matrices F and H and vector h }
\hlstd{FHh_values} \hlkwb{<-} \hlkwa{function}\hlstd{(}\hlkwc{mu}\hlstd{,}\hlkwc{Sigma}\hlstd{,}\hlkwc{R}\hlstd{=}\hlnum{1}\hlstd{) \{}
  \hlstd{tvals} \hlkwb{<-} \hlkwd{make_Theta}\hlstd{(mu,Sigma)}
  \hlstd{zeta_sq} \hlkwb{<-} \hlstd{tvals}\hlopt{$}\hlstd{zeta_sq}
  \hlstd{zeta} \hlkwb{<-} \hlkwd{sqrt}\hlstd{(zeta_sq)}
  \hlstd{mp} \hlkwb{<-} \hlkwd{t}\hlstd{(}\hlkwd{t}\hlstd{(}\hlopt{-}\hlstd{tvals}\hlopt{$}\hlstd{Theta_i[}\hlnum{2}\hlopt{:}\hlstd{tvals}\hlopt{$}\hlstd{pp1,}\hlnum{1}\hlstd{]))}

  \hlstd{Fmat} \hlkwb{<-} \hlstd{(}\hlnum{1} \hlopt{/} \hlstd{R}\hlopt{^}\hlnum{2}\hlstd{)} \hlopt{*} \hlstd{((}\hlkwd{ogram}\hlstd{(mu)} \hlopt{/} \hlstd{zeta)} \hlopt{-} \hlstd{(zeta} \hlopt{*} \hlstd{Sigma))}
  \hlstd{H1} \hlkwb{<-} \hlopt{-} \hlkwd{cbind}\hlstd{(}\hlkwd{cbind}\hlstd{((}\hlnum{1} \hlopt{/} \hlstd{(}\hlnum{2} \hlopt{*}\hlstd{zeta_sq))} \hlopt{*} \hlstd{mp,}
                      \hlstd{(R} \hlopt{/} \hlstd{zeta)} \hlopt{*} \hlkwd{diag}\hlstd{(tvals}\hlopt{$}\hlstd{p)),}
                \hlkwd{matrix}\hlstd{(}\hlnum{0}\hlstd{,}\hlkwc{nrow}\hlstd{=tvals}\hlopt{$}\hlstd{p,}\hlkwc{ncol}\hlstd{=tvals}\hlopt{$}\hlstd{p} \hlopt{*} \hlstd{tvals}\hlopt{$}\hlstd{pp1))}
  \hlstd{H2} \hlkwb{<-} \hlopt{-} \hlkwd{AkronA}\hlstd{(tvals}\hlopt{$}\hlstd{Theta_i)}
  \hlstd{Hmat} \hlkwb{<-} \hlstd{H1} \hlopt{%*%} \hlstd{H2} \hlopt{%*%} \hlkwd{Dupp}\hlstd{(tvals}\hlopt{$}\hlstd{pp1)}

  \hlstd{hvec} \hlkwb{<-} \hlopt{-} \hlkwd{AkronA}\hlstd{(tvals}\hlopt{$}\hlstd{Theta_i[}\hlnum{1}\hlstd{,,}\hlkwc{drop}\hlstd{=}\hlnum{FALSE}\hlstd{])} \hlopt{%*%} \hlkwd{Dupp}\hlstd{(tvals}\hlopt{$}\hlstd{pp1)}
  \hlkwd{list}\hlstd{(}\hlkwc{Fmat}\hlstd{=Fmat,} \hlkwc{Hmat}\hlstd{=Hmat,} \hlkwc{hvec}\hlstd{=hvec)}
\hlstd{\}}

\hlcom{# compute the expected bias and variance}
\hlcom{# in two ways, one directly from the identity of}
\hlcom{# H, F, h and Omega_0}
\hlcom{# the other from the theorem}
\hlcom{# }
\hlcom{# then compute relative errors of the theorem's approximation.}
\hlstd{testit} \hlkwb{<-} \hlkwa{function}\hlstd{(}\hlkwc{mu}\hlstd{,}\hlkwc{Sigma}\hlstd{,}\hlkwc{kurtf}\hlstd{=}\hlnum{1}\hlstd{,}\hlkwc{R}\hlstd{=}\hlnum{1}\hlstd{) \{}
  \hlstd{p} \hlkwb{<-} \hlkwd{length}\hlstd{(mu)}
  \hlstd{tvals} \hlkwb{<-} \hlkwd{make_Theta}\hlstd{(mu,Sigma)}

  \hlstd{OmegMat} \hlkwb{<-} \hlkwd{Omega_0}\hlstd{(mu,Sigma,}\hlkwc{kurtf}\hlstd{=kurtf)}
  \hlstd{Omega} \hlkwb{<-} \hlkwd{qoform}\hlstd{(}\hlkwc{A}\hlstd{=OmegMat,}\hlkwd{Elim}\hlstd{(p}\hlopt{+}\hlnum{1}\hlstd{))}
  \hlstd{FHh} \hlkwb{<-} \hlkwd{FHh_values}\hlstd{(mu,Sigma,}\hlkwc{R}\hlstd{=R)}

  \hlcom{# now test corollary 'true' values}
  \hlstd{HtFH} \hlkwb{<-} \hlkwd{qform}\hlstd{(}\hlkwc{A}\hlstd{=FHh}\hlopt{$}\hlstd{Fmat,}\hlkwc{x}\hlstd{=FHh}\hlopt{$}\hlstd{Hmat)}
  \hlstd{Vari} \hlkwb{<-} \hlkwd{as.numeric}\hlstd{(}\hlkwd{qoform}\hlstd{(Omega,FHh}\hlopt{$}\hlstd{hvec)} \hlopt{/} \hlstd{(}\hlnum{4} \hlopt{*} \hlstd{tvals}\hlopt{$}\hlstd{zeta_sq))}
  \hlstd{Eb1} \hlkwb{<-} \hlkwd{as.numeric}\hlstd{((}\hlnum{1}\hlopt{/}\hlnum{2}\hlstd{)} \hlopt{*} \hlkwd{matrace}\hlstd{((HtFH)} \hlopt{%*%} \hlstd{Omega))}
  \hlstd{Eb2} \hlkwb{<-} \hlstd{Vari} \hlopt{/} \hlstd{(}\hlnum{2} \hlopt{*} \hlstd{tvals}\hlopt{$}\hlstd{zeta)}
  \hlstd{Ebias} \hlkwb{<-} \hlstd{Eb1} \hlopt{+} \hlstd{Eb2}

  \hlcom{# Theorem says:}
  \hlstd{T_Vari}  \hlkwb{<-} \hlstd{(((}\hlnum{3} \hlopt{*} \hlstd{kurtf} \hlopt{-} \hlnum{1}\hlstd{)} \hlopt{/} \hlnum{4}\hlstd{)} \hlopt{*} \hlstd{tvals}\hlopt{$}\hlstd{zeta_sq} \hlopt{+} \hlnum{1}\hlstd{)}
  \hlstd{T_Eb1} \hlkwb{<-} \hlstd{((kurtf} \hlopt{*} \hlstd{tvals}\hlopt{$}\hlstd{zeta_sq} \hlopt{+} \hlnum{1}\hlstd{)} \hlopt{*}
            \hlstd{(}\hlnum{1} \hlopt{-} \hlstd{tvals}\hlopt{$}\hlstd{p))} \hlopt{/} \hlstd{(}\hlnum{2} \hlopt{*} \hlstd{tvals}\hlopt{$}\hlstd{zeta)}
  \hlstd{T_Ebias} \hlkwb{<-} \hlstd{((kurtf} \hlopt{*} \hlstd{tvals}\hlopt{$}\hlstd{zeta_sq} \hlopt{+} \hlnum{1}\hlstd{)} \hlopt{*}
              \hlstd{(}\hlnum{1} \hlopt{-} \hlstd{tvals}\hlopt{$}\hlstd{p)} \hlopt{+} \hlstd{T_Vari)} \hlopt{/} \hlstd{(}\hlnum{2} \hlopt{*} \hlstd{tvals}\hlopt{$}\hlstd{zeta)}

  \hlstd{ERR_Ebias} \hlkwb{<-} \hlstd{T_Ebias} \hlopt{-} \hlstd{Ebias}
  \hlstd{ERR_Vari} \hlkwb{<-} \hlstd{T_Vari} \hlopt{-} \hlstd{Vari}
  \hlkwd{max}\hlstd{(}\hlkwd{c}\hlstd{(}\hlkwd{max}\hlstd{(}\hlkwd{abs}\hlstd{(ERR_Ebias))} \hlopt{/} \hlstd{(}\hlkwd{abs}\hlstd{(Ebias)),}
        \hlkwd{max}\hlstd{(}\hlkwd{abs}\hlstd{(ERR_Vari))} \hlopt{/} \hlstd{(}\hlkwd{abs}\hlstd{(Vari))))}
\hlstd{\}}

\hlcom{# create a population}
\hlstd{myp} \hlkwb{<-} \hlnum{5}
\hlkwd{set.seed}\hlstd{(}\hlnum{1234}\hlstd{)}
\hlstd{mu} \hlkwb{<-} \hlkwd{matrix}\hlstd{(}\hlkwd{rnorm}\hlstd{(myp,}\hlnum{1}\hlstd{),}\hlkwc{ncol}\hlstd{=}\hlnum{1}\hlstd{)}
\hlstd{ZZ} \hlkwb{<-} \hlkwd{matrix}\hlstd{(}\hlkwd{rnorm}\hlstd{(}\hlnum{100}\hlstd{,myp),}\hlkwc{ncol}\hlstd{=myp)}
\hlstd{Sigma} \hlkwb{<-} \hlkwd{t}\hlstd{(ZZ)} \hlopt{%*%} \hlstd{ZZ}

\hlcom{# test it:}
\hlkwd{testit}\hlstd{(mu,Sigma,}\hlkwc{kurtf}\hlstd{=}\hlnum{2}\hlstd{,}\hlkwc{R}\hlstd{=}\hlnum{1}\hlstd{)}
\end{alltt}
\end{kframe}[1] 2.4e-14

\end{document}